\documentclass[a4paper,11pt,twoside]{article}

\usepackage{xcolor}
\usepackage{fontenc}
\usepackage{amsthm,amsmath,amsfonts,amssymb}
\RequirePackage[colorlinks,citecolor=purple]{hyperref}
\RequirePackage{epsfig, graphicx, color}
\RequirePackage{latexsym}
\usepackage{siunitx}
\usepackage{chapterbib}

\usepackage{fullpage}
\usepackage{epsf,epsfig}
\usepackage{caption}
\usepackage{subcaption}
\usepackage{natbib}
\usepackage{booktabs}
\usepackage{graphics,graphicx}
\usepackage{enumitem}
\usepackage{float}
\usepackage[ruled]{algorithm2e}
\usepackage{bbm}
\usepackage{color}
\usepackage{dsfont}

\theoremstyle{plain}
\newtheorem{theorem}{Theorem}

\newtheorem{lemma}[theorem]{Lemma}

\newtheorem{corollary}[theorem]{Corollary}
\newtheorem{prop}[theorem]{Proposition}

\theoremstyle{definition}

\newtheorem{example}{Example}
\newtheorem{assumption}{Assumption}

\newtheorem*{assumptionext}{Assumption A3 Extended}

\theoremstyle{remark}

\newcommand{\R}{\mathbb{R}}
\newcommand{\PP}{\mathbb{P}}
\newcommand{\E}{\mathbb{E}}

\newcommand{\iid}{\stackrel{\mathrm{i.i.d.}}{\sim}}

\newcommand{\norm}[1]{\left\lVert #1 \right\rVert}

\newcommand{\ind}{\mathbbm{1}}
\newcommand{\pr}{\mathbb{P}}

\newcommand{\Cov}{\mathrm{Cov}}
\newcommand{\Corr}{\mathrm{Corr}}
\newcommand{\Var}{\mathrm{Var}}
\newcommand{\tr}{\mathrm{tr}}

\newcommand{\diag}{\textrm{diag}}

\newcommand{\argmin}{\textrm{argmin}}
\newcommand{\arginf}{\textrm{arginf}}

\newcommand{\op}{\mathrm{op}}
\newcommand{\F}{\mathrm{F}}

\newcommand{\GEE}{\mathrm{GEE}}

\newcommand{\ML}{\mathrm{ML}}

\newcommand{\SL}{\mathrm{SL}}

\newcommand{\mbb}{\boldsymbol}

\newcommand{\vertiii}[1]{{\left\vert\kern-0.25ex\left\vert\kern-0.25ex\left\vert #1 
    \right\vert\kern-0.25ex\right\vert\kern-0.25ex\right\vert}}
    
\newcommand\independent{\protect\mathpalette{\protect\independenT}{\perp}}
    \def\independenT#1#2{\mathrel{\rlap{$#1#2$}\mkern2mu{#1#2}}}
\newcommand\given{\,|\,}

\newcommand{\RN}[1]{\textup{\uppercase\expandafter{\romannumeral#1}}}

\makeatletter
\newcommand{\mylabel}[2]{#2\def\@currentlabel{#2}\label{#1}}
\makeatother

\title{Sandwich Boosting for Accurate Estimation in Partially Linear Models for Grouped Data}

\date{\today}
\usepackage{authblk}
\author[1]{Elliot H.\ Young}
\author[1]{Rajen D.\ Shah}
\affil{Statistical Laboratory, University of Cambridge, UK}

\begin{document}
\maketitle
\begin{cbunit}
\begin{abstract}
	We study partially linear models in settings where observations are arranged in independent groups but may exhibit within-group dependence. Existing approaches estimate linear model parameters through weighted least squares, with optimal weights (given by the inverse covariance of the response, conditional on the covariates) typically estimated by maximising a (restricted) likelihood from random effects modelling or by using generalised estimating equations. We introduce a new `sandwich loss' whose population minimiser coincides with the weights  of these approaches when the parametric forms for the conditional covariance are well-specified, but can yield arbitrarily large improvements in linear parameter estimation accuracy when they are not. Under relatively mild conditions, our estimated coefficients are asymptotically Gaussian and enjoy minimal variance among estimators with weights restricted to a given class of functions, when user-chosen regression methods are used to estimate nuisance functions. We further expand the class of functional forms for the weights that may be fitted beyond parametric models by leveraging the flexibility of modern machine learning methods within a new gradient boosting scheme for minimising the sandwich loss. We demonstrate the effectiveness of both the sandwich loss and what we call  `sandwich boosting' in a variety of settings with simulated and real-world data.
\end{abstract}

\section{Introduction}\label{sec:intro}
Grouped data are commonplace in many scientific, econometric and sociological disciplines. Prime examples include: repeated measures data (e.g.\ multiple readings of patient data), longitudinal data (e.g.\ where
weekly sales are recorded across multiple stores), and hierarchical data (e.g.\ educational datasets clustered by school and potentially further sub-clustered by classroom).

To fix ideas, consider a regression setting where we have available grouped data $(Y_i,D_i,X_i)\in\R^{n_i}\times\R^{n_i}\times\R^{n_i\times d}$ for $i=1,\ldots,I$, with $Y_i$ a vector of responses and predictors $(D_i, X_i)$ separated out into a covariate $D_i$, whose contribution to the response we are particularly interested in, and remaining covariates $X_i$; for instance $D_i$ may be a treatment whose effect we wish to assess after controlling for additional covariates $X_i$. In total therefore we have $N:=\sum_{i=1}^I n_i$ observations, though typically not all independent.
A simple but popular approach to modelling this data in practice is via a linear model of the form
\begin{equation}\label{eq:linear-model}
    Y_i = \beta D_i + X_i\gamma + \varepsilon_i.
\end{equation}
Here $\varepsilon_i\in\R^{n_i}$ is a vector of errors such that $\E\left[\varepsilon_i \given D_i,X_i\right]=0$ and $(\beta,\gamma)\in\R \times \R^d$ are regression coefficients to be estimated, with $\beta$ our primary  target of  inference.

A challenge in such settings is properly accounting for potential correlations between components of $\varepsilon_i$ in order to obtain accurate estimates of the parameters. This may be achieved through a weighted least squares regression yielding estimates
\begin{equation}\label{eq:betahatsimple}
     \begin{pmatrix}\hat{\beta}\\\hat{\gamma}\end{pmatrix} :=
     M^{-1} \begin{pmatrix}\sum_iD_i^{\top}\hat{W}_iY_i\\\sum_iX_i^{\top}\hat{W}_iY_i\end{pmatrix}, \qquad \text{where} \qquad M:=      \begin{pmatrix}\sum_iD_i^{\top}\hat{W}_iD_i & \sum_iD_i^{\top}\hat{W}_iX_i \\ \sum_iX_i^{\top}\hat{W}_iD_i & \sum_iX_i^{\top}\hat{W}_iX_i\end{pmatrix},
 \end{equation}
in terms of weight matrices $\hat{W}_i \in \R^{n_i \times n_i}$ to be chosen.
The optimal choice $\hat{W}_i \propto \Cov[Y_i \given D_i,X_i]^{-1}$ results in semiparametric efficient estimation of $(\beta, \gamma)$. A variety of approaches have been proposed for constructing the $\hat{W}_i$. Among the most popular are multilevel models (also known as random or mixed effects models) \citep{membook, fahrmeir}, which additionally make distributional assumptions on the errors $\varepsilon_i$, typically that of Gaussianity, and implicitly specify a particular parametrisation of $\Cov[Y_i \given D_i,X_i]$ in terms of the covariates through the introduction of latent random coefficients. Parameters are typically estimated through (restricted) maximum likelihood estimation \citep{hartley, corbeil, membook}. An alternative is the marginal models framework \citep{heagerty, diggle, fahrmeir}, which directly models the conditional covariance through a parametric form often estimated via generalised estimating equations \citep{liangzeger, geehardin, ziegler}. Provided the forms of the conditional covariance are well-specified, any of these approaches will result in efficient estimates for $\beta$ and $\gamma$. 

It is however well known that all models are wrong \citep{box}, and it is of interest to understand, under misspecification, which approaches remain useful. Below we discuss the consequences of the two potential sources of misspecification, that of the conditional covariance $\Cov[Y_i \given D_i,X_i]$, and the conditional mean $\E[Y_i \given D_i,X_i]$.

\subsection{Conditional covariance misspecification} \label{sec:cond_cov}
Misspecification of the conditional covariance has been given a good deal of attention in the literature.
The generalised estimating equation approach that has come to be known as GEE1 \citep{gee1} explicitly recognises the possibility of misspecification, and instead specifies what is referred to as a
working model for the conditional covariance, with which to construct the weights.
Valid inference is guaranteed even with arbitrary (fixed) weights as the estimator \eqref{eq:betahatsimple} is unbiased and standard errors may be based on a sandwich estimate of the variance of $\hat{\beta}$ \citep{huber, gourieroux2, royall, liangzeger},
\begin{equation} \label{eq:sandwich}
    \left\{ M^{-1}
        \begin{pmatrix}
		\sum_i D_i^{\top} \hat{W}_i \hat{R}_i\hat{R}_i^{\top} \hat{W}_i D_i
        &
        \sum_i D_i^{\top} \hat{W}_i \hat{R}_i\hat{R}_i^{\top} \hat{W}_i X_i
        \\
        \sum_i X_i^{\top} \hat{W}_i \hat{R}_i\hat{R}_i^{\top} \hat{W}_i D_i
        &
        \sum_i X_i^{\top} \hat{W}_i \hat{R}_i\hat{R}_i^{\top} \hat{W}_i X_i
	  \end{pmatrix}
    M^{-1} \right\}_{1,1}
    , \, \hat{R}_i := Y_i - \hat{\beta}D_i - X_i \hat{\gamma} \in \R^{n_i}.
\end{equation}
The idea behind the working covariance model however is to approximate the ground truth sufficiently well such that the resulting $\hat{\beta}$ has reasonably low variance; various estimation methods have been proposed for this purpose \citep{prentice, crowder, lumley, geepack}.

While this intuition is basically well-founded, perhaps surprisingly, for a given model for the covariance, the success of these approaches depends crucially on the method of estimation, as we now demonstrate with two simple examples; specific details for these are given in Appendix~\ref{appex:cond-var-misspec}.

\begin{example}[Conditional correlation misspecification] \label{ex:cond-corr-misspec}
Consider a version of \eqref{eq:linear-model} with $X_i$ omitted for simplicity, $n_i \equiv n$ and each error vector  $\varepsilon_i \in\R^n$ given by the first $n$ realisations of an $\text{ARMA}(2,1)$ model. Suppose the weights are estimated in the parametric class consisting of inverses of the covariance matrices of $\text{AR}(1)$ processes, indexed by a single autoregressive parameter $\rho$; note that the scale of the weights does not affect the resulting $\hat{\beta}$ so we do not consider this as a parameter. We consider two specific settings within this general setup. Figure~\ref{fig:working-ARMA21} plots the objective functions that are minimised for two well-established methods for constructing weights based on estimates $\hat{\rho}$ of $\rho$.
	
The so-called quasi-pseudo-Gaussian maximum likelihood approach (ML) \citep{gourieroux, mccullagh, ziegler} treats the errors as if they were normally distributed with correlation matrix given by the $\text{AR}(1)$ process for a given $\rho$, and proceeds to maximise jointly over the unknown $\beta,\rho$ and variance $\sigma^2$, what would be the likelihood were this model to hold.
\begin{figure}
	\centering
	  \includegraphics[width=\textwidth]{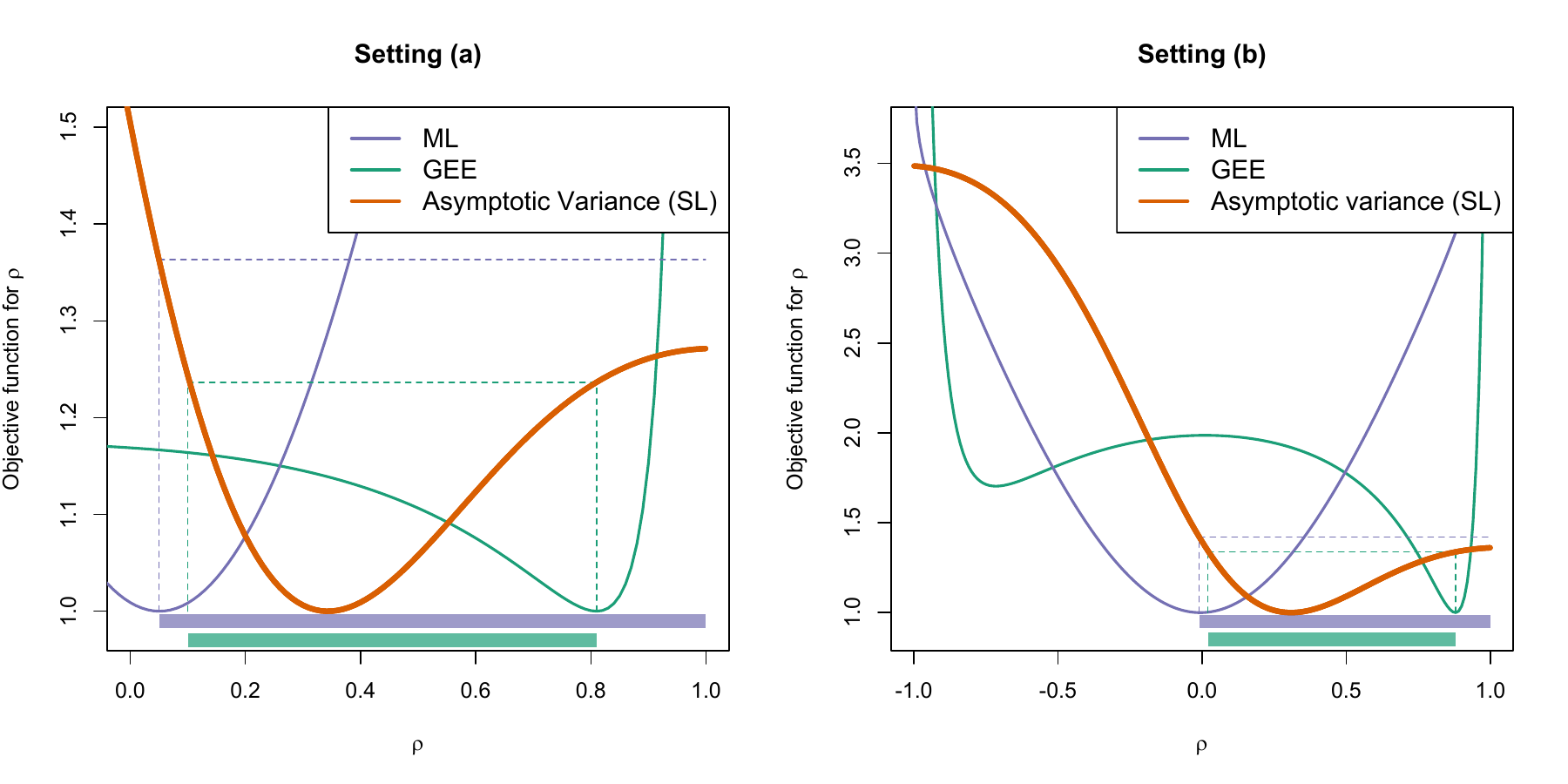}
	  \caption{The (scaled) objective functions of the parameter $\rho$ in the $\text{AR}(1)$ working model for quasi-pseudo-Gaussian maximum likelihood (ML), GEE1 (GEE) and the corresponding  asymptotic MSE relative to the minimum asymptotic MSE, in a settings with a ground truth given by an $\text{ARMA}(2,1)$ model of Example~\ref{ex:cond-corr-misspec}; coloured intervals at the bottom indicate ranges of $\rho$ where the (global) minimisers corresponding to ML and GEE would be outperformed in terms of MSE.\label{fig:working-ARMA21}}
\end{figure}

Motivated by the moment equation $\E\left[\varepsilon_{ij}\varepsilon_{ik}\right]=\sigma^2\rho^{|j-k|}$, a second approach (GEE) falling within the GEE1 framework estimates $\rho$ by the minimiser of
\begin{equation} \label{eq:GEErho}
\sum_{i=1}^I \sum_{j,k = 1}^{n} \big(\hat{\varepsilon}_{ij}\hat{\varepsilon}_{ik} - \hat{\sigma}^2\rho^{|j-k|}\big)^2, \qquad \text{with} \qquad \hat{\sigma} := \underset{\sigma > 0}{\argmin}\sum_{i=1}^I\sum_{j=1}^{n} (\hat{\varepsilon}_{ij}^2 - \sigma^2)^2 ;
\end{equation}
here the $\hat{\varepsilon}_i$ are the residuals from an initial unweighted least squares regression of $Y_i$ on $D_i$.

We also plot in orange the asymptotic variance (equivalent to the mean squared error), i.e. the population equivalent of \eqref{eq:sandwich}, of the $\beta$-estimator weighted by the inverse of an $\text{AR}(1)$ working correlation matrix for a given value of $\rho$ (the nomenclature `SL' in the legend is explained in Section~\ref{sec:contrib}). In Setting (a), we see that optimising either of these objectives can result in suboptimal weights in terms of the resulting mean squared error, and any choice of $\rho\in[0.1,0.8]$ would result in improved estimation of $\beta$. Setting (b) tells a similar story, but also illustrates issues that can arise due to local minima of the objective functions, which, in particular are typically not guaranteed to be convex. Attempting to optimise the GEE objective by initialising at $\rho=0$ results in gradient descent converging to the highly suboptimal local optimum on the left as the derivative of the objective at $\rho=0$ is slightly positive (see also Table~\ref{tab:misspec-corr} in Appendix~\ref{appex:cond-corr-misspec}, which in particular demonstrates this issue also persists for Newton optimisation schemes). The resulting asymptotic variance of this final $\hat{\beta}$ is substantially worse than even that of the unweighted choice corresponding to $\rho=0$.
\end{example}

\begin{example}[Conditional variance misspecification] \label{ex:cond-var-misspec}
    Consider an instance of model \eqref{eq:linear-model} with $n_i \equiv 1$, so the data are ungrouped, and $d=1$. Suppose the distribution of the data is such that the true conditional variance $\Var(\varepsilon_i \given D_i, X_i=x) = \Var(\varepsilon_i \given X_i=x) =: \sigma_0^2(x)$ with 
$\sigma_0(x)=2+\tanh(\lambda(x-\mu))$; we shall consider different settings for the pair of parameters $(\lambda,\mu)$. Consider using a misspecified class of functions of the form
\begin{equation} \label{eq:sigma_star}
 \sigma(x;\eta) = 1 + 2 \ind_{[\eta, \infty)}(x),
 \end{equation}
 where $\eta \in \R$, so the smooth curve of the $\tanh$ function is to be approximated by a step function.
 
 \begin{figure}
 	\centering
 	\includegraphics[width=\textwidth]{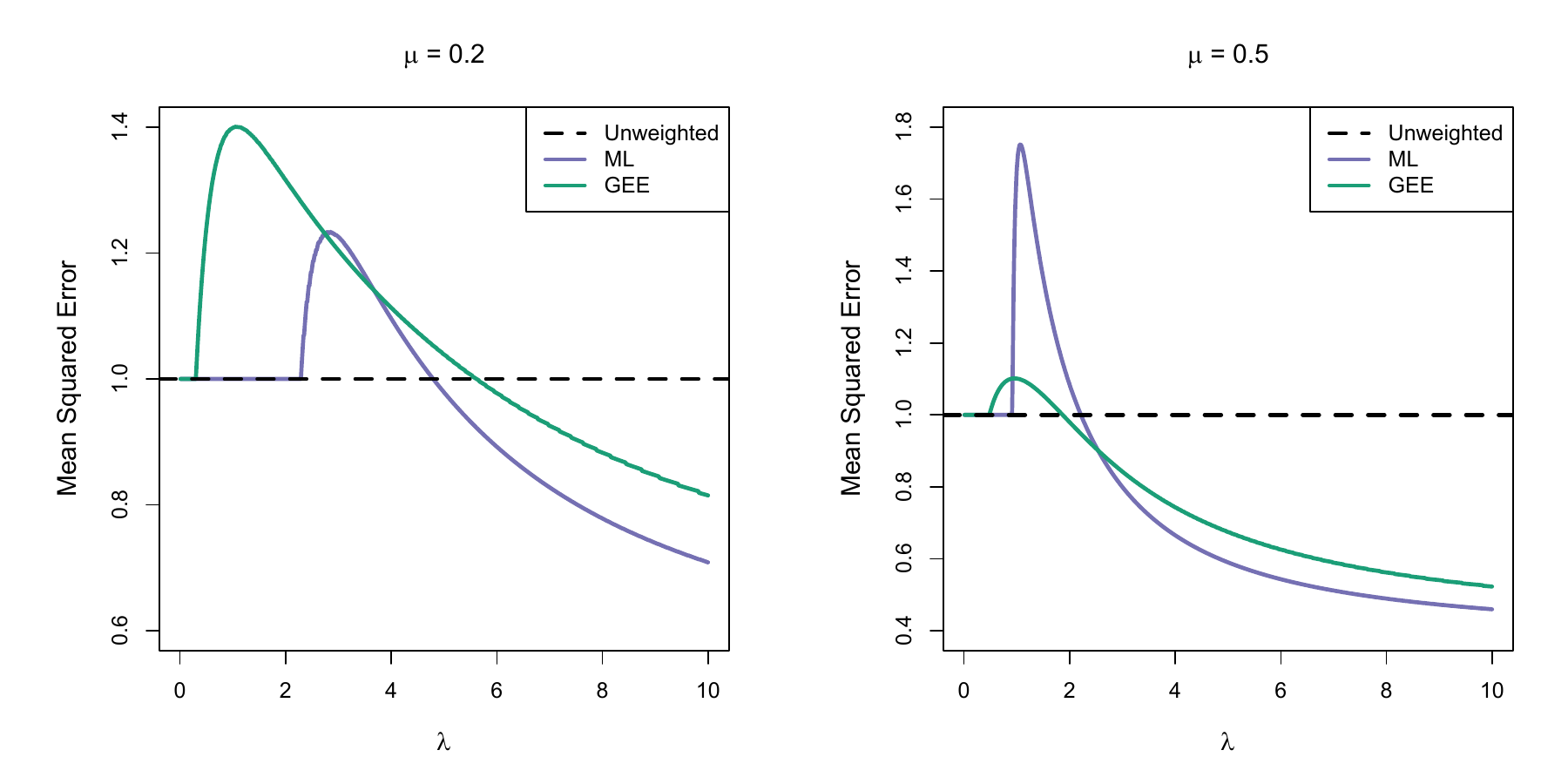}
 	\caption{Asymptotic  MSEs of \(\beta\)-estimators relative to an unweighted estimator, using weights in a misspecified model class \eqref{eq:sigma_star} estimated by each of ML and GEE in the settings of Example~\ref{ex:cond-var-misspec} parametrised by $(\lambda, \mu)$.\label{fig:MLE_vs_AV}}
 \end{figure}

Figure~\ref{fig:MLE_vs_AV} plots the relative asymptotic mean squared errors of the weighted least squares estimates of $\beta$, given by the first component of \eqref{eq:betahatsimple}, for quasi-pseudo-Gaussian maximum likelihood (ML) and GEE1-based approaches (GEE) for constructing weights $\hat{W}_i = 1 / \{\sigma(X_i;\hat{\eta})\}^2$ based on estimates $\hat{\eta}$ of $\eta$. The former involves treating the errors as if they were normally distributed with standard deviation $\sigma(x;\eta)$ for some $\eta$ while the latter estimates $\eta$ by the minimiser of the sum of the squared differences between the squared residuals from an initial least squares regression of $Y_i$ on $(D_i, X_i)$, and $(\sigma(X_i;\eta))^2$; see also \citet{robinson,carroll,tsiatis,you} for examples of estimation of the conditional variance through a similar least squares approach for improving estimation of $\beta$ in (partially) linear models.
 
We compare these strategies to a naive unweighted estimator, that is \eqref{eq:betahatsimple} with $\hat{W}_i$ constant, which makes no attempt to take advantage of the heteroscedasticity in the errors to improve estimation. Note that such weights are permitted in the model class \eqref{eq:sigma_star} used by ML and GEE in this example by taking $\eta$ large, so \eqref{eq:sigma_star} for some $\eta$ necessarily gives a better approximation to the ground truth compared to the unweighted approach.
 
A first interesting observation is the quite different behaviour of the ML and GEE approaches here, with none appearing to be uniformly preferable to the other across all parameter settings. Perhaps more surprising however is the fact that for certain values of $(\mu, \lambda)$, the performances of these more sophisticated approaches lead to an inflation of the variance over an unweighted estimator (of up to almost 80\%). This worrying behaviour can obfuscate model selection via AIC or BIC, as even at the population level they can favour models that result in poorer estimation of the parameter of interest $\beta$. 

To resolve this apparent paradox, notice that although the model classes used by ML and GEE here are richer, the optimal `projections' of the nuisance function $\sigma_0$ (and corresponding optimal weight function $W_0:=\sigma_0^{-2}$) relating to their respective losses do not necessarily coincide with the optimal projection in the sense of the mean squared error of $\hat{\beta}$; in fact in general there is no reason for them to do so, as illustrated schematically in Figure~\ref{fig:schematic} (see also Proposition~\ref{prop:modelmissspec} and Theorem~\ref{thm:divergingratio} for a formalisation of this phenomenon). We return to these issues in Section~\ref{sec:contrib}, but first turn our attention to misspecification of the conditional mean $\E[Y_i \given D_i, X_i]$.
\begin{figure}
   \centering
   	\includegraphics[width=0.6\textwidth]{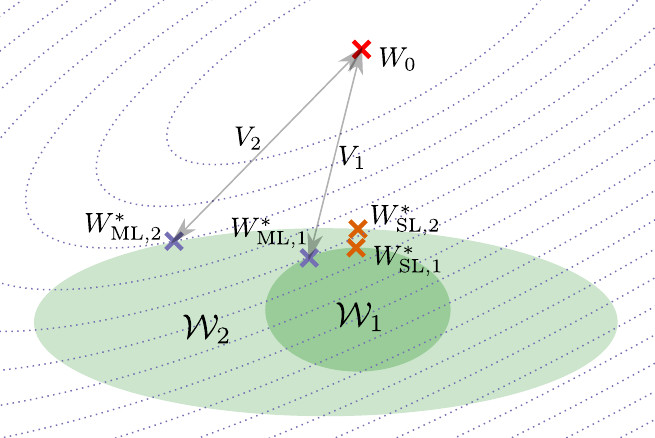}
   	\caption{Schematic of how the ML `projections' $W^*_{\text{ML},1}$ and $W^*_{\text{ML},2}$ of the optimal weight function $W_0 \propto \Cov(Y \given D, X)^{-1}$ onto the nested weight classes $\mathcal{W}_1 \subset \mathcal{W}_2$ may be such that $W^*_{\text{ML},1}$ in fact improves on $W^*_{\text{ML},2}$ when the contours of the negative likelihood loss (dotted blue ellipses) do not align with those of the asymptotic MSE; the contours of the latter are here to be thought of as circles centred at $W_0$ (not shown), so that the distances, $V_1<V_2$, are the variances of $\hat{\beta}$ constructed using weights $W^*_{\text{ML},1}$ and $W^*_{\text{ML},2}$ respectively. The projections $W^*_{\text{SL},1}$ and $W^*_{\text{SL},2}$ corresponding to the loss given by the asymptotic MSE itself (`sandwich loss') do not suffer from any such paradoxes.\label{fig:schematic}}
\end{figure}
\end{example}

\subsection{Conditional mean misspecification} \label{sec:cond_mean}
When the conditional expectation of the response given $(D_i, X_i)$ is not necessarily expected to be linear, a popular model to consider is the partially linear model
\begin{equation}\label{eq:PLR}
 	\begin{split}
 		Y_i&=\beta D_i+g_0(X_i)+\varepsilon_i,\\
 		D_i&=m_0(X_i)+\xi_i.
 	\end{split}
\end{equation}
Here $(Y_i,D_i,X_i,\varepsilon_i)$ are as in \eqref{eq:linear-model}, so in particular $\E[\varepsilon_i \given D_i, X_i] = 0$; $g_0$ and $m_0$ are potentially nonlinear row-wise functions, that is e.g.\ $g_0:\R^d \to \R$, and writing $X_{ij}$ for the $j$th row of matrix $X_i$, with a slight abuse of notation $g_0:\R^{n_i\times d}\to\R^{n_i}$ is then defined via $(g_0(X_i) )_j:=g_0(X_{ij})$;
and error $\xi_i\in\R^{n_i}$ in the $D_i$ on $X_i$ regression satisfies $\E\left[\xi_i\given X_i\right]=0$. Note that the model entails the conditional mean independence assumptions that 
\begin{align*}
    \E(Y_{ij} \given D_i, X_i) = \E(Y_{ij} \given D_{ij}, X_{ij}) \qquad \text{and} \qquad \E(D_{ij} \given X_i) = \E(D_{ij} \given X_{ij}).
\end{align*}
Nevertheless, the model is flexible enough to well-approximate a wide variety of data generating processes, yet still permits easy interpretation of the contribution of $D_i$ to the response. 
The second equation serves to model confounding due to $X_i$. In the ungrouped setting, i.e.\ where $n_i \equiv 1$, estimating $m_0$ in addition to $g_0$ forms a key part of the double / debiased machine learning (DML) framework \citep{chern} for inference about $\beta$, which in recent years has emerged as the dominant approach for estimation in partially linear models. The popularity of this paradigm is due to the fact that it accommodates the use of arbitrary machine learning methods for estimating the nuisance functions $m_0$ and $g_0$, and requires only a relatively slow rate of $1/N$ for the product of the corresponding mean squared errors in order to yield estimates of $\beta$ that converge at the parametric $1/\sqrt{N}$ rate.

\citet{emmenegger} recently extended this approach to the grouped data setting with $n_i > 1$, assuming a parametric form of the covariance $\Cov(\varepsilon_i \given D_i, X_i)$ governed by a random effects model. To estimate $\beta$, they considered regressing each of $Y_i$ and $D_i$ on $X_i$ using some independent auxiliary data and, with the resulting estimated regression functions, formed corresponding residuals $\hat{R}^Y_i$ and $\hat{R}^D_i$ to give
\begin{equation} \label{eq:DML}
	\hat{\beta} = \left(\sum_i\hat{R}_i^{D^{\top}}\hat{W}_i\hat{R}_i^D\right)^{-1}\left(\sum_i\hat{R}_i^{D^{\top}}\hat{W}_i\hat{R}_i^Y\right).
\end{equation}
Here the weight matrices $\hat{W}_i$ are formed as the inverse conditional covariances estimated using (restricted) maximum likelihood. In practice, sample-splitting and cross-fitting are used in place of auxiliary data, permitting semiparametric efficient estimates, provided the model is well-specified.

However, the flexibility in modelling the conditional mean afforded by DML comes with implications for potential misspecification of the conditional covariance. One key requirement of the approach above is that, in addition to assuming a parametric form for the conditional covariance, it should also not depend on $D_i$, i.e.\ we must have $\Cov(\varepsilon_i \given D_i, X_i) =  \Cov(\varepsilon_i \given X_i)$. This restriction is in fact a fundamental limitation of approaches based on DML. It comes as a consequence of requiring Neyman orthogonality, a certain first-order insensitivity to plugging in potentially biased machine learning estimators. Writing $l_0(X_i) := \E(Y_i \given X_i)$ and $\hat{l}$ for the corresponding regression estimate, in our case, this entails
\begin{equation} \label{eq:neyman}
	\sqrt{I}\, \E\left[\left(\hat{R}^Y_i-\beta\hat{R}^D_i\right)^{\top}\hat{W}_i\xi_i\right]= \sqrt{I} \, \E\left[\left(l_0(X_i)-\hat{l}(X_i)\right)^{\top} \E\left[\hat{W}_i\xi_i \,\big|\, X_i\right]\right]
\end{equation}
being approximately zero,
which may not hold unless $\hat{W}_i$ is a function of $X_i$ alone.
Thus misspecification of the conditional covariance, and the worrying consequences this may bring, deserve even greater attention when modelling the conditional mean in a flexible fashion through a partially linear model.

\subsection{An overview of our contributions}\label{sec:contrib}
To address the difficulties resulting from (inevitable) misspecification of the conditional covariance, we introduce a new approach for determining weight matrices in weighted least squares estimators \eqref{eq:betahatsimple} or the DML estimator in a partially linear model setting \eqref{eq:DML}. Our proposal is to  minimise a sandwich estimate of the variance of $\hat{\beta}$ (i.e.\ \eqref{eq:sandwich} or the equivalent for the DML estimate) over a given parametrisation of weight matrices, thereby directly targeting the primary objective of interest: the estimation performance of $\hat{\beta}$. We thus treat this sandwich estimate of the variance as a loss function---the `sandwich loss' (SL)---by which we determine the weights.

The asymptotic variance plotted in Figure~\ref{fig:working-ARMA21} is precisely the population version of this sandwich loss, and minimising this will, by the very definition of the loss, deliver an estimator of $\beta$ of minimal asymptotic variance among those considered. Returning to Example~\ref{ex:cond-var-misspec}, Table~\ref{tab:MLE_vs_AV} demonstrates that although ML and GEE are to be preferred in terms of estimating the true variance function $\sigma_0^2$, they are worse when estimating $\beta$ compared to our choice of weights tailored specifically for this purpose.
\begin{table}
	\begin{center}
		\begin{tabular}{ccc}
			\toprule
			Objective function & Asymptotic MSE of $\hat{\beta}$ & Asymptotic integrated MSE of $\hat{\sigma}$\\
			\midrule
			GEE &   4.47   &    {\bf 0.63}   \\
			ML &   6.82   &    0.87   \\
			Sandwich loss &   {\bf 4.10}   &    1.10 \\
			\bottomrule
		\end{tabular}                                                                      
		\caption{Quality of the estimated weights from GEE, ML and minimising the sandwich loss in Example~\ref{ex:cond-var-misspec} in terms of the derived $\hat{\beta}$ and estimate $\hat{\sigma}$ of $\sigma_0$ given by the square root of the inverse of the weights.}\label{tab:MLE_vs_AV}
	\end{center}
\end{table}
While in these simple examples, the performances of the different methods are noticeably, but not radically different, we show later in Theroem~\ref{thm:divergingratio} that there exist data generating distributions for which for a large class of misspecified working covariance models, the ratio of the variance for the $\beta$-estimators between either GEE or ML and our optimal weighting scheme, can be arbitrarily large.
 
One key message of our paper therefore is that particularly when there is a high risk of misspecification of the conditional covariance, the sandwich loss may be preferable over existing criteria when inference on the (partially) linear model parameter $\beta$ is of primary interest. In fact, even in the case that the conditional covariance is well-specified, there is a danger that ML or GEE approaches could converge to weights that are only locally optimal for their respective losses. Since it is only the global optima of  these losses that correspond to weights with favourable properties for estimation of $\beta$, there is no guarantee that the weights obtained are even locally optimal in terms of the resulting asymptotic variances. In this sense GEE and ML approaches may be more vulnerable to the consequences of local optima in their objective functions, compared to the sandwich loss; see Section~\ref{sec:sim-complexity}.
 
A second main aim of our work is to introduce a new modelling strategy for working conditional covariances that can harness the power and flexibility of machine learning methods, similarly to how DML uses machine learning to accurately estimate nuisance functions. A challenge however is that standard regression methods cannot be directly deployed to construct weight matrices. In order to make use of these, we first decompose the inverse of the weights into a working conditional variance of each entry of $\varepsilon_i$, and a working correlation that we model parametrically. We introduce a new gradient boosting approach for estimating these two components through minimising our sandwich loss; this takes as input a user-chosen regression method that is used within the boosting procedure to estimate the conditional variances. We demonstrate the favourable performance of our resulting `sandwich boosting' method in a variety of numerical experiments.

In Section~\ref{sec:sandwich_loss} we introduce the sandwich loss and compare its population version to ML and GEE-based equivalents. In Section~\ref{sec:sample} we verify that despite the unusual form of the sandwich loss, under relatively mild conditions, we can expect a minimiser of the sample version to converge to its population counterpart (Theorem~\ref{thm:parametric-consistency}). 
We introduce our general cross-fitted weighted estimation approach in Section~\ref{sec:cross} before describing our proposed sandwich boosting scheme in Section~\ref{sec:boosting}. Section~\ref{sec:theory} presents theory showing that our resulting estimator for the partially linear model coefficient $\beta$ is asymptotically Gaussian under relatively mild conditions on the predictive ability of nuisance function estimators, permitting the construction of honest confidence intervals for $\beta$. In contrast to existing results in this context, our theory permits the group sizes $n_i$ to grow with the number of groups $I$, and importantly accommodates misspecification of the conditional covariance. We present the results of a variety of numerical experiments on simulated and real-world data in Section~\ref{sec:numerical_results} that further explore the themes hinted at in Examples~\ref{ex:cond-corr-misspec} and~\ref{ex:cond-var-misspec}, and demonstrate the effectiveness of our sandwich boosting approach. We conclude with a discussion in Section~\ref{sec:discussion} outlining avenues for further work, including a sketch of an extension to estimating a coefficient function in a version of \eqref{eq:PLR} with $\beta$ replaced by a function $\beta(X_i)$ that is a linear combination of known basis functions, using a generalisation of the sandwich loss. The supplementary material contains the proofs of all results presented in the main text, additional theoretical results, further details on the examples and numerical experiments, and a detailed computational analysis of the sandwich boosting methodology. Sandwich boosting is implemented in the R package \texttt{sandwich.boost}\footnote{\url{https://github.com/elliot-young/sandwich.boost/}}. 
Below we briefly review some related work not necessarily covered elsewhere in the introduction, and collect together some notation used throughout the paper.

\subsection{Other related literature} \label{sec:related}
As indicated  in the previous sections, our work connects to a vast literature on mixed effects models, also known as multilevel models or hierarchical models, and generalised estimating equations. Some recent developments in this area have looked at such models in high-dimensional contexts. In particular \citet{li2022inference} consider using a particular proxy conditional covariance parametrised by a single parameter for computational simplicity. \citet{li2018doubly} considers a flexible conditional covariance specification through selecting from high-dimensional random effects via regularising terms in the Cholesky decomposition of the covariance matrix of the random effects.

Most closely related to our setup here however is the work of \citet{emmenegger} who consider partially linear mixed effect models \citep{zeger1994semiparametric} in the double machine learning (DML) framework, popularised by \citet{chern}; see also \citet{kennedy2022semiparametric} for a recent review of this broad topic. Earlier work considered specific nonparametric estimators for $g_0$ in the (grouped) partially linear model framework, for example \citet{huang2007efficient} use regression splines to estimate $g_0$ and a GEE approach for estimating weights.
Within the DML area, work related to the setting of the (ungrouped) partially linear model includes \citet{vansteelandt2022assumption} who propose new targets of inference and DML estimation strategies in potentially misspecified generalised partial linear models, and \citet{emmenegger2021regularizing} who look at estimation in partially linear models with unobserved confounding in an instrumental variables setting using a DML approach and additional regularisation to reduce variance.

Boosting \citep{schapire1, freund1996experiments}, on which our sandwich boosting proposal is built, has received a lot of interest in recent years due to its success on modern datasets of interests. 
A long line of work (see for example \citet{breiman2, mason, statisticalboosting}, and \citet{buhlmann1} for a review) in machine learning has resulted in the functional gradient descent perspective of boosting, which we make use of in developing our sandwich boosting proposal.

In general terms, our use of the sandwich loss involves selecting among estimators (in our case determined by weight functions) based on estimates of their quality (in our case, their MSEs). In this sense it is related to a number of statistical approaches, including, for example cross-validation. Of particular note is the recent work of \citet{park2021more} who look at average treatment effect estimation in multilevel studies and pick from among a family of estimators based on augmented inverse propensity weighting \citep{RRZ94, RR95}, one minimising an estimate of their variance.

\subsection{Notation}\label{sec:notation}
We denote by $\mathbb{S}^m$ the set of symmetric $m \times m$ matrices, and for $M \in \mathbb{S}^m$, we write \(\Lambda_{\max}(M)\) and \(\Lambda_{\min}(M)\) for its maximum and minimum eigenvalues respectively. We write $\mathbb{S}^m_{++} \subset \mathbb{S}^m$ for the set of positive definite matrices. Let \(\Phi\) denote the cumulative distribution function of a standard Gaussian distribution. We will also use the shorthand \([I]:=\left\{1,\ldots,I\right\}\). For the uniform convergence results we will present, it will be helpful to write, for a law $P$ governing the distribution of a random vector \(U\in\R^d\), $\E_P U$ for its expectation and $\pr_P(U \in B)=: \E \ind_B (U)$ for any measurable $B \subseteq \R^d$. Further, given a sequence of families of probability distributions \(\left(\mathcal{P}_I\right)_{I\in\mathbb{N}}\), and for a sequence of families of real-valued random variables \(\left(A_{P,I}\right)_{P\in\mathcal{P}_I,I\in\mathbb{N}}\) (which we note are permitted to depend on $P \in \mathcal{P}_I$), we write \(A_{P,I}=o_{\mathcal{P}}(1)\) if \(\lim_{I\to\infty}\sup_{P\in\mathcal{P}_I}\PP_{P}\left(\left|A_{P,I}\right|>\epsilon\right)=0\) for all \(\epsilon>0\), $A_{P,I}=o_{\mathcal{P}}(g(I))$ for a given function $g:(0, \infty) \to (0, \infty)$ if $g(I)^{-1}A_{P,I}=o_{\mathcal{P}}(1)$, and \(A_{P,I}=O_{\mathcal{P}}(1)\) if for any \(\epsilon>0\) there exist \(M_{\epsilon}, I_{\epsilon}>0\) such that \(\sup_{I\geq I_{\epsilon}}\sup_{P\in\mathcal{P}_I}\PP_{P}\left(\left|A_{P,I}\right|>M_{\epsilon}\right)<\epsilon\).

\section{The sandwich loss}\label{sec:sandwich_loss}
In this section, we first outline a general weighted least squares framework for estimating $\beta$, within which we formally introduce the notion of the sandwich loss. In Section~\ref{sec:pop} we study properties of the sandwich loss at the population level, compared to the ML and GEE-based approaches described in Section~\ref{sec:cond_cov}. In Section~\ref{sec:sample} we then study the behaviour of the sample version of the sandwich loss and show that under mild conditions, a minimiser will converge in probability to the population minimiser. For the remainder of this paper, we will work in the setting of the partially linear model \eqref{eq:PLR}, which includes as a special case, the linear model \eqref{eq:linear-model}.
\subsection{Weighted estimation} \label{sec:weighted} 

Here we outline a general strategy for weighted estimation of $\beta$, with which we will introduce our proposed sandwich loss. For simplicity, we describe our approach in terms of an auxiliary dataset, independent of our main data. In practice however, we use sample splitting and cross-fitting to construct our estimator, as described in Section~\ref{sec:cross}. As in the approach of \citet{emmenegger}, we begin by regressing each of $Y$ and $D$ onto $X$ using our main data to give estimates $\hat{l}$ and $\hat{m}$ of the row-wise conditional expectations $l_0(X_i) = \E(Y_i \given X_i)$ and $m_0(X_i) = \E(D_i \given X_i)$. With these we form
respective vectors of residuals $\tilde{R}^Y$ and $\tilde{R}^D$ using which we find an initial estimate $\tilde{\beta}$ of $\beta$ using \eqref{eq:DML} with the weights $\hat{W}_i$ set to identity matrices:
\[
\tilde{\beta} = \left(\sum_{i}\tilde{R}_i^{D^{\top}}\tilde{R}_i^D\right)^{-1}\left(\sum_{i}\tilde{R}_i^{D^{\top}}\tilde{R}_i^Y\right).
\]
We then form `estimates' of the errors $\xi_i$ and $\varepsilon_i$ given by $\tilde{\xi}_i:=\tilde{R}_i^D$ and $\tilde{\varepsilon}_i:=\tilde{R}^Y_i - \tilde{\beta}\tilde{R}^D_i$ to obtain a sandwich estimate of ($N$ times) the variance of a $\beta$-estimator utilising given weight matrices $W(X_i)$:
 \begin{equation}\label{eq:sandwichloss}
 \hat{L}_{\SL}(W):=\left(\frac{1}{N} \sum_{i}\tilde{\xi}_i^{\top}W(X_i)\tilde{\xi}_i\right)^{-2}\left(\frac{1}{N} \sum_{i}\left(\tilde{\xi}_i^{\top}W(X_i)\tilde{\varepsilon}_i\right)^2\right).
 \end{equation}
In the above, $W$ is a function that given $X \in \R^{n \times d}$ for any $n  \in \mathbb{N}$, outputs a matrix $W(X) \in \R^{n \times n}$. The sandwich loss $\hat{L}_{\SL}$ views the variance estimate as a function of $W$, and specifically as a measure of the quality of the weights $W$, somewhat analogously to how the likelihood views the density of the data as a function of parameters to be determined.
Given a class $\mathcal{W}$ of functions $W$, we then propose to find an (approximate) minimiser $\hat{W}$ of $\hat{L}_{\SL}$ over $\mathcal{W}$ (see Section~\ref{sec:boosting} for our sandwich boosting approach for carrying this out).
While the sandwich loss is thus a rather trivial (re-)definition, as we try to demonstrate in the present work, this shift in perspective from variance estimate to a loss function to be minimised, can lead to non-trivial improvements in terms of estimating $\beta$. Note that the sandwich loss is unaffected by any positive scaling of $W$: this is to be expected since the resulting $\hat{\beta}$ is equally invariant.

On the auxiliary data, we then set $\hat{W}_i := \hat{W}(X_i)$, and form a weighted $\beta$-estimate of the form \eqref{eq:DML}.
Our $\hat{W}$ should then deliver a low variance estimate of $\beta$ since this is precisely the way in which it was constructed. Indeed, we show in Section~\ref{sec:theory} that the variance estimate \eqref{eq:sandwichloss} consistently estimates the true variance (times $N$).

Note that although we have introduced the sandwich loss via a specific construction of the error estimates $\tilde{\varepsilon}_i$ and $\tilde{\xi}_i$, it is applicable more broadly to other such estimates. For example, in the simpler linear model setting, applying the Woodbury matrix identity shows that the sandwich variance estimate \eqref{eq:sandwich} also takes the form given in \eqref{eq:sandwichloss}, for certain error estimates formed through weighted least squares regressions.
 
In the following section, we compare the sandwich loss to the ML and GEE-based approaches outlined in Section~\ref{sec:cond_cov} theoretically by studying population versions of these.
 
\subsection{Population level analysis} \label{sec:pop}
The examples in Section~\ref{sec:cond_cov} hint at potential issues with the ML and GEE-based approaches, which unlike the sandwich loss, are not explicitly geared towards minimising the variance of the resulting $\hat{\beta}$: whilst when the conditional covariance is well-specified, their goals coincide with those of the sandwich loss, in the case of misspecification, this need not be the case. Recall that since the weight matrices are restricted to be functions of $X_i$ alone due to the requirement of Neyman orthogonality, it is plausible that some form of misspecification is unavoidable.

To study the properties of the approaches under potential misspecification, we work for simplicity in a setting where we observe i.i.d.\ instances of the partially linear model~\eqref{eq:PLR} with fixed finite group size $n_i\equiv n$, and consider population versions of the respective losses:
\begin{align*}
    L_{\ML}(W) &:= \E\left[-\log\det W(X) + \varepsilon^{\top}W(X)\varepsilon\right], \\
    L_{\GEE}(W) &:= \E\left[\big\|W(X)^{-1}-\varepsilon\varepsilon^{\top}\big\|^2\right], \\
    L_{\SL}(W) &:= \Big(\E\left[\xi^{\top}W(X)\xi\right]\Big)^{-2}\left(\E\left[\left(\xi^{\top} W(X)\varepsilon\right)^2\right]\right).
\end{align*}
Here $\varepsilon, \xi$ and $X$ are to be understood as generic versions of their counterparts with subscript $i$ satisfying $\Lambda_{\min}\E(\varepsilon\varepsilon^{\top}\given D,X) > \epsilon$ and $\Lambda_{\min}\E(\xi\xi^{\top}\given X) > \epsilon$ almost surely for some $\epsilon>0$, and $\E[\{\Lambda_{\max}\E(\xi\xi^{\top}\given X)\}^2] < \infty$. Note that the GEE loss $L_{\GEE}$ is defined here with respect to an arbitrary matrix norm $\| \cdot\|$ derived from an inner product, such as the Frobenius norm.

In practice we would minimise empirical versions of these loss functions over a restricted class $\mathcal{W}$ of weight functions $W$. The following family of such classes will facilitate our theoretical comparison of the loss functions. Let $q:\R^{n \times d} \to \R^{n \times d}$ be a measurable function representing some potential `coarsening' of its input data $x \in \R^{n \times d}$. For example, $q$ may be based on a partition of $\R^{n \times d}$ into disjoint regions, with $x \mapsto q(x)$ reducing the initial data $x \in \R^{n \times d}$ to a representative of the particular region into which $x$ falls. Given some such $q$, we consider the class of all weight functions given by a measurable function $\tilde{W}: \R^{n \times d} \to \mathbb{S}^n$ composed with $q$ such that the final outputs are invertible almost surely and obey mild integrability conditions:
\begin{equation} \label{eq:W_def}
\mathcal{W} := \Big\{ \tilde{W} \circ q \text{ for } \tilde{W} :\R^{n \times d} \to \mathbb{S}^n  : \E [\{\Lambda_{\max}(\tilde{W} \circ q(X))\}^2],\,\, \E [\{\Lambda_{\min}(\tilde{W} \circ q(X))\}^{-2} ] < \infty   \Big\}.
\end{equation}
The above setup includes as a special case, the setting where $\mathcal{W}$ includes all appropriately integrable weight matrices $W:\R^{n \times d} \to \mathbb{S}^{n}; x \mapsto W(x)$ so in particular $\mathcal{W}$ would then contain $x \mapsto \E(\varepsilon \varepsilon^{\top} \given X = x)^{-1}$. On the other hand, when the image of $q$ is finite, the resulting $\mathcal{W}$ is essentially a parametric class.

Let us write $V_{\text{ML/GEE}} := L_{\SL}(\E (\varepsilon \varepsilon^{\top} \given q(X))^{-1})$ and $V_{\SL}$ for the infimum of  $L_{\SL}$ over $\mathcal{W}$. Proposition~\ref{prop:modelmissspec} below shows that $V_{\text{ML/GEE}}$ is the asymptotic variance of both the ML and GEE losses, and gives a condition under which this coincides with $V_{\SL}$.

\begin{prop}\label{prop:modelmissspec}
\begin{enumerate}[label=(\alph*)]
\item $L_{\ML}$ and $L_{\GEE}$ are both minimised over $\mathcal{W}$ \eqref{eq:W_def} by $\E (\varepsilon \varepsilon^{\top} \given q(X))^{-1}$. 
\item The asymptotic variance $V_{\SL}$ of the sandwich loss satisfies $V_{\SL}\leq V_{\mathrm{ML/GEE}}$, with equality if and only if for some constant $c > 0$,
\begin{equation}\label{eq:GEE=MLE=AV}
    \E\left[\tr\big[\Cov(\varepsilon\given q(X))^{-1}\Cov(\varepsilon\given D,X)\big]\,\xi\xi^{\top} \,\Big|\, q(X)\right] = c\, \E\big[\xi\xi^{\top}\,\big|\, q(X)\big].
\end{equation}
\end{enumerate}
\end{prop}

Note the condition~\eqref{eq:GEE=MLE=AV} holds in the instance that $\Cov(\varepsilon \given D, X) = \Cov(\varepsilon \given q(X))$; but given variable $D$ is considered here to be important enough for its associated parameter to be the target of our inference, it is not inconceivable that the errors depend on them after conditioning on $q(X)$. Moreover, we can expect that $\Cov(\varepsilon \given X) \neq \Cov(\varepsilon \given q(X))$ unless $\mathcal{W}$ is a sufficiently rich class of functions; in the former case 
\eqref{eq:GEE=MLE=AV} may fail to hold even in the favourable case where $\Cov(\varepsilon \given D, X) = \Cov(\varepsilon \given X)$.
In settings where \eqref{eq:GEE=MLE=AV} fails, it is possible for the ratio $V_{\mathrm{ML/GEE}} / V_{\SL}$ to be arbitrarily large, as Theorem~\ref{thm:divergingratio} below shows.

\begin{theorem}\label{thm:divergingratio}
Suppose $q(X)$ is not almost surely constant. Then for all \(M\geq1\), and for all pairs \(\Sigma,\Omega:\R^{n \times d} \to \mathbb{S}^n_{++}\) of positive definite matrices that are functions of $q(X)$, there exists a law on \((\varepsilon,\xi)\) satisfying the conditions of the model~\eqref{eq:PLR} with \(\E\left[\varepsilon\varepsilon^{\top}\given q(X)\right]=\Sigma(q(X))\), \(\E\left[\xi\xi^{\top}\given q(X)\right]=\Omega(q(X))\) and
\begin{equation*}
    \frac{V_{\mathrm{ML/GEE}}}{V_{\SL}}\geq M.
\end{equation*}
\end{theorem}

While the discrepancy between $V_{\mathrm{ML/GEE}}$ and $V_{\SL}$ is not always expected to be very large, it is nevertheless potentially a cause for concern that this can happen even when $\Sigma$ and $\Omega$ are identity matrices, for example.

\subsection{Sample level considerations} \label{sec:sample}
The previous section illustrated some of the advantages of the sandwich loss at the population level. The sandwich loss $\hat{L}_{\SL}$ \eqref{eq:sandwichloss} we propose however is unusual in the sense that it is not composed of a sum of independent terms typical of the objective functions of M-estimators. A further complication is that $\hat{L}_{\SL}$ involves estimates of the errors $\xi_i$ and $\varepsilon_i$, rather than these errors themselves. The classical theory of M-estimation \citep[Chap.~5]{vandervaart} is therefore not immediately applicable here, and it is not clear whether the useful advantages of the population level sandwich loss $L_{\SL}$ transfer over to its empirical counterpart. The result below however shows that under relatively mild conditions, minimisation of $\hat{L}_{\SL}$ over a parametric class $\mathcal{W}$ \eqref{eq:W_para} yields convergence in probability to the minimiser of $L_{\SL}$. Note that since $\hat{L}_{\SL}$ is invariant to any positive rescaling of its argument, we are permitted to fix the scale of any weight function in $\mathcal{W}$; we do this by asking for the maximum eigenvalue of any output weight matrix to be 1. We continue to work in the setup of the previous section, where our data consist of $I$ i.i.d.\ groups of finite size $n$ following the partially linear model \eqref{eq:PLR}. 

\begin{theorem}\label{thm:parametric-consistency}
Let $\Psi \subset \R^d$ be some compact set and
    \begin{equation} \label{eq:W_para}
    \mathcal{W} := \Big\{ W(\psi) : \R^{n \times d} \to \R^{n \times n}  \text{ with } \sup_{x \in \R^{n \times d}}\Lambda_{\max}(W(\psi)(x)) = 1 \text{ for }\psi \in\Psi \Big\}.
    \end{equation}
    Suppose $\psi^* \in \Psi$ is such that for all $\epsilon>0$, $\inf_{\psi \in \Psi :\norm{\psi-\psi^*}\geq\epsilon}L_{\SL}(\psi) > L_{\SL}(\psi^*)$.
    Assume the regularity conditions set out in Appendix~\ref{appsec:regularitycond}, and additionally that for estimates $(\tilde{\xi}_i,\tilde{\varepsilon}_i)$ of the error terms $(\xi_i,\varepsilon_i)$ either
    \begin{enumerate}[label=(\alph*)]
        \item\label{ass:lp4}   $\E\left[\frac{1}{N}\sum_{i=1}^I\|\tilde{\xi}_i-\xi_i\|_2^4\right]\vee\E\left[\frac{1}{N}\sum_{i=1}^I\|\tilde{\varepsilon}_i-\varepsilon_i\|_2^4\right]=o(1)$, or
        \item\label{ass:lp2} $\E\left[\frac{1}{N}\sum_{i=1}^I\|\tilde{\xi}_i-\xi_i\|_2^2\right]\vee\E\left[\frac{1}{N}\sum_{i=1}^I\|\tilde{\varepsilon}_i-\varepsilon_i\|_2^2\right]=o(N^{-\frac{1}{2}})$. 
    \end{enumerate}
        Then any sequence of approximate minimisers $\hat{\psi}_I$ with $\hat{L}_{\SL}(\hat{\psi}_I)\leq\hat{L}_{\SL}(\psi^*)+o_P(1)$ satisfies
        \[
        \hat{\psi}_I = \psi^* +  o_P(1).
        \]
\end{theorem}

\section{Methodology}\label{sec:method}

In this section we present our sandwich boosting weighted regression procedure. We first describe
the basic outline of our approach for generic weight matrices employing cross-fitting in Section~\ref{sec:cross} and then in Section~\ref{sec:boosting} present our boosting strategy to approximately optimise the sandwich loss within a flexible class of working covariance models. 

\subsection{Cross-fitting}\label{sec:cross}
In Section~\ref{sec:weighted}, we outlined a simplified approach for estimating $\beta$ in the partially linear model \eqref{eq:PLR} involving first obtaining `estimates' $\tilde{\xi}_i$ and $\tilde{\varepsilon}_i$ of the errors $\xi_i$ and $\varepsilon_i$ with which to determine a weight function $\hat{W}$ through minimising the sandwich loss \eqref{eq:sandwichloss}. The second stage involved forming on an independent dataset, an estimate $\hat{\beta}$ through \eqref{eq:DML}, so in particular, conditioning on the initial dataset, $\hat{W}$ would be fixed. In practice, only a single dataset would be available, and we construct the two independent datasets through sample splitting, employing a $K$-fold cross-fitting scheme to recover the loss in efficiency in using only part of the data to construct the final estimator. 

Cross-fitting is a popular approach in semiparametric problems for ensuring the independence of nuisance parameter estimates from the data on which the final estimate of the target parameter is formed. This independence means that certain empirical process terms can be controlled straightforwardly even when arbitrary nuisance parameter estimators are used. \citet{chern} and \citet{emmenegger} use such cross-fitting in the regular and grouped partially linear models respectively, where the nuisance function estimates in question are $\hat{l}$ and $\hat{m}$. In our case, cross-fitting additionally serves to guarantee independence of the weight function estimates.

Algorithm~\ref{alg:main} details our method, with observation groups indexed by $\mathcal{I}_k^c$ playing the roles of the initial datasets for $k=1,\ldots,K$, and those indexed by $\mathcal{I}_k$ involved in the construction of the final estimator $\hat{\beta}$. Note that rather than forming separate estimates of $\beta$ corresponding to each $\mathcal{I}_k$, we instead form sets of residuals $\hat{R}^D_i$ and $\hat{R}^Y_i$, finally constructing $\hat{\beta}$ using these via \eqref{eq:DML}, an approach known as DML2 \citep{chern}. A typical choice of $K$ would be $5$ or $10$ and averaging over estimators $\hat{\beta}$ constructed using different random partitions $\mathcal{I}_1,\ldots,\mathcal{I}_K$ can improve the performance in finite samples; see Section~\ref{sec:cross-fitting_K} in the supplementary material for further discussion.
As well as obtaining an estimate $\hat{\beta}$, we also calculate a sandwich estimate $\hat{V}$ of the variance, with which an approximate $(1-\alpha)$-level confidence interval $\hat{C}(\alpha)$ may be constructed.

\begin{algorithm}[ht]
 \KwIn{Data $(Y_i,D_i,X_i)_{i\in[I]}$; number of folds for cross-fitting $K$; choices of regression methods for estimating $(l_0,m_0)$;
 method for finding weight estimates;
 significance level $\alpha$.}
 Partition $\left[I\right]$ randomly into $K$ disjoint sets $\mathcal{I}_1,\ldots,\mathcal{I}_K$ approximately equal in size.\\
 \For{$k\in[K]$}{
    Using data indexed by $\mathcal{I}_k^c$, fit estimates $\hat{l}^{(k)}$ and $\hat{m}^{(k)}$ of the functions $l_0$ and $m_0$ by regressing $Y$ and $D$ against $X$ respectively. \\
   \For{$i\in\mathcal{I}_k^c$}{
        Calculate residuals $\tilde{R}_i^Y := Y_i-\hat{l}^{(k)}(X_i)$ and $\tilde{R}_i^D := D_i-\hat{m}^{(k)}(X_i)$. \\
        Calculate $\tilde{\beta}_k := \left(\sum_{i\in\mathcal{I}_k^c}\tilde{R}_i^{D^{\top}}\tilde{R}_i^D\right)^{-1}\left(\sum_{i\in\mathcal{I}_k^c}\tilde{R}_i^{D^{\top}}\tilde{R}_i^Y\right)$.\\
        Calculate $\tilde{\xi}_i := \tilde{R}_i^D$ and $\tilde{\varepsilon}_i:=\tilde{R}_i^Y-\tilde{\beta}_k\tilde{R}_i^D$. \\
   }
    Find an (approximate) minimiser $\hat{W}^{(k)}$ of
    \vspace{-4mm}
    \begin{equation*}
    	\hat{L}_{\SL}^{(k)}(W):=\bigg(\sum_{i\in\mathcal{I}_k^c}\tilde{\xi}_i^{\top}W(X_i)\tilde{\xi}_i\bigg)^{-2}\bigg(\sum_{i\in\mathcal{I}_k^c}\left(\tilde{\xi}_i^{\top}W(X_i)\tilde{\varepsilon}_i\right)^2\bigg),
     \vspace{-4mm}
    \end{equation*}
    over some class of functions $\mathcal{W}$ corresponding to a working covariance structure (e.g.\ using sandwich boosting, see Section~\ref{sec:boosting}).\\
    \For{$i\in\mathcal{I}_k$}{
            Set $\hat{R}_i^Y:=Y_i-\hat{l}^{(k)}(X_i)$ and $\hat{R}_i^D:=D_i-\hat{m}^{(k)}(X_i)$ and $\hat{W}_i := \hat{W}^{(k)}(X_i)$. \\
    		Define $\hat{\xi}_i:=\hat{R}_i^D \text{ and } \hat{\varepsilon}_i:=\hat{R}_i^Y-\tilde{\beta}_k\hat{R}_i^D$.
    }
    }
    Calculate $\hat{\beta} :=\Big(\sum_{k=1}^K\sum_{i\in\mathcal{I}_k}\hat{R}_i^{D^{\top}}\hat{W}_i\hat{R}_i^D\Big)^{-1}\Big(\sum_{k=1}^K\sum_{i\in\mathcal{I}_k}\hat{R}_i^{D^{\top}}\hat{W}_i\hat{R}_i^Y\Big)$. \\
    Calculate $\hat{V}:=N\Big(\sum_{k=1}^K\sum_{i\in\mathcal{I}_k}\hat{\xi}_i^{\top}\hat{W}_i\hat{\xi}_i\Big)^{-2}\Big(\sum_{k=1}^K\sum_{i\in\mathcal{I}_k}\big(\hat{\xi}_i^{\top}\hat{W}_i\hat{\varepsilon}_i\big)^2\Big).$ \\
    Calculate $\hat{C}(\alpha):=\Big[\hat{\beta}-N^{-\frac{1}{2}}\hat{V}^{\frac{1}{2}}\Phi^{-1}\left(1-\frac{\alpha}{2}\right),\hat{\beta}+N^{-\frac{1}{2}}\hat{V}^{\frac{1}{2}}\Phi^{-1}\left(1-\frac{\alpha}{2}\right)\Big].$ \\
    \KwOut{Estimator $\hat{\beta}$ for $\beta$, an estimator of its asymptotic variance $\hat{V}$ and $1-\alpha$ level confidence interval $\hat{C}(\alpha)$ for $\beta$.}
 \caption{Construction of $\beta$-estimator}
 \label{alg:main}
\end{algorithm}

\subsection{Sandwich boosting}\label{sec:boosting}
Algorithm~\ref{alg:main} introduced a generic approach for incorporating weight functions learnt from the data into an estimator for $\beta$ via approximately minimising the sandwich loss over some class of functions $\mathcal{W}$. We now introduce an approach for performing this approximate minimisation over a class $\mathcal{W}$ defined implicitly  through a user-chosen regression method.

To introduce our approach, it is helpful to consider a class of proxy conditional covariances parametrised as
\begin{equation}\label{eq:sigma-rho-workingcov}
    D_\sigma(\cdot)\, C_\theta(\cdot)\, D_\sigma(\cdot),
\end{equation}
where, given an input $X_i\in\R^{n_i\times d}$, the functions $D_\sigma$ and $C_\theta$ output
\begin{align*}
	D_\sigma(X_i) &:=\diag(\sigma(X_{i1}),\ldots,\sigma(X_{in_i})) \\
	C_\theta(X_i)&:=\left(\ind_{\{j=k\}}+ \ind_{\{j \neq k\}}\rho_\theta(X_{ij},X_{ik})\right)_{(j,k)\in[n_i]^2}.
\end{align*}
Here, $\sigma : \R^d \to (0,\infty)$ and $\rho_\theta : \R^{d} \times \R^d \to [0,1]$ for $\theta \in \Theta$ (where $\Theta$ is some closed convex set) are proxy conditional standard deviation and correlation functions that are to be modelled nonparametricaly and parametrically respectively. Note that the working covariances \eqref{eq:sigma-rho-workingcov}
have the property that the $jk$th entry depends only on $X_{ij}$ and $X_{ik}$; this need not be the case for $\{\Cov( Y_i \given X_i)\}_{jk}$ for example, but is nevertheless a reasonable simplification.

Redefining $s \equiv 1/\sigma$, the corresponding weight class consists of functions of the form
\begin{equation}\label{eq:workingcovariance}
    W(s,\theta)(\cdot) = D_s (\cdot) \, C^{-1}_{\theta}(\cdot)\, D_s(\cdot),
\end{equation}
understanding, $C^{-1}_\theta(X_i) := \{C_\theta(X_i)\}^{-1}$ up to an arbitrary positive scale factor. As an example, an equicorrelated working correlation may be parametrised as
\[
C^{-1}_\theta(X_i) = \left(\ind_{\{j=k\}} - \frac{\theta}{1 + \theta n_i}\right)_{(j,k)\in[n_i]^2}
\]
for $\theta \in [0,\infty)$,
with corresponding correlation $\rho_{\theta}$ given by $\theta / (1+\theta)$.
We also consider a version for nested group structures permitting two constant correlations and an autoregressive form suitable for longitudinal data where $\{C_\theta(X_i)\}_{jk} = \theta^{|j-k|}$; see Appendix~\ref{appsec:workingcorrelations}.
Such inverse working correlations are among the classes of weights considered in the GEE1 framework \citep{liangzeger, zegerliang, ziegler}. A key difference here however is the greater flexibility afforded by learning the working inverse standard deviation $s$ function through a boosting scheme, as we now explain. We also outline in Section~\ref{sec:init} how our boosting scheme may be initialised at estimated weight functions derived using existing ML or GEE-based methods, for example, thereby increasing the flexibility of the functional forms considered.

Boosting has emerged as one of the most successful learning methods, with the XGBoost implementation \citep{xgboost} in particular dominating machine learning competitions such as those hosted on Kaggle \citep{bojer2021kaggle}. Since its introduction in the work of \citet{schapire1}, it has been generalised and reinterpreted as a form of functional gradient descent of an objective function $\hat{L}$ based on the data \citep{statisticalboosting, mason, buhlmann2, buhlmann1}. For an objective function $f \mapsto \hat{L}(f) \in \R$ applied to function $f:\R^d \to \R$, an individual boosting iteration involves perturbing the $f$-function by a step in the `direction' of the $f$-score
\begin{equation}\label{eq:gateauxderiv}
    U^{(f)}(f)\left(x\right) := \partial_f\hat{L}(f)\left(x\right) := \frac{\partial}{\partial\alpha}\hat{L}\left(f+\alpha\delta_x\right)\Big|_{\alpha=0},
\end{equation}
where $\delta_x:\R^d \to \R$ is the indicator function at $x\in\R^d$. Procedurally, the $f$-score $U^{(f)}(f)$ is evaluated at the data points and the regressed onto the data using a user-chosen `base learner'.

Typically $\hat{L}$ takes the form of an empirical risk, so $\hat{L}(f) = \sum_i \ell(y_i, f(x_i))$ for some loss  function $\ell$ and predictor--response pairs $(x_i, y_i)$. The corresponding $f$-score evaluated at the data point $x_i$ then takes the simple form $\frac{\partial}{\partial \alpha} \ell(y_i, \alpha) |_{\alpha=f(x_i)}$, a function of the $i$th observation alone. This allows for $f$-score calculation in linear time, as well as the possibility of parallelising computations for large data sets as exploited by XGBoost. In our case however where we wish to minimise the sandwich loss $(s, \theta) \mapsto \hat{L}_{\SL}(W(s, \theta))$ (which recall is defined in terms of estimates $\tilde{\varepsilon}_i$ and $\tilde{\xi}_i$ of the errors) over weight functions parametrised by $(s, \theta)$ \eqref{eq:workingcovariance}, we obtain as the $s$-score
\begin{equation} \label{eq:s_score}
U_{\text{SL}}^{(s)}(s,\theta)(X_{ij})= -2 b^{-3} \left\{ \left(\tilde{\xi}_i^{\top}A_{ij}\tilde{\xi}_i\right) \sum_{i'} c_{i'}^2 - b c_i \left(\tilde{\xi}_i^{\top}A_{ij}\tilde{\varepsilon}_i\right)\right\},  
\end{equation}
where $W_i(s, \theta) := W(s, \theta)(X_i)$ and
\begin{gather*}
A_{ij}:=\left(\left\{C_\theta^{-1}(X_i)\right\}_{kj}s(X_{ik})\ind_{\{l=j\}} + \left\{C_\theta^{-1}(X_i)\right\}_{jl}s(X_{il})\ind_{\{k=j\}}\right)_{k,l\in[n_i]^2},
\\
b:= \sum_{i=1}^I\tilde{\xi}_{i}^{\top}W_{i}(s,\theta)\tilde{\xi}_{i},
\qquad
c_i := \tilde{\xi}_{i}^{\top}W_{i}(s,\theta)\tilde{\varepsilon}_{i}.
\end{gather*}
Thus the $s$-score at $X_{ij}$ is a function of all the data points. Nevertheless, it may be computed for all $X_{ij}$ at a cost of $O\big(\sum_i n_i^3\big)$.
However, as we show in Appendix~\ref{appsec:workingcorrelations},  for the equicorrelated, nested and autoregressive working correlation structures, this cost is reduced to $O(N)$ and may be parallelised similarly to the standard setting of minimising an empirical risk. The critical factors in allowing this are: (a) that computing the matrix inverse $C_\theta^{-1}(X_i)$ present in $W_i(s,\theta)$, which for an arbitrary correlation $\rho_\theta$ may take $O(n_i^3)$ time, has a simple closed  form; and (b) computation of the terms $\tilde{\xi}_i^{\top}A_{ij}\tilde{\xi}_i$ and $\tilde{\xi}_i^{\top}A_{ij}\tilde{\varepsilon}_i$ involving the sparse matrix $A_{ij}$ can be arranged to be $O(1)$ by precomputing other terms appropriately.

Along with updating $s$ by regressing the $s$-score above onto the $X_{ij}$ and taking a step in the direction of the negative of this fitted regression function, we may also perform a regular projected gradient descent update for $\theta$ using the $\theta$-score vector
\begin{equation} \label{eq:theta_score}
U_{\text{SL}}^{(\theta)}(s,\theta)=-2b^{-3} \left\{ \left(\sum_i c_i^2\right)\left(\sum_i \tilde{\xi}_i^{\top}\partial_\theta W_i(s,\theta) \tilde{\xi}_i\right) - b \sum_{i} c_i \left( \tilde{\xi}_i^{\top}\partial_\theta W_i(s,\theta) \tilde{\varepsilon}_i \right)\right\},
\end{equation}
which may be computed at no greater cost than the $s$-score above.

With these scores, our sandwich boosting algorithm is summarised in Algorithm~\ref{alg:boosting}; note $\pi_\Theta$ denotes projection onto the set $\Theta$. Recall that in our cross-fitting scheme (Algorithm~\ref{alg:main}), we envisage applying boosting to approximately minimise a version of the sandwich loss corresponding to subsets of the observation indices. As is standard in boosting, the algorithm requires a choice of initialisers (in our case $\hat{s}_1$ and $\hat{\theta}_1$) and a base learner. In all of our numerical experiments, we take $\hat{s}_1 \equiv 1$, $\hat{\theta}_1=0$ and use additive penalised cubic regression splines implemented in the R package \texttt{mgcv} \citep{gam-code}. We select the number of boosting iterations $m_{\mathrm{stop}}$ by cross-validation, as recommended by \citep{buhlmann1}, though using our sandwich loss as the evaluation criterion. Note that the algorithm is stated for fixed step sizes $\lambda^{(\theta)}$ and $\lambda^{(s)}$ for simplicity; in Appendix~\ref{appsec:numerical_results}, we describe the specific choices and variable step size schemes used in our numerical results.

\subsubsection{Initialising from other weighting schemes} \label{sec:init}
The classes of weight functions that may be fitting using our computationally efficient boosting schemes with equicorrelated, autoregressive or nested correlations can be rather rich when used in conjunction with a flexible base learner. However, these classes would not encompass all those available using classical mixed  effects modelling, for example. In order to further broaden the classes of weight functions that may be considered, one can start with an initial estimated weight or conditional covariance function $\hat{\Sigma}_{\text{init}}$ estimated through GEE or ML-based approaches, and fit a weight function of the form 
\[
\{\hat{\Sigma}_{\text{init}}(\cdot)\}^{-1/2} \, W( \cdot) \, \{\hat{\Sigma}_{\text{init}}(\cdot)\}^{-1/2},
\]
where $W(\cdot)$ is of the form given by \eqref{eq:workingcovariance}, using sandwich boosting. The boosting algorithm then serves to push the initial $\hat{\Sigma}_{\text{init}}$ in a better direction for the purposes of estimating $\beta$. This may be carried out easily by running Algorithm~\ref{alg:boosting} on transformed error estimates $\tilde{\varepsilon}_i \mapsto \{\hat{\Sigma}_{\text{init}}(X_i)\}^{-1/2}\tilde{\varepsilon}_i $ and similarly for $\tilde{\xi}_i$. In fact, one can use multiple initialisers in this way, and pick among the best sandwich-boosted versions via cross-validation with the sandwich loss as the quality criterion.

\begin{algorithm}[ht]
 \KwIn{Index set $\mathcal{I}$ for boosting training set $(X_i, \tilde{\varepsilon}_i, \tilde{\xi}_i)_{i \in \mathcal{I}}$; boosting initialisers $\hat{s}_1 : \R^d \to (0, \infty)$ and $\hat{\theta}_1 \in\R$; $\theta$-descent and $s$-descent step sizes $\lambda^{(\theta)}, \lambda^{(s)}>0$; total number of boosting iterations $m_{\mathrm{stop}}$; base learner for regressing scores onto the data.}
 
 \For{$m = 1,\ldots,m_{\mathrm{stop}}$}{
   
    Calculate the $N$ scores $U^{(s)}_{\SL}(\hat{s}_m,\hat{\theta}_m)(X_{ij})$~\eqref{eq:s_score} and $\theta$-score $U^{(\theta)}_{\SL}(\hat{s}_m,\hat{\theta}_m)$~\eqref{eq:theta_score} using Algorithm~\ref{alg:scores-arbitrary}.
    
  Regress $\left(U^{(s)}_{\SL}(\hat{s}_m,\hat{\theta}_m)(X_{ij})\right)_{i \in \mathcal{I}, j \in [n_i]}$ onto $(X_{ij})_{i \in \mathcal{I}, j \in [n_i]}$ using the base learner to obtain the fitted regression function $x \mapsto \hat{u}_m(x)$.

  $\hat{s}_{m+1} := \hat{s}_m-\lambda^{(s)}\hat{u}_m$.

  $\hat{\theta}_{m+1} := \pi_{\Theta}(\hat{\theta}_m - \lambda^{(\theta)}U_{\SL}^{(\theta)}(\hat{s}_m,\hat{\theta}_m))$.
 }
 
 \KwOut{Weight function $\hat{W} := W(\hat{s}_{m_{\mathrm{stop}}+1},\hat{\theta}_{m_\mathrm{stop}+1})$.}
 \caption{Sandwich boosting algorithm}
 \label{alg:boosting}
\end{algorithm}

\section{Theory}\label{sec:theory}

In this section we present results on asymptotic normality of the $\beta$-estimator of Algorithm~\ref{alg:main}, and coverage guarantees for the confidence interval construction therein. Recalling the setup of Section~\ref{sec:boosting}, we consider the case where the estimated weight functions are such that $\{\hat{W}^{(k)}(\cdot)\}^{-1}$ takes the form \eqref{eq:sigma-rho-workingcov}; these may be obtained using our sandwich boosting approach (Algorithm~\ref{alg:boosting}), but this is not required for our theoretical results. Let us define $\hat{\sigma}^{(k)}$ and $\hat{\theta}^{(k)}$ by
\[
 D_{\hat{\sigma}^{(k)}}(\cdot) \,C_{\hat{\theta}^{(k)}}(\cdot) \, D_{\hat{\sigma}^{(k)}}(\cdot) := \{\hat{W}^{(k)}(\cdot)\}^{-1},
\]
and write $\hat{\rho}^{(k)} := \rho_{\hat{\theta}^{(k)}}$.

For simplicity of the exposition, similarly to Sections~\ref{sec:pop}, here we consider the case where our data are i.i.d.\ copies of the group of $n$ observations $(Y,D,X) \in \R^n \times \R^n \times \R^{n\times d}$ following the partially linear model \eqref{eq:PLR}. Recall that $X$ is a matrix whose rows, denoted $X_j \in \R^d$, are not necessarily independent or identically distributed.
We also relax the i.i.d.\ assumption at the group level to permit non-identically distributed groups of unequal size in Appendix~\ref{appsec:apptheory}.

Our results here are based in part on \citet{emmenegger}, but build on them in two key ways. Firstly, we permit the conditional covariance $\Cov(Y \given D, X)$ to be misspecified, i.e., for the (likely) possibility that the probability limit of the $\hat{W}^{(k)}(X)$ is not some multiple of $\Cov(Y \given D, X)^{-1}$. Secondly, we consider asymptotic regimes that allow the group size $n$ to diverge with the total number of observations $N=nI$ at  rates we will formalise later. Throughout, we assume that the number of folds $K$ used in cross-fitting is finite.

We state our results as uniform convergence results over the sequence of classes of distributions $(\mathcal{P}_I)_{I \in \mathbb{N}}$ such that for all $I$ sufficiently large and for all $P \in \mathcal{P}_I$, the following hold. Note that in the below, $\delta$, $\mu_{\Sigma}$, $\mu_{\Omega}$, $\gamma$, $\alpha$ and $\kappa$ are to be thought of as constants, not depending on $P$. The values of these are not relevant in the case where the group size $n$ is finite, but play a role in the rate of growth permitted when it is diverging. Moreover $a \lesssim b$ denotes $a \leq c b$ for constant $c>0$ not depending on $P$. We have however suppressed  the dependence on $P$ in $l_0$, $\sigma^*$ etc. 

\begin{assumption}[Moment assumptions]\label{ass:mathcal(P)}
\hfill
\begin{enumerate}
    \item[(\mylabel{A1.1}{A1.1})] There exists $\delta > 0$ such that $\E_P\big[\norm{\varepsilon}_2^{4+\delta}\big]^{\frac{1}{4+\delta}}\lesssim\sqrt{n}$ and $\E_P\big[\norm{\xi}_2^{4+\delta}\big]^{\frac{1}{4+\delta}} \lesssim \sqrt{n}$.
    \item[(\mylabel{A1.2}{A1.2})] The covariance matrices \(\Sigma(D,X) := \E_P\left[\varepsilon\varepsilon^{\top}\,|\,D,X\right]\) and \(\Omega(X) := \E_P\left[\xi\xi^{\top}\,|\,X\right]\) satisfy \(\Lambda_{\min}(\Sigma)\gtrsim 1\) and \(\Lambda_{\min}(\Omega)\gtrsim 1\) almost surely. Further, \(\Lambda_{\max}(\Sigma)\lesssim n^{\mu_{\Sigma}}\) and \(\Lambda_{\max}(\Omega)\lesssim n^{\mu_{\Omega}}\) almost surely for some \(\mu_{\Sigma},\mu_{\Omega}\in[0,1]\).
\end{enumerate}
\end{assumption}
Lower values of $\mu_\Sigma$ and $\mu_\Omega$ will permit faster rates of divergence of $n$ (see Appendix~\ref{appsec:apptheory}).
Note that when $\Sigma$ is close to the equal correlation working covariance of Section~\ref{sec:boosting}, we can expect $\mu_{\Sigma}=1$. For our simplified result in Corollary~\ref{cor:main} we set $\mu_\Omega=1$. In the below, whenever we condition on estimated regression functions, this is to be understood as conditioning on the data used to train these.

\begin{assumption}[Accuracy of regression function estimators]\label{ass:nuisance_rates}
Define the maximum within group estimation errors of regression functions $l_0$ and $m_0$:
\begin{equation*}
    \mathcal{R}_l := \max_{j\in[n]}\E_P\left[\left(\hat{l}^{(1)}(X_{j})-l_0(X_{j})\right)^2\,\bigg|\,\hat{l}^{(1)}\right],
    \quad
    \mathcal{R}_m := \max_{j\in[n]}\E_P\left[\left(\hat{m}^{(1)}(X_{j})-m_0(X_{j})\right)^2\,\bigg|\,\hat{m}^{(1)}\right].
\end{equation*}
Then the errors of these nuisance function estimators satisfy:
\begin{enumerate}
    \item[(\mylabel{A2.1}{A2.1})] $\mathcal{R}_m\left(\mathcal{R}_l \vee \mathcal{R}_m\right) = o_{\mathcal{P}}(N^{-1})$,
    \item[(\mylabel{A2.2}{A2.2})] $\mathcal{R}_m \vee  \mathcal{R}_l= o_{\mathcal{P}}(1)$,
    \item[(\mylabel{A2.3}{A2.3})] 
    $\underset{j\in[n]}{\max}\,\E_P\left[\big|\hat{l}^{(1)}(X_{j})-l_0(X_{j})\big|^{4+\delta}\,\Big|\,\hat{l}^{(1)}\right] = O_{\mathcal{P}}(1)$ and
    \\
    $\underset{j\in[n]}{\max}\,\E_P\left[\big|\hat{m}^{(1)}(X_{j})-m_0(X_{j})\big|^{4+\delta}\,\Big|\,\hat{m}^{(1)}\right] = O_{\mathcal{P}}(1)$.
\end{enumerate}
\end{assumption}
The assumptions on the regression function estimates are relatively weak and identical to those in \citet{emmenegger}, with what is typically the strongest requirement (\ref{A2.1}) permitting nonparametric rates of $o_{\mathcal{P}}(N^{-1/2})$ for each of $\mathcal{R}_m$ and $\mathcal{R}_l$. Faster rates than this however weaken conditions on how $n$ may diverge; see Corollary~\ref{cor:main} below.

\begin{assumption}[Stability of weight function estimates] \label{ass:weight}
Suppose there also exists deterministic functions $\sigma^*:\R^d\to\R$ and $\rho^*:\R^d\times\R^d\to (0,\infty)$ whose estimators satisfy:
\begin{enumerate}
    \item[(\mylabel{A3.1}{A3.1})] $\mathcal{R}_{\sigma} := \underset{j\in[n]}{\max}\,\E_P\left[\left(c^*\hat{\sigma}^{(1)}(X_{j})-\sigma^*(X_{j})\right)^2\,\Big|\,\hat{\sigma}^{(1)}\right] = o_{\mathcal{P}}(1)$ \\ where $c^* := \underset{c>0}{\arginf}\,\underset{j\in[n]}{\max}\,\E_P\left[\left(c\hat{\sigma}^{(1)}(X_{j})-\sigma^*(X_{j})\right)^2\right]$,
    \item[(\mylabel{A3.2}{A3.2})] $\mathcal{R}_{\rho} := \underset{j\neq j'}{\max}\,\E_P\left[\left(\hat{\rho}^{(1)}(X_{j},X_{j'})-\rho^*(X_{j},X_{j'})\right)^2\,\Big|\,\hat{\rho}^{(1)}\right] = o_{\mathcal{P}}(1)$.
    \end{enumerate}
Further suppose the associated weights $W^{\ast}:=W(\sigma^*,\rho^*)(X)$ satisfy
\begin{enumerate}
    \item[(\mylabel{A3.3}{A3.3})] $\Lambda_{\max}(W^*)\lesssim 1$ and $\Lambda_{\min}(W^*)\gtrsim n^{-\gamma}$ for some $\gamma\in[0,\mu_{\Sigma}]$ almost surely,
     \item[(\mylabel{A3.4}{A3.4})] $\Lambda_{\min}\Big(W^{*\frac{1}{2}}\Sigma W^{*\frac{1}{2}}\Big)\gtrsim n^{-\kappa}$ and $\Lambda_{\max}\Big(W^{*\frac{1}{2}}\Sigma W^{*\frac{1}{2}}\Big)\lesssim n^{\alpha}$ for some $\kappa\in[0,\gamma]$ and $\alpha\in[0,\mu_{\Sigma}]$ almost surely.
\end{enumerate}
\end{assumption}
Assumptions \ref{A3.1} and \ref{A3.2} require a probabilistic limit for our estimates of the weight function, but this need not correspond closely to the inverse of $\Sigma$. The eigenvalue assumption \ref{A3.4} however does loosely quantify the discrepancy between these, and $\kappa$ and $\alpha$ impact the permitted divergence rate of $n$. The reason for introducing the $c^*$ is that the estimated weights need not be on the same scale as $\Sigma^{-1}$ (recall that the sandwich loss is invariant to positive scaling of its argument).

\begin{theorem}\label{thm:main}
Consider Algorithm~\ref{alg:main}.
Let the sequence of distribution families $(\mathcal{P}_I)_{I \in \mathbb{N}}$ for $(Y, D, X)$ be such that for all $I$ sufficiently large, and for all $P \in \mathcal{P}_I$, Assumptions~\ref{ass:mathcal(P)}, \ref{ass:nuisance_rates} and \ref{ass:weight} are satisfied. Further, suppose that the group size $n$ is either finite, or diverges at a rate satisfying Assumption~\ref{ass:n_max} in Appendix~\ref{appsec:apptheory}. Then defining
\[
V:=\left(\E_{P}\left[\xi^{\top}W^*\xi\right]\right)^{-2}\left(\E_{P}\left[\left(\varepsilon^{\top}W^*\xi\right)^2\right]\right),
\]
we have that $\hat{\beta}$ is uniformly asymptotically Gaussian
    \[
\lim_{I\to\infty}\sup_{P\in\mathcal{P}_I}\sup_{t\in\R}\left|\PP_{P}\left(\sqrt{N/V}\big(\hat{\beta}-\beta\big)\leq t\right)-\Phi\left(t\right)\right|=0,
    \]
    and moreover the above holds with $V$ replaced by $\hat{V}$.
\end{theorem}

The result shows in particular that $\hat{C}(\alpha)$ constructed in Algorithm~\ref{alg:main} is an asymptotic $(1-\alpha)$-level honest confidence interval, under the given assumptions. Corollary~\ref{cor:main} below specialises a version of Theorem~\ref{thm:main} for diverging group sizes for two cases of interest where relatively simple forms of (conservative) rate requirements on $n$ are available.

\begin{corollary} \label{cor:main}
    Adopt the setup and notation of Theorem~\ref{thm:main} but suppose $\delta \geq 4$ in \ref{A1.1} and \ref{A2.3}, and additionally that $\mathcal{R}_{\rho} = O_{\mathcal{P}}(N^{-1})$. Suppose the estimated weight functions $\hat{W}^{(k)}$ are constructed to fall within classes $\mathcal{W}$ (see Section~\ref{sec:boosting}) corresponding to one of the following two settings: 
    \begin{enumerate}[label=(\roman*)]
        \item  Equicorrelated working correlation, but where the true conditional correlation $\Corr(Y \given X, D)$ may be arbitrary;
        \item Autoregressive $\text{AR}(1)$ working correlation and when $\mu_\Sigma=0$ (see \ref{A1.2}).
    \end{enumerate}
    Then the conclusions of Theorem~\ref{thm:main} holds for diverging group sizes at the following rates:
    \begin{align*}
        \text{\bf Equicorrelated (i): }& \qquad n = o\left(N^{\frac{1}{3+2\kappa}} \wedge \left(N\,  \mathcal{R}_m \big(\mathcal{R}_l\vee\mathcal{R}_m\big) \right)^{-\frac{1}{\kappa}} \wedge \mathcal{R}_{\sigma}^{-\frac{1}{2(\kappa+2)}} \right),
        \\
        \text{\bf Autoregressive (ii): }& \qquad n = o\Big(N^{\frac{1}{3}} \wedge \mathcal{R}_l^{-1} \wedge \mathcal{R}_{\sigma}^{-1} \Big).
    \end{align*}
\end{corollary}
We discuss each of the cases (i) and (ii) in turn. Case (i) places few restrictions on the true conditional covariance and so results in a more stringent requirement on the growth rate of $n$. Recall that the middle term $N\,  \mathcal{R}_m \big(\mathcal{R}_l\vee\mathcal{R}_m)$ is required to be $o_{\mathcal{P}}(1)$ by \ref{A2.1}
 and small values of the parameter $\kappa \in [0,1]$ indicate better approximation of $\Sigma$ in \ref{A3.4}. 
 Case (ii) enforces that $\mu_\Sigma=0$: this occurs for instance when the true correlation function is upper bounded by an exponentially decaying function with separation (satisfied e.g.\ for ARMA processes). As such, this condition may be appropriate for longitudinal data. The rate requirement on $n$ is relatively weak: both  $\mathcal{R}_l$ and $\mathcal{R}_{\sigma}$ need only satisfy a rate requirement weaker than that on $\mathcal{R}_m$ entailed by \ref{A2.1} in order for $n$ to be permitted to grow at any rate $o(N^{1/3})$.

\section{Numerical experiments}\label{sec:numerical_results}

In this section we explore the empirical properties of the sandwich boosting estimator on a number of simulated and real-world datasets. In all cases where covariates $X_i$ are available in addition to our covariate of interest $D_i$, we fit partially linear models using the approach of Algorithm~\ref{alg:main}, estimating nuisance regression functions  $m_0$ and $l_0$ using cubic regression splines implemented in the \texttt{mgcv} package \citep{gam-code}; however in addition to using weight functions $\hat{W}^{(k)}$ selected using sandwich boosting, we also compare to versions with these selected using quasi-pseudo-Gaussian maximum likelihood (ML) and GEE1 (GEE) based methods; see Section~\ref{sec:loss_details} of the Appendix for further details. Note that the use of Algorithm~\ref{alg:main} with ML is essentially the approach of \citet{emmenegger}, as in the \texttt{dmlalg} package~\citep{dmlalg}, but using a robust sandwich estimate of  the variance.

Section~\ref{sec:simulations} explores four simulated settings with varying degrees of misspecification of the conditional covariance. Section~\ref{sec:real-data} looks at two datasets: the first, on orange juice sales grouped by store, highlights the benefits of flexible variance modelling via the nonparametric $s\equiv 1/\sigma$ component in our sandwich boosting scheme (Section~\ref{sec:boosting}); and the second, a longitudinal study on women's wages, provides a real-life example of the phenomena seen in Examples~\ref{ex:cond-corr-misspec} and \ref{ex:cond-var-misspec} where more complex conditional covariance models can lead to poorer $\beta$-estimation when weights are selected using ML or GEE-based approaches, and where the minimiser of the sandwich loss is rather different to those corresponding to the ML and GEE objectives.

\subsection{Simulated data}\label{sec:simulations}

We look at four simulated scenarios, and in each case consider different weight classes $\mathcal{W}$: we describe these classes below in terms of their implied working covariances i.e.\ in terms of the inverses of the weight matrices. For all the approaches we use the cross-fitting scheme (Algorithm~\ref{alg:main}) with $K=2$ folds.
\begin{itemize}
    \item \textbf{Homoscedastic: } Depending on the setting, this consists of either equicorrelated or autoregressive $\text{AR}(1)$ working correlations scaled by a constant variance, with the single parameter estimated either by maximising a Gaussian likelihood (ML), an approach of the form given in \eqref{eq:GEErho} (GEE) or minimising the sandwich loss through projected gradient descent (only used in the setting of Section~\ref{sec:sim-corr}).
    \item \textbf{Heteroscedastic: } This uses the same working correlations as in the homoscedastic case, but allows for more flexibility in the working conditional variance function with specifics depending on the estimation method used.
    \begin{itemize}
        \item \textbf{ML: } We model the logarithm of the conditional variance function with a polynomial basis in the covariate $X_i$ with the number of basis functions (restricted to at most $4$ to avoid numerical instabilities) determined by cross-validation using Gaussian log-likelihood loss. This is carried out using the \texttt{nlme} package \citep{nlme-code}.
        \item \textbf{GEE: } We perform a cubic penalised spline regression of the squared residuals $\tilde{\varepsilon}_{ij}^2$ onto $X_i$ using the \texttt{mgcv} package \citep{gam-code}) to obtain an estimate of the working conditional variance function with the parameter of the correlation subsequently estimated using the \texttt{geeglm} package \citep{geepack}.
        \item \textbf{Sandwich loss: }We use sandwich boosting as described in Section~\ref{sec:boosting}.
    \end{itemize}
\end{itemize}
Section~\ref{sec:sim-complexity} considers a well-specified setting where the optimal true conditional covariance weights can in principle be replicated by the heteroscedastic weight estimators; Section~\ref{sec:sim-misspec} looks at a misspecified setting where the optimal weights $W_0$ depend on $D$ and varies the degree of misspecification; and Sections~\ref{sec:sim-corr} and \ref{sec:sim-var} explore the effect of varying group sizes in a settings with mildly misspecified conditional correlation and variance respectively. In all cases, we simulate data with equal independent and identically distributed groups of equal size $n$, which we vary across the settings.
The mean squared errors of $\beta$-estimators corresponding to each method and setting pair, averaged over 500 repetitions, are shown in Figure~\ref{fig:simdata_all}. Coverage of associated confidence intervals is reported in Section~\ref{sec:add_results} of the Appendix; in particular the coverage probabilities of nominal 95\% confidence intervals based on our sandwich boosting approach varies from 0.938--0.984.

\subsubsection{Increasing model complexity}\label{sec:sim-complexity}

Consider $I=2000$ i.i.d.\ instances of $(Y,D,X)\in\R^{10}\times\R^{10}\times\R^{10}$ following the partially linear model ~\eqref{eq:PLR} with the following properties: the group size $n=10$ where $\beta=1$ is the target parameter of interest, $X$ is componentwise i.i.d.\ uniform $U[-5,5]$, $m_0(x)=\cos x$, $g_0(x)=\tanh x$ (with the $\cos$ and $\tanh$ functions applied componentwise), $\xi\given X\sim N_{10}({\bf 0},\Omega(X))$ and $\varepsilon\given (D,X)\sim N_{10}({\bf 0},\Sigma(X))$ with covariance matrices given by
\begin{align*}
    \Sigma(X) &=
        \Big(\big(\ind_{\{j=k\}} + 0.2\cdot\ind_{\{j\neq k\}}\big)\sigma_0(X_j)\sigma_0(X_k)\Big)_{(j,k)\in[10]^2},
    \quad
    \sigma_0(x) = 2+\cos(\lambda x),
    \\
    \Omega(X) &=
        \big(\ind_{\{j=k\}} + 0.1\cdot\ind_{\{j\neq k\}}\big)_{(j,k)\in[10]^2}.
\end{align*}
Here $\lambda\geq0$ is the Lipschitz constant (complexity parameter) of the conditional variance function $\sigma_0$, that we will vary. We use homoscedastic and heteroscedastic working covariance classes with equicorrelated working correlation (noting that the true correlation is also constant here). As $\varepsilon\given (D,X)\overset{d}{=}\varepsilon\given X$, by Proposition~\ref{prop:modelmissspec} the population minimiser of the sandwich loss and those corresponding to the ML and GEE approaches should all coincide, and so from this perspective, the former should have no clear advantage in terms of the performance of the resulting $\beta$-estimator. One might therefore expect all heteroscedasticity-accommodating methods to perform similarly here since they need only model the conditional variance $\sigma_0$ sufficiently well to yield the semiparametrically optimal MSE (referred to as the oracle MSE).

The top left panel of Figure~\ref{fig:simdata_all} shows that this appears to be the case when $\lambda \leq 0.75$, but for larger values of the complexity parameter, both the ML and GEE approaches appear to struggle. The latter displays a somewhat erratic trajectory, peaking at an MSE 9 times that of the oracle, and even greatly exceeding that of its homoscedastic counterpart (note that the curves for the homoscedastic GEE and ML estimators almost coincide). This behaviour seems likely to be due to the GEE approach finding local optima, which recall, need not be local optima for the asymptotic MSE objective, i.e., the sandwich loss; see also Example~\ref{ex:cond-corr-misspec} Setting (b) and Appendix~\ref{appex:cond-corr-misspec} for a similar phenomenon.
The sandwich boosted estimator in comparison remains relatively robust to this increase in model complexity, maintaining performance comparable to the oracle estimator.

\subsubsection{Increasing covariance misspecification}\label{sec:sim-misspec}

Here we consider $I=10^4$ instances of the partially linear model~\eqref{eq:PLR} with $n = 4$, $\beta=1$, $g_0(x)=\tanh x$ and $m_0(x)=\cos x$. Errors $(\varepsilon,\xi)$ are generated by introducing an unobserved confounder $\zeta$ between $\varepsilon$ and $\xi$ inspired by the proof of Theorem~\ref{thm:divergingratio} as follows:
\begin{gather*}
    X\sim N_4({\bf 0},0.9{\bf 1}{\bf 1}^{\top}+0.1I_4),
    \quad
    B\given X
     \sim\text{Ber}(p(X)),
    \quad
    p(x) := \ind_{[0,\infty)}(\bar{x}) + \eta^{-1}\ind_{(-\infty,0)}(\bar{x}),
    \\
    \zeta := p(X)^{-1}B,
    \quad
    \xi\given(X,\zeta) := N_4({\bf 0}, \zeta I_4),
    \quad
    \varepsilon\given(D,X,\zeta) := N_4({\bf 0},\zeta\Sigma),
    \quad
    \Sigma = \big(0.2^{|j-k|}\big)_{(j,k)\in[4]^2}.
\end{gather*}
Here $\eta\geq 1$ acts as `misspecification parameter', with larger values indicating greater confounding and violation of the condition \eqref{eq:GEE=MLE=AV} for equivalence of the ML, GEE and sandwich losses. Note that $\bar{x}$ denotes the mean of the entries $x\in\R^n$.

As we see in the top right panel of Figure~\ref{fig:simdata_all}, the performances of the heteroscedastic ML and GEE estimators deteriorate with increasing  $\eta$ and yield worse MSEs than even an unweighted least squares estimator as the extent of covariance misspecification increases. In contrast, despite being equally restricted to use a misspecified class of weights that are a function of $X$ alone, the sandwich boosted $\hat{\beta}$ has a substantially smaller MSE compared  to the approaches  considered, with its advantage increasing with  increasing $\eta$.

\begin{figure}[ht]
    \centering
    \includegraphics[width=0.8\textwidth]{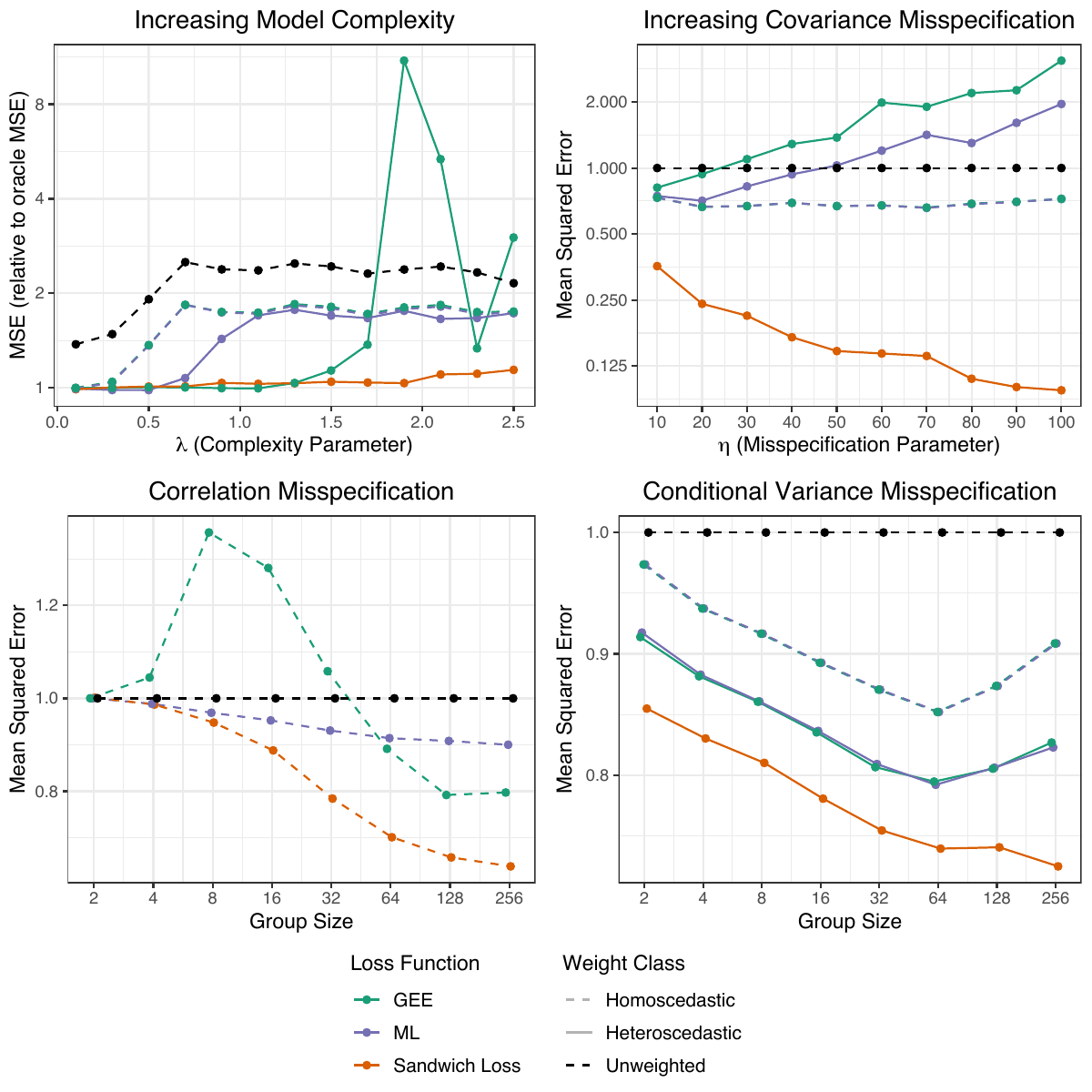}
    \caption{Mean squared error of $\beta$-estimators for the four simulated experiments in Section~\ref{sec:simulations} over 500 independent repeats. The top left plot gives the MSE relative to the oracle, while the others are relative to the unweighted  estimator.}
    \label{fig:simdata_all}
\end{figure}

\subsubsection{Mild conditional correlation misspecification}\label{sec:sim-corr}

We now consider the simple setting of (grouped) linear regression:
\begin{gather*}
    Y = \beta D + \varepsilon,
    \qquad
    D \sim N_n\left({\bf 0}, \frac{1}{8}{\bf 1}{\bf 1}^{\top} + \frac{7}{8}I_{n}\right),
    \\
    \varepsilon \sim \text{ARMA}(2,1),
    \quad
    \phi=(0.3,0.6),
    \quad
    \vartheta=-0.5.
\end{gather*}
Here $\phi$ and $\vartheta$ are the autoregressive and moving average parameters respectively.
We take $\beta=1$ and i.i.d.\ Gaussian innovations for the ARMA process. We consider settings with $I=2^{15}/n$ and group sizes $n\in\left\{2^r:r=1,2,\ldots,8\right\}$. To demonstrate the effect of correlation misspecification, we use a constant working variance and fit all models with an $\text{AR}(1)$ working correlation; this yields a misspecified correlation for group sizes $n\geq 3$. Example~\ref{ex:cond-corr-misspec}, Setting (a) corresponds to this setting with $n=100$.
The bottom left panel of Figure~\ref{fig:simdata_all} shows that the sandwich loss here outperforms the competing approaches, with the GEE approach leading to an inflation in MSE over an unweighted approach for moderate group sizes.

\subsubsection{Mild conditional variance misspecification}\label{sec:sim-var}
We now consider $I=2^{15}/n$ i.i.d.\ instances of the partially linear model~\eqref{eq:PLR} via
\begin{gather*}
    X_j\iid U[-2,2],
    \qquad
    m_0(x) = -6e^{-x},
    \qquad
    g_0(x) = \tanh(x),
    \qquad
    \xi\given X \sim N_n\left({\bf 0}, 9I_{n}\right)
    \\
    \varepsilon\given (D,X) \sim N_n\left({\bf 0},\Sigma(D,X)\right),
    \qquad
    \Sigma_{jk}(D,X) =
        \begin{cases}
        \sigma_0^2(D_j,X_j) & \text{if $j=k$}\\
        0.2\sigma_0(D_j,X_j)\sigma_0(D_k,X_k) & \text{if $j\neq k$}\\
        \end{cases},
    \\
    \sigma_0(d,x) = 2+\tanh\left(d-3x\right),
\end{gather*}
for $n\in\left\{2^r: r=1,2,\ldots,8\right\}$, taking $\beta=1$. Note therefore that as $\Cov(\varepsilon \given D, X) \neq \Cov(\varepsilon \given X)$, the optimal weights that depend also on $D$ are not in the weight classes considered.
All methods are fitted using an equicorrelated working correlation. We again see in the bottom right Figure~\ref{fig:simdata_all} superior performance of sandwich boosting over the competitors here.

\subsection{Real-world data analyses}\label{sec:real-data}
Here we present analyses of two datasets. We fit GEE and (heteroscedastic) sandwich boosting approaches as in the previous section, but use mixed effects models (MEM) to give a family of (working) conditional covariance functions in the ML framework by taking certain covariates as random effects. We continue to use these within Algorithm~\ref{alg:main}, using the \texttt{lme4} package \citep{lmer} to obtain the weights in the MEM case as in \citet{emmenegger}, but reporting robust sandwich estimates of the variance of the $\hat{\beta}$ constructed.
We use $K=5$ folds for cross-fitting. To mitigate the randomness of the resulting estimators on the sample splits themselves we aggregate the $\beta$ and variance estimators obtained over 50 random independent sample splits using the approach of \citet{chern, emmenegger}; see Appendix~\ref{appsec:numerical_results} for details.

\subsubsection{Orange juice price elasticity}

We analyse historical data on orange juice sales, available from the James M. Kilts Center, University of Chicago Booth School of Business \citep{ojdataset}. The dataset is composed of grouped store-level scanner price and sales data over a 121 week period from 83 Dominick’s Finer Foods stores and consists of $N=9649$ observations.

Our goal is to estimate
the price elasticity of a brand of orange juice (Tropicana) during this time period. We do this via a partially linear model~\eqref{eq:PLR}
of the logarithm of the quantity of sales ($Y=\log(\text{Sales})$)
on the logarithm of the price ($D=\log(\text{Price})$) accounting for confounding by events in time ($X=\text{week number}$),  the coefficient $\beta$ of $D$ giving the price elasticity.

\begin{table}[ht]
	\begin{center}
		\begin{tabular}{c*{3}{S[table-format=3.2]}}
    \toprule
	       Method & $\hat{\beta}$ & $\hat{V}$ & {\begin{tabular}{@{}c@{}c@{}}Reduction in $\hat{V}$ relative to\\  homoscedastic GEE estimator (\%)\end{tabular}}
        \\
        \midrule
        \bf{ Sandwich Boosting} & -2.97 & \bf{30.9} & \bf{26.0}
        \\
        Homoscedastic GEE & -3.18 & 41.7 & 0
        \\
        Heteroscedastic GEE & -3.19 & 44.6 & -7.0
        \\
        Intercept only MEM & -3.17 & 41.8 & -0.17
        \\
        Intercept + Time MEM & -3.18 & 42.5 & -2.0
        \\
        \bottomrule
        \end{tabular}
	\caption{Estimates for the price elasticity of orange juice and corresponding variance estimates.}\label{tab:oj}
	\end{center}
\end{table}

Table~\ref{tab:oj} shows the price elasticity estimators and associated sandwich variances estimates $\hat{V}$; note GEE and sandwich boosting approaches here use an equicorrelated working correlation. We see that our sandwich boosting estimator has a 26.0\% reduction in variance compared to the second best homoscedastic GEE estimator. Note that in the cases of both the mixed effects models (MEM) and GEE estimators, as a broader weight class is used, we observe poorer performance in estimating the price elasticity, illustrating the phenomenon described in Figure~\ref{fig:schematic}. Also note that the sandwich boosted and heteroscedastic GEE estimators model weights in the same class. However, whilst the heteroscedastic GEE estimator does not usefully model any heteroscedasticity in the data (and in fact exhibits the worst performance out of all the estimators considered), sandwich boosting successfully estimates a helpful weighting scheme. Figure~\ref{fig:oj-s-function} demonstrates an $s$-function output from sandwich boosting, effectively capturing the general and seasonal trends in volatility that are not learnt by the other estimators.

\subsubsection{National longitudinal survey of young working women}
Here we consider a dataset from the National Longitudinal Survey of Young Women \citep{workingwomen}  containing the wages of 4711 young working women, each measured at approximately 6 time points per woman, totalling $N=28\,534$ observations. We measure the effect of work experience in their current related sector ($D = \text{work experience}$) on the logarithm of wages ($Y = \log(\text{wage})$), controlling for age and tenure ($X = (\text{age}, \text{tenure})$), with weights a function of $X$, using an equicorrelated working correlation for sandwich boosting and GEE approaches. 
\begin{table}[ht]
	\begin{center}
		\begin{tabular}{cccS}
    \toprule
	       Method & $\hat{\beta}$ & $\hat{V}$ & {\begin{tabular}{@{}c@{}c@{}}Reduction in $\hat{V}$ relative to\\homoscedastic GEE estimator (\%)\end{tabular}}
        \\
        \midrule
        \bf{Sandwich Boosting} & 0.040 & \bf{0.083} & \bf{7.4}
        \\
        Homoscedastic GEE & 0.040 & 0.089 & 0
        \\
        Heteroscedastic GEE & 0.039 & 0.091 & -1.6
        \\
        Intercept only MEM & 0.040 & 0.092 & -3.5
        \\
        Intercept + Age + Tenure MEM & 0.042 & 0.120 & -34.9
        \\
        \bottomrule
        \end{tabular}
	\caption{Estimates of $\beta$ and associated variance estimates $\hat{V}$ relating to the effect of work experience on wages from the National Longitudinal Survey of Young Women.}\label{tab:working_women}
	\end{center}
\end{table}

Table~\ref{tab:working_women} gives the $\beta$-estimators and associated variances $\hat{V}$ for the approaches considered. Similarly to the orange juice price elasticity analysis, we see that the broader model classes of the mixed effects model and heteroscedastic GEE estimators yield larger variances than their respective homoscedastic counterparts.
The sandwich boosting estimator gives a modest 7.4\% reduction in variance over the second smallest homoscedastic GEE variance. Interestingly however, this improvement is entirely due  to the sandwich loss rather than the potentially richer model class used by sandwich boosting: the $s$ function output by sandwich boosting is almost constant, and so it effectively uses a constant working variance and estimates a single working correlation parameter. The sandwich loss objective for this correlation (parametrised in terms of $\theta$; see Section~\ref{sec:boosting}) is plotted in Figure~\ref{fig:working-women} alongside  the objectives corresponding to the homoscedastic GEE and intercept only MEM approaches. We see that the respective $\theta$-minimisers differ substantially, with the asymptotic variance (the sandwich loss) evaluated at the minimisers of the GEE and MEM approaches being larger than the asymptotic variance evaluated at any $\theta\in[0.14,1.09]$, corresponding to any working correlation in the range~$\rho\in[0.12,0.51]$.

\begin{figure}[ht]
    \centering
    \begin{subfigure}[b]{0.48\textwidth}
         \centering
        \includegraphics[width=\textwidth]{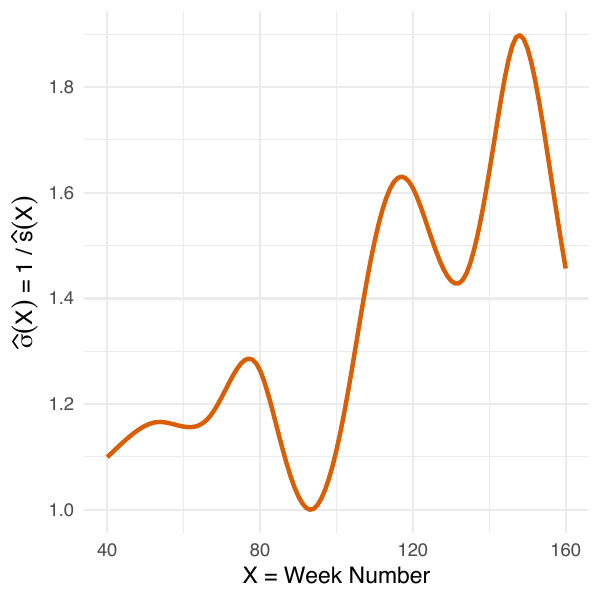}
        \caption{Sandwich boosted $\hat{s}$-function output for the orange juice price elasticity dataset.}
    \label{fig:oj-s-function}
    \end{subfigure}
    \hfill
    \begin{subfigure}[b]{0.48\textwidth}
        \centering
        \includegraphics[width=\textwidth]{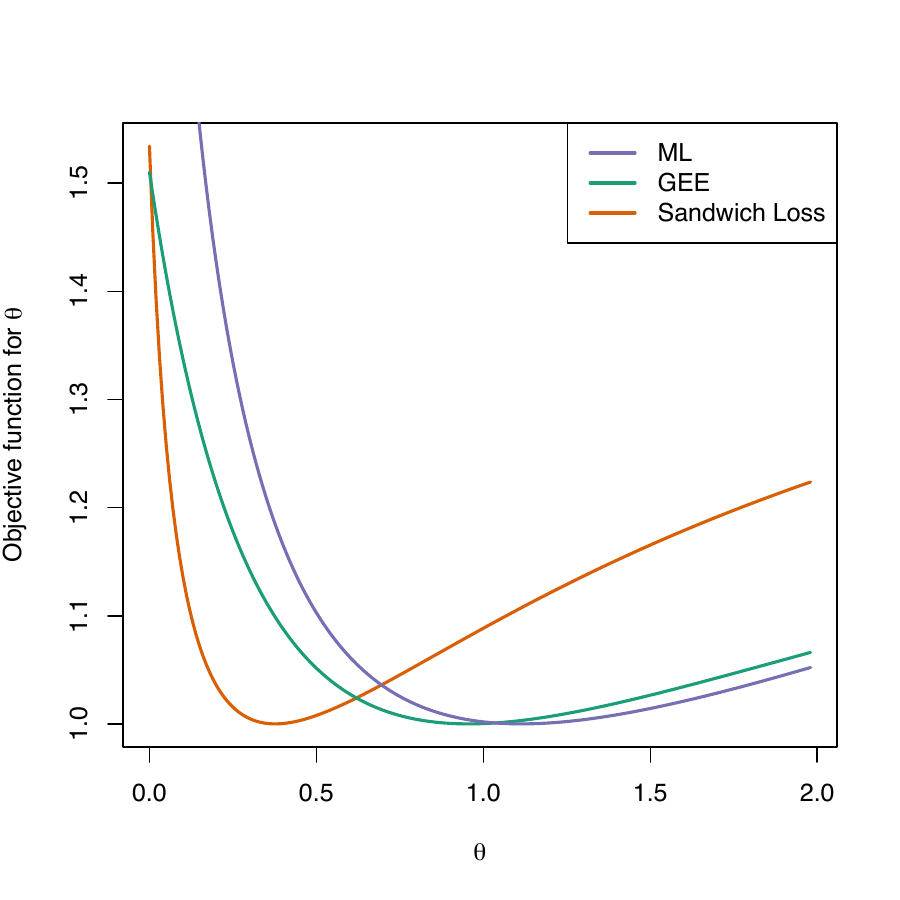}
    \caption{The (scaled) objective functions used to determine $\theta$ for each of the intercept only mixed effects model (ML), homoscedastic GEE estimator (GEE) and sandwich loss with equicorrelated working correlation.}
    \label{fig:working-women}
    \end{subfigure}
    \caption{Outputs relating to data analyses in Section~\ref{sec:real-data}.}
\end{figure}

\section{Discussion}\label{sec:discussion}
In this work we have highlighted and clarified the shortcomings of some popular classical methods in the estimation of weights for weighted least squares-type estimators in partially linear models when the conditional covariance is misspecified. We instead advocate for choosing weights to minimise a sandwich estimate of the variance, what we call the sandwich loss in this context.  A main contribution of ours, in the spirit of the trend towards using machine learning methods for the purposes of statistical inference, is  a practical gradient boosting scheme for approximately minimising this loss over a potentially flexible family of functions defined implicitly through a user-chosen base-learner. Despite the unusual form of our loss that does not decompose as a sum over data points as with the standard case of the empirical risk, we show that for certain versions of our algorithm, the boosting updates can be performed in linear time.

Our work offers a number of directions for future research. On the computational side, it would be useful to investigate broader classes of working correlations that could be accommodated within sandwich boosting to yield linear time updates.
It could also be fruitful to consider the use of the sandwich loss in other classes of models, for example it would be of interest to develop these ideas in the context of generalised (partially) linear marginal models, and beyond.

Thus far we have only considered estimators of a single scalar quantity. In other situations, one may be interested in estimating several parameters simultaneously, and in such cases there are several modifications of the basic sandwich loss that may be helpful to explore. For example, consider the following generalisation of the partially linear model \eqref{eq:PLR}:
 \begin{equation}\label{eq:hetPLR}
 	\begin{split}
 		Y &=\beta(X)\circ D+g_0(X)+\varepsilon,\\
 		D&=m_0(X)+\xi.
 	\end{split}
 \end{equation}
Here $\circ$ denotes the Hadamard product, $\beta(X)$ is a row-wise function of $X$, and all other terms are as before; thus the response of the $j$th observation within the $i$th group satisfies (with a slight abuse of notation) the mean function relationship $\E[Y_{j} \given D,X] = \beta(X_{j})D_{j} + g_0(X_{j})$. We suppose $\beta(X)$ admits the basis expansion
\begin{equation} \label{eq:beta_basis}
    \beta(X) = \sum_{l=1}^L\phi_l\varphi_l(X),
\end{equation}
for some $L$ known basis functions $\left(\varphi_l\right)_{l\in[L]}$ (also row-wise functions of $X$), and unknown vector of parameters $\mbb{\phi}:=(\phi_1,\ldots,\phi_L)\in\R^L$ to be estimated. For example, the model where $\beta(X)=\phi_1+\phi_2 X$ corresponds in the classical linear model setting to fitting an `interaction term' between $D$ and $X$. Given a consistent estimator $\mbb{\hat{\phi}}$ of $\mbb{\phi}$, the mean squared error of the resulting plug-in $\beta$ function estimator $\hat{\beta}(X)=\sum_{l=1}^L\hat{\phi}_l\varphi_l(X)$ satisfies
\begin{gather}\label{eq:beta_function_MSE}
    \E\left[\left(\hat{\beta}(X)-\beta(X)\right)^2\right] \approx 
    \tr\left(\Phi V_{\mbb{\phi}}\right), \;
    \Phi := \big(\E\left[\varphi_l(X)\varphi_{l'}(X)\right]\big)_{l \in [L], l' \in [L]}
    \quad \text{and} \quad
    V_{\mbb{\phi}} := \Var\,\mbb{\hat{\phi}}.
\end{gather}
This suggests the following approach. Consider a class of weighted $\mbb\phi$-estimators $\hat{\mbb\phi} = \hat{\mbb\phi}(W)$ where $W$ is a weight function among a class of $\mathcal{W}$, similarly to our previous framework in Section~\ref{sec:weighted} of weighted $\beta$-estimators. Then given estimates $\hat{\Phi}$ of $\Phi$ and $\hat{V}_{\mbb{\phi}}(W)$ of $\Var \,\hat{\mbb \phi}(W)$, we can consider a generalised sandwich loss of the form
\begin{equation} \label{eq:gen_sand}
\hat{L}_{\mathrm{MSE}(\hat{\beta})}(W) := \tr\big(\hat{\Phi}\hat{V}_{\mbb{\phi}}(W)\big),
\end{equation}
which we may attempt to minimise using a sandwich boosting approach; further details are given in Appendix~\ref{appsec:het_treatment_effect}.

Our sandwich loss $\hat{L}_{\SL}$ is defined with respect to estimated errors $\tilde{\xi}_i$ and $\tilde{\varepsilon}_i$ derived from initial regressions, which in particular, take no advantage of the dependence structure in the data, unlike the final estimate $\hat{\beta}$. Clearly initial weighted  regressions could deliver improved estimates of the errors, in turn giving an improved estimate of $\hat{\beta}$. This suggests a scheme with weights and residuals being updated iteratively, analogous to iterative generalised least squares \citep{goldstein1986multilevel,goldstein1989restricted}. In the simple linear model setting, where the residuals are derived from linear regressions, a generalised sandwich loss of the form \eqref{eq:gen_sand} may be appropriate for delivering accurate estimates of the errors. How to do this for a general regression is less clear but would certainly be worthy of further investigation.

\paragraph*{Funding.} EHY was supported by the EPSRC Doctoral Training Partnership and RDS was supported by an EPSRC Programme Grant EP/N031938/1.

\paragraph*{Data availability.} The data supporting the findings of this study are available at \url{https://www.chicagobooth.edu/research/kilts/research-data/dominicks} and \url{https://www.stata-press.com/data/r10/nlswork.dta}, and all code used to produce the numerical results in this article is available at \url{https://github.com/elliot-young/sandwichboost_simulations}.

\end{cbunit}

\newpage
\appendix

\begin{cbunit}
\begin{center}
	\Large \bf Supplementary material relating to `Sandwich Boosting for Accurate Estimation in Partially Linear Models for Grouped Data', by Elliot H.\ Young and Rajen D.\ Shah
\end{center}

\section{Supplementary material relating to Section~\ref{sec:intro}}

\subsection{Further details for Example~\ref{ex:cond-corr-misspec}}\label{appex:cond-corr-misspec}
Recall that in Example~\ref{ex:cond-corr-misspec} we consider the grouped linear model
\begin{equation*}
    Y_i = \beta D_i + \varepsilon_i,
\end{equation*}
where $(Y_i,D_i)\in\R^{n}\times\R^{n}$ are iid; note that since we plot the population level objective functions of the ML and GEE approaches, along with the asymptotic variance of the resulting $\hat{\beta}$ for each given value of $\rho$, there is no need to specify the total number of groups $I$. In Setting (a), we take $n=100$ and $\varepsilon_i$ as the first $n$ realisations of an $\text{ARMA}(2,1)$ model with unit variance and autoregressive parameters $\phi=(0.3,0.6)$ and moving average parameter $\vartheta=-0.5$. Further, we take $D_i\sim N_{n}({\bf 0}, \frac{1}{8}{\bf 1}{\bf 1}^{\top}+\frac{7}{8}I_{n})$.

In Setting (b), we take $n=30$ and the $\text{ARMA}(2,1)$ model parameters $\phi=(0.1, 0.85)$, $\vartheta=-0.4$. Compared to the Setting (a), the $\text{AR}(1)$ working correlation model is less severely misspecified. In this setting, the $\hat{\rho}$ obtained by minimising the ML objective results in a $\hat{\beta}$ that has the same asymptotic mean squared error (MSE) as the unweighted estimator (obtained by setting $\rho=0$ in this working correlation model, i.e.\ not attempting to model the correlation at all). The GEE objective is also non-convex in $\rho$, and in fact a gradient descent procedure minimising the GEE objective initialised at $\rho=0$ (or a Newton optimisation approach as used in the \texttt{geepack} package \citep{geepack}) will incorrectly estimate $\hat{\rho}$ at the local minimum at $\hat{\rho}=-0.71$. Using this estimate for $\rho$ results in a $\hat{\beta}$ with MSE $2.4$ times that of the unweighted estimator (and $3.4$ times that of the $\hat{\beta}$ resulting from minimising the sandwich loss); see Table~\ref{tab:misspec-corr}.

\begin{table}[ht]
	\begin{center}
		\begin{tabular}{cSc}
  \toprule
            Objective function & $\hat{\rho}$ & {\begin{tabular}{@{}c@{}c@{}}Asymptotic MSE of $\hat{\beta}$\\  relative to sandwich loss\end{tabular}}
              \\
            \midrule
            Unweighted & 0.00 & 1.4
            \\ 
			GEE (\texttt{geepack}) &   -0.71   &    3.4   \\
			ML (\texttt{gls}) &   0.00   &    1.4   \\
			Sandwich loss (\texttt{sandwich.boosting}) &   0.30   &    1.0  \\
   \bottomrule
		\end{tabular}
		\caption{Mean squared errors of the weighted $\beta$-estimators weighted using ML (implemented in \texttt{gls}), GEE1 (\texttt{geepack}) and gradient descent of the sandwich loss (\texttt{sandwich.boosting}) for the simulation example in Section~\ref{appex:cond-corr-misspec}.}\label{tab:misspec-corr}
	\end{center}
\end{table}

\subsection{Further details on Example~\ref{ex:cond-var-misspec}} \label{appex:cond-var-misspec}
We give further details on misspecified conditional variance motivating example in Section~\ref{sec:cond_cov}. Note specifically this example compares the three loss functions considered at the population level (and not the sample level) therefore exploring asymptotic properties of each loss function. We consider the ungrouped ($n_i\equiv 1$) setting where the random variables $(Y,D,X)\in\R\times\R\times\R$ satisfy
\begin{equation*}
    X\sim U[0,1], \hspace{1cm} D=\xi, \hspace{1cm} Y=D\beta+\epsilon,
\end{equation*}
for some (arbitrary) $\beta\in\R$ and where $\E\left[\varepsilon\given D,X\right]=\E\left[\xi\given X\right]=0$, $\E\left[\xi^2\given X\right]=1$ and $\sigma_0^2(X):=\Var\left(\varepsilon\given D,X\right)=2+\tanh(\lambda(X-\mu))$ for constants $\lambda\geq0$ and $\mu\in[0,1]$.

We consider using a misspecified class of functions of the form
\begin{equation*}
     \sigma(x;\eta) = 1 + 2 \ind_{[\eta, \infty)}(x),
\end{equation*}
where $\eta\in\R$, with estimates $\hat{\eta}$ of an optimal $\eta$ minimising one of the ML, GEE or SL losses (each method discussed in Section~\ref{sec:intro}). At the population level for this example these losses can be expressed as
\begin{equation*}
\begin{gathered}
    L_{\text{ML}}(\eta) = \E\left[\log\left(\sigma^{2}(X;\eta)\right)+\frac{\sigma_0^2(X)}{\sigma^{2}(X;\eta)}\right],
    \quad
    L_{\GEE}(\eta) = \E\left[\left(\sigma_0^2(X)-\sigma^{2}(X;\eta)\right)^2\right],
    \\
    \text{and }\hspace{3mm}L_{\text{SL}}(\eta) = \left(\E\left[\frac{1}{\sigma^{2}(X;\eta)}\right]\right)^{-2}\left(\E\left[\frac{\sigma_0^2(X)}{\sigma^{4}(X;\eta)}\right]\right).
\end{gathered}
\end{equation*}
The model class considered is therefore parametrised by a single parameter $\eta$, with each loss function minimised either analytically or via a line search. 

Given a minimising $\hat{\eta}$ of one the three losses, the (scaled) mean squared errors of the $\sigma_0$-estimator and $\beta$-estimator (weighted using the $\sigma_0$-estimator) can be calculated analytically as
\begin{equation*}
    \E\left[\left(\sigma(X;\hat{\eta})-\sigma_0(X)\right)^2\right],
    \text{ and }
    \left(\E\left[\frac{1}{\sigma^{2}(X;\hat{\eta})}\right]\right)^{-2}\left(\E\left[\frac{\sigma_0^2(X)}{\sigma^{4}(X;\hat{\eta})}\right]\right),
\end{equation*}
respectively, with results as given in Table~\ref{tab:MLE_vs_AV}

\section{Proof of results in Section~\ref{sec:method}}\label{appsec:sec2results}
\subsection{Proof of Proposition~\ref{prop:modelmissspec}}\label{appsec:modelmissspec-proof}

We first show that $L_{\GEE}$ is minimised over $\mathcal{W}$ by $\E\left[\varepsilon\varepsilon^{\top}\given q(X)\right]^{-1}$; this follows immediately from the decomposition of $L_{\GEE}(W)$ as
\begin{align*}
    \E\left[\big\|W(X)^{-1}-\varepsilon\varepsilon^{\top}\big\|^2\right] = & \E\left[\big\|W(X)^{-1}-\E\big[\varepsilon\varepsilon^{\top}\given q(X)\big]\big\|^2\right] + \E\left[\big\|\varepsilon\varepsilon^{\top}\big\|^2-\big\|\E\big[\varepsilon\varepsilon^{\top}\given q(X)\big]\big\|^2\right] \\
    & \qquad -2\E\left[\E\left[\big\langle\varepsilon\varepsilon^{\top}-\E\big[\varepsilon\varepsilon^{\top}\given q(X)\big],W(X)^{-1}\big\rangle \,\Big|\, q(X)\right]\right], \\
    = & \E\left[\big\|W(X)^{-1}-\E\big[\varepsilon\varepsilon^{\top}\given q(X)\big]\big\|^2\right] + \E\left[\big\|\varepsilon\varepsilon^{\top}\big\|^2-\big\|\E\big[\varepsilon\varepsilon^{\top}\given q(X)\big]\big\|^2\right].
\end{align*}
Note this result holds for any matrix norm derived from an inner product.

We now show that $L_{\ML}$ is minimised over $\mathcal{W}$ by $\E\left[\varepsilon\varepsilon^{\top}\given q(X)\right]^{-1}$. Noting that $L_{\ML}(W) = \E\left[\ell_{\ML}(W(X))\right]$ where
\begin{equation*}
    \ell_{\ML}(W(X)):=-\log\det W(X) + \tr\left(W(X)\E\big[\varepsilon\varepsilon^{\top}\given q(X)\big]\right),
\end{equation*}
we aim to minimise the function \(x\mapsto\ell_{\ML}(W(x))\) for arbitrary $x\in\R^{n\times d}$. We define $\bar{W}\in\mathcal{W}$ by $\bar{W}(x):=\E\left[\varepsilon\varepsilon^{\top}\given q(X)=q(x)\right]^{-1}$. Then given some fixed $x\in\R^{n\times d}$ and $W\in\mathcal{W}$ we define the function \(g(t)=\ell_{\ML}((1-t)\bar{W}(x)+tW(x))\) for \(t\in[0,1]\). The second order Taylor expansion of $g$ then gives
\begin{equation*}
    \ell_{\ML}(W(x)) = \ell_{\ML}(\bar{W}(x)) + g'(0) + \frac{g''(\tilde{t})}{2},
\end{equation*}
for some \(\tilde{t}\in[0,1]\). Now, by Jacobi's formula,
\begin{equation*}
    g'(t) = -\tr\left(\left[\left\{(1-t)\bar{W}(x)+tW(x)\right\}^{-1}-\E\big[\varepsilon\varepsilon^{\top}\given q(X)=q(x)\big]\right]\left\{W(x)-\bar{W}(x)\right\}\right),
\end{equation*}
and so
\begin{equation*}
    g'(0) = -\tr\left(\left\{\tilde{W}(x)^{-1}-\E\big[\varepsilon\varepsilon^{\top}\given q(X)=q(x)\big]\right\}\left\{W(x)-\bar{W}(x)\right\}\right) = 0.
\end{equation*}
Further, by direct calculation
\begin{multline*}
    g''(\tilde{t}) = \tr\left(\left(\left\{(1-\tilde{t})\bar{W}(x)+\tilde{t}W(x)\right\}\left\{W(x)-\bar{W}(x)\right\}\right)^2\right) \\
    = \norm{\left\{(1-\tilde{t})\bar{W}(x)+\tilde{t}W(x)\right\}^{\frac{1}{2}}\left\{W(x)-\bar{W}(x)\right\}\left\{(1-\tilde{t})\bar{W}(x)+\tilde{t}W(x)\right\}^{\frac{1}{2}}}_{\F}^2\geq0,
\end{multline*}
where $\norm{\cdot}_{\F}$ denotes the Frobenius norm, and with equality holding if and only if \(W(x)=\bar{W}(x)\). Therefore whenever $W(x) \neq \bar{W}(x)$,
\begin{equation*}
    \ell_{\ML}(W(x))>\ell_{\ML}(\tilde{W}(x)),
\end{equation*}
and so $L_{\ML}$ is minimised over $\mathcal{W}$ by \(\E\left[\varepsilon\varepsilon^{\top}\given q(X)\right]^{-1}\).

Finally, we will show that
\begin{align*}
    V_{\SL} :=& \inf_{W\in\mathcal{W}} L_{\SL}(W) 
    \\
    =& \left(\E\left[\E\big[\text{vec}\big(\xi\xi^{\top}\big)\given q(X)\big]\E\Big[\text{vec}\big(\xi\varepsilon^{\top}\big)\text{vec}\big(\varepsilon\xi^{\top}\big)^{\top} \,\big|\, q(X)\Big]^{-1}\E\big[\text{vec}\big(\xi\xi^{\top}\big)\given q(X)\big]\right]\right)^{-1}.
\end{align*}
This can be shown by noting that for $W\in\mathcal{W}$
\begin{align*}
    L_{\SL}(W) &= \frac{\E\left[\xi^{\top}W\varepsilon\varepsilon^{\top}W\xi\right]}{\left(\E\left[\xi^{\top}W\xi\right]\right)^2}
    = \frac{\E\left[\text{vec}(W)^{\top}\E\left[\text{vec}\big(\xi\varepsilon^{\top}\big)\text{vec}\big(\varepsilon\xi^{\top}\big)^{\top} \,\big|\, q(X)\right]\text{vec}(W)\right]}{\left(\E\left[\E\big[\text{vec}\big(\xi\xi^{\top}\big)\given q(X)\big]^{\top}\text{vec}(W)\right]\right)^2} 
    \\
    &\geq \left(\E\left[\E\big[\text{vec}\big(\xi\xi^{\top}\big)\given q(X)\big]\E\left[\text{vec}(\xi\varepsilon^{\top})\text{vec}\big(\varepsilon\xi^{\top}\big)^{\top} \,\big|\, q(X)\right]^{-1}\E\big[\text{vec}\big(\xi\xi^{\top}\big)\given q(X)\big]\right]\right)^{-1}
    \\
    &= \left(\E\left[\E\big[\text{vec}\big(\xi\xi^{\top}\big)\given q(X)\big]\E\left[\big(\varepsilon\xi^{\top}\big)\otimes\big(\xi\varepsilon^{\top}\big) \,\big|\, q(X)\right]^{-1}\E\big[\text{vec}\big(\xi\xi^{\top}\big)\given q(X)\big]\right]\right)^{-1}
    ,
\end{align*}
by the Cauchy--Schwarz inequality, and by using the identities $a^{\top}Ab=\text{vec}(ab^{\top})^{\top}\text{vec}(A)$ and $\text{vec}(ab^{\top})\text{vec}(cd^{\top})^{\top}=(bd^{\top})\otimes(ac^{\top})$ (for vectors \(a,b,c,d\in\R^n\) and matrix \(A\in\R^{n\times n}\)), where $\otimes$ denotes the Kronecker product. Equality holds in the above if and only if
\begin{equation*}
    W\in\left\{W\in\mathcal{W}: \text{vec}(W(X))=c\,\E\left[\big(\varepsilon\xi^{\top}\big)\otimes\big(\xi\varepsilon^{\top}\big) \,\big|\, q(X)\right]^{-1}\E\big[\text{vec}\big(\xi\xi^{\top}\big) \given q(X)\big] \text{ for some } c>0\right\}.
\end{equation*}
Therefore $V_{\SL}\leq V_{\text{ML/GEE}}$ with equality if and only if the condition
\begin{equation*}
    \E\left[\big(\varepsilon\xi^{\top}\big)\otimes\big(\xi\varepsilon^{\top}\big) \,\big|\, q(X)\right]\text{vec}\left(\E\big[\varepsilon\varepsilon^{\top}\given q(X)\big]^{-1}\right) = c\,\text{vec}\left(\E\big[\xi\xi^{\top}\given q(X)\big]\right),
\end{equation*}
holds for some $c>0$. Using the identity $\text{vec}(ABC)=(C^{\top}\otimes A)\text{vec}(B)$ (for matrices $A,B,C\in\R^{n\times n}$) this condition is equivalent to
\begin{equation*}
    \E\left[\big(\xi\varepsilon^{\top}\big)\,\E\big[\varepsilon\varepsilon^{\top}\given q(X)\big]^{-1}\big(\varepsilon\xi^{\top}\big)\,\big|\, q(X)\right] = c\, \E\big[\xi\xi^{\top}\given q(X)\big].
\end{equation*}
Further,
\begin{align*}
    \E\left[\big(\xi\varepsilon^{\top}\big)\,\E\big[\varepsilon\varepsilon^{\top}\given q(X)\big]^{-1}\big(\varepsilon\xi^{\top}\big)\,\big|\, q(X)\right] &= 
    \E\left[\xi\,\tr\big[\Cov(\varepsilon\given q(X))^{-1}\Cov(\varepsilon\given D,X)\big]\xi^{\top} \,\big|\, q(X)\right]
    \\
    & =\E\left[\tr\big[\Cov(\varepsilon\given q(X))^{-1}\Cov(\varepsilon\given D,X)\big]\,\xi\xi^{\top} \,\big|\, q(X)\right],
\end{align*}
and so this condition is equivalent to
\begin{equation}\label{eq:appendix-GEE=MLE=SL-condition}
    \E\left[\tr\big[\Cov(\varepsilon\given q(X))^{-1}\Cov(\varepsilon\given D,X)\big]\,\xi\xi^{\top} \,\big|\, q(X)\right] = c\, \E\big[\xi\xi^{\top}\given q(X)\big].
\end{equation}
In  particular when $q(X)=X$ and $\Cov(\varepsilon\given D,X)=\Cov(\varepsilon\given X)$ the condition~\eqref{eq:appendix-GEE=MLE=SL-condition} holds.

\subsection{Proof of Theorem~\ref{thm:divergingratio}}\label{appsec:divergingratio-proof}
    We will adopt the shorthand $\Sigma_q = \Sigma(q(X))$ and $\Omega_q = \Omega(q(X))$. As $q(X)$ is almost surely non-constant there exists measurable $\Gamma \subset \R^{n \times d}$ such that \(\PP(q(X)\in\Gamma)\in(0,1)\). With  this, we define the random variables \(Q:=\eta^{-1} \ind_{\{q(X)\in\Gamma\}} + \ind_{\{q(X)\in\Gamma^c\}}\) for some \(\eta>1\) to be determined later, and \(\zeta:=Q^{-1}B-1\) where \(B\sim\text{Ber}(Q)\) with \(B\independent X \given Q\). Further, define \(\xi:=\sqrt{1+\zeta}R_{\xi}\) and \(\varepsilon:=\sqrt{1+\zeta}R_{\varepsilon}\) where \(R_{\xi}\given (X,B)\sim N(0,\Omega_q)\) and \(R_{\varepsilon}\given (X,B,R_{\xi})\sim N(0,\Sigma_q)\).
    
    By construction we then have \(\E[\zeta\given X]=0\) and hence $\E\left[\xi\given X\right]=\E\left[\varepsilon\given D,X\right]=0, \E\left[\xi\xi^{\top}\given q(X)\right]=\Omega_q$ and $\E\left[\varepsilon\varepsilon^{\top}\given q(X)\right]=\Sigma_q$. 
    Further, we can show that
    \begin{gather*}
        \E\left[\zeta^2\given q(X)\right]=\frac{1-Q}{Q},
    \end{gather*}
    and thus
    \begin{align*}
        \E\left[\big(\varepsilon\xi^{\top}\big)\otimes\big(\xi\varepsilon^{\top}\big) \,\big|\, q(X)\right]
        &= \E\left[(1+\zeta)^2\big(R_{\varepsilon}R_{\xi}^{\top}\big)\otimes\big(R_{\xi}R_{\varepsilon}^{\top}\big) \,\big|\,  q(X)\right] \\
        &= \left(1+\E\left[\zeta^2\given q(X)\right]\right)\E\left[\big(R_{\varepsilon}R_{\xi}^{\top}\big)\otimes\big(R_{\xi}R_{\varepsilon}^{\top}\big) \,\big|\, q(X)\right]\\
        &=Q^{-1}\E\left[\big(R_{\varepsilon}R_{\xi}^{\top}\big)\otimes\big(R_{\xi}R_{\varepsilon}^{\top}\big) \,\Big|\,  q(X)\right].
    \end{align*}
    Therefore, using $V_\ML, V_\GEE$ and $V_\SL$ as in the proof of Proposition~\ref{prop:modelmissspec} (also noting that \(V_\ML=V_\GEE\)),
\begin{align*}
         \frac{V_{\text{ML/GEE}}}{V_{\SL}}
          = & \left(\E\left[\text{vec}\left(\E\big[\xi\xi^{\top}\given q(X)\big]\right)^{\top}\text{vec}\left(\E\big[\varepsilon\varepsilon^{\top}\given q(X)\big]^{-1}\right)\right]\right)^{-2}
          \\ & \cdot
          \E\left[\text{vec}\left(\E\big[\varepsilon\varepsilon^{\top}\given q(X)\big]^{-1}\right)^{\top}\E\left[\big(\varepsilon\xi^{\top}\big)\otimes\big(\xi\varepsilon^{\top}\big) \,\big|\, q(X)\right]\text{vec}\left(\E\big[\varepsilon\varepsilon^{\top}\given q(X)\big]^{-1}\right)\right]
          \\ & \cdot
          \E\left[\text{vec}\left(\E\big[\xi\xi^{\top}\given q(X)\big]\right)^{\top}\E\left[\big(\varepsilon\xi^{\top}\big)\otimes\big(\xi\varepsilon^{\top}\big) \,\big|\, q(X)\right]^{-1}\text{vec}\left(\E\big[\xi\xi^{\top}\given q(X)\big]\right)\right] \\
          = & \left(\E\left[\text{vec}\left(\Omega_q\right)^{\top}\text{vec}\left(\Sigma^{-1}_q\right)\right]\right)^{-2}
          \\ & \cdot
          \underbrace{\E\left[Q^{-1}\text{vec}\left(\Sigma^{-1}_q\right)^{\top}\E\left[\big(R_\varepsilon R_\xi^{\top}\big)\otimes\big(R_\xi R_\varepsilon^{\top}\big) \,\big|\, q(X)\right]\text{vec}\left(\Sigma^{-1}_q\right)\right]}_{=:A_1}
          \\ & \cdot
          \underbrace{\E\left[Q\,\text{vec}\big(\Omega_q\big)^{\top}\E\left[\big(R_\varepsilon R_\xi^{\top}\big)\otimes\big(R_\xi R_\varepsilon^{\top}\big) \,\big|\, q(X)\right]^{-1}\text{vec}\big(\Omega_q\big)\right]}_{=:A_2}.
\end{align*}
Then, 
\begin{align*}
    A_1 &= \E\left[Q^{-1}\text{vec}\left(\Sigma^{-1}_q\right)^{\top}\E\left[\big(R_\varepsilon R_\xi^{\top}\big)\otimes\big(R_\xi R_\varepsilon^{\top}\big) \,\big|\, q(X)\right]\text{vec}\left(\Sigma^{-1}_q\right)\right]
    \\
    &= \E\left[Q^{-1}\E\left[\text{vec}\left(\Sigma^{-1}_q\right)^{\top} \big\{\big(R_\varepsilon R_\xi^{\top}\big)\otimes\big(R_\xi R_\varepsilon^{\top}\big)\big\} \text{vec}\left(\Sigma^{-1}_q\right) \,\big|\, q(X)\right]\right]
    \\
    &= \E\left[Q^{-1}\E\left[\text{vec}\left(\Sigma^{-1}_q\right)^{\top} \text{vec}\left(\E\left[R_\xi R_\varepsilon^{\top}\Sigma^{-1}_qR_\xi R_\varepsilon^{\top} \,\big|\, q(X)\right]\right) \,\big|\, q(X)\right]\right]
    \\
    &= \E\left[Q^{-1}\E\left[\tr\Big(\Sigma^{-1}_qR_\xi R_\xi\Sigma^{-1}_qR_\varepsilon R_\varepsilon^{\top}\Big) \,\big|\, q(X)\right]\right]
    \\
    &= \E\left[Q^{-1} \tr\big(\Sigma^{-1}_q\Omega_q\big) \right],
\end{align*}
and
\begin{align*}
    A_2 &= \E\left[Q\,\text{vec}\big(\Omega_q\big)^{\top}\E\left[\big(R_\varepsilon R_\xi^{\top}\big)\otimes\big(R_\xi R_\varepsilon^{\top}\big) \,\big|\, q(X) \right]^{-1}\text{vec}\big(\Omega_q\big)\right]
    \\
    &= \E\left[Q\,\text{vec}\big(\Omega_q\big)^{\top}\text{vec}\big(\Sigma^{-1}_q\big)\right]
    \\
    &= \E\left[Q \,\tr\big(\Sigma^{-1}_q\Omega_q\big) \right],
\end{align*}
where we use
\begin{equation*}
    \E\left[\big(R_\varepsilon R_\xi^{\top}\big)\otimes\big(R_\xi R_\varepsilon^{\top}\big) \,\big|\, q(X) \right]^{-1} \, \text{vec}(\Omega_q) = \text{vec}\big(\Sigma_q^{-1}\big),
\end{equation*}
which follows as
\begin{align*}
    \Omega_q &= \Omega_q \Sigma_q^{-1} \Sigma_q
    \\
    &= \E\left[R_\xi R_\xi^{\top}\Sigma^{-1}_qR_\varepsilon R_\varepsilon^{\top} \,\big|\, q(X)\right]
    \\
    &= \E\left[R_\xi R_\varepsilon^{\top}\Sigma^{-1}_qR_\xi R_\varepsilon^{\top} \,\big|\, q(X)\right]
    \\
    &= \E\left[\big(R_\varepsilon R_\xi^{\top}\big)\otimes\big(R_\xi R_\varepsilon^{\top}\big) \,\big|\, q(X) \right]
    \\
    \Rightarrow
    \text{vec}(\Omega_q) &= \E\left[\big(R_\varepsilon R_\xi^{\top}\big)\otimes\big(R_\xi R_\varepsilon^{\top}\big) \,\big|\, q(X) \right] \, \text{vec}\big(\Sigma_q^{-1}\big).
\end{align*}
Define
\begin{equation*}
    t(X) := \tr\big(\Sigma^{-1}_q\Omega_q\big).
\end{equation*}
Then
\begin{align*}
    \frac{V_{\text{ML/GEE}}}{V_{\SL}} &= \frac{\E\left[Qt(X)\right]\E\left[Q^{-1}t(X)\right]}{\left(\E\left[t(X)\right]\right)^2}
    \\
    &= 
    {\left(\E\left[t(X)\right]\right)}^{-2}
    \\
    &\qquad\cdot
    \big\{\eta^{-1}\E\left[t(X)\given q(X)\in\Gamma\right]\PP(q(X)\in\Gamma) + \E\left[t(X)\given q(X)\in\Gamma^c\right]\PP(q(X)\in\Gamma^c)\big\}
    \\
    &\qquad\cdot
    \big\{\eta\,\E\left[t(X)\given q(X)\in\Gamma\right]\PP(q(X)\in\Gamma) + \E\left[t(X)\given q(X)\in\Gamma^c\right]\PP(q(X)\in\Gamma^c)\big\}
    \\
    &\geq \eta\,\frac{\E\left[t(X)\given q(X)\in\Gamma\right]\E\left[t(X)\given q(X)\in\Gamma^c\right]}{\left(\E\left[t(X)\right]\right)^2}
    \cdot
    \PP(q(X)\in\Gamma)\PP(q(X)\in\Gamma^c) \geq M,
\end{align*}
taking
\begin{equation*}
    \eta = \left( M^{-1}\frac{\E\left[t(X)\given q(X)\in\Gamma\right]\E\left[t(X)\given q(X)\in\Gamma^c\right]}{\left(\E\left[t(X)\right]\right)^2}
    \cdot
    \PP(q(X)\in\Gamma)\PP(q(X)\in\Gamma^c) \right)^{-1} \vee 2.
\end{equation*}

\subsection{Regularity conditions for Theorem~\ref{thm:parametric-consistency}}\label{appsec:regularitycond}

For the statement of Theorem~\ref{thm:parametric-consistency} we also impose the following regularity conditions:
\begin{enumerate}[label=(\roman*)]
    \item\label{ass:cts_map} The map $\psi\mapsto W(\psi)(x)$ is continuous for all $x\in\R^{n\times d}$.
    \item\label{ass:inf_min_eval} $\inf\limits_{\psi\in\Psi}\inf\limits_{x\in\R^{n\times d}}\Lambda_{\min}(W(\psi)(x)) > 0$.
    \item\label{ass:l4_moments} The errors satisfy the moment bounds $\E\left[\|\varepsilon\|_2^4\right]\leq C_\varepsilon$ and $\E\left[\|\xi\|_2^4\right]\leq C_\xi$ for finite constants $C_\varepsilon,C_\xi>0$.
\end{enumerate}

\subsection{Proof of Theorem~\ref{thm:parametric-consistency}}

Using Theorem~5.7 of \citet{vandervaart} it suffices to show
\begin{equation}\label{eq:uniflln}
    \sup_{\psi\in\Psi}\left| \hat{L}_{\SL,I}(\psi) - L_{\SL}(\psi) \right|=o_{P}(1).
\end{equation}
Define the error terms
\begin{equation*}
        \mathcal{E}_i^{(\xi)} := \tilde{\xi}_i-\xi_i, \qquad \mathcal{E}_i^{(\varepsilon)} := \tilde{\varepsilon}_i-\varepsilon_i,
\end{equation*}
and subsequently the terms
\begin{align*}
    a_i^{(1)}(\psi) &:= \mathcal{E}_i^{(\xi)T}W_i(\psi)\mathcal{E}_i^{(\varepsilon)}, \\
    a_i^{(2)}(\psi) &:=  \mathcal{E}_i^{(\xi)T}W_i(\psi)\varepsilon_i + \mathcal{E}_i^{(\varepsilon)T}W_i(\psi)\xi_i \\
    \mathcal{A}_I(\psi) &:= \frac{1}{N}\sum_{i=1}^I\left(\xi_i^{\top}W(\psi)(X_i)\varepsilon_i\right)^2 - \frac{1}{n}\E\left[\left(\xi^{\top}W(\psi)(X)\varepsilon\right)^2\right],\\
    A_I(\psi) &:= \mathcal{A}_I(\psi) + \underbrace{\frac{2}{N}\sum_{i=1}^I\left(\xi_i^{\top}W_i(\psi)\varepsilon_i\right)\left(a_i^{(1)}(\psi) + a_i^{(2)}(\psi)\right) + \frac{1}{N}\sum_{i=1}^I\left(a_i^{(1)}(\psi)+a_i^{(2)}(\psi)\right)^2}_{=:G_I(\psi)},\\
    b_i(\psi) &:= 2\mathcal{E}_i^{(\xi)T}W_i(\psi)\xi_i + \mathcal{E}_i^{(\xi)T}W_i(\psi)\mathcal{E}_i^{(\xi)}, \\
    \mathcal{B}_I(\psi) &:= \frac{1}{N}\sum_{i=1}^I\xi_i^{\top}W(\psi)(X_i)\xi_i - \frac{1}{n}\E\left[\xi^{\top}W(\psi)(X)\xi\right] ,\\
    B_I(\psi) &:= \mathcal{B}_I(\psi) + \frac{1}{N}\sum_{i=1}^Ib_i(\psi) .
\end{align*}
Then we can decompose
\begin{align*}
    \sup_{\psi\in\Psi}\left| \hat{L}_{\SL,I}(\psi) - L_{\SL}(\psi) \right| 
    &=
    \sup_{\psi\in\Psi}\left|\frac{N^{-1}\sum_{i=1}^I\left(\hat{\xi}_i^{\top}W(\psi)(X_i)\hat{\varepsilon}_i\right)^2}{\left(N^{-1}\sum_{i=1}^I\hat{\xi}_i^{\top}W(\psi)(X_i)\hat{\xi}_i\right)^2}-\frac{n^{-1}\E\left[\left(\xi^{\top}W(\psi)(X)\varepsilon\right)^2\right]}{\left(n^{-1}\E\left[\xi^{\top}W(\psi)(X)\xi\right]\right)^2}\right| \\
    &= \sup_{\psi\in\Psi}\left|\frac{n^{-1}\E\left[\left(\xi^{\top}W(\psi)(X)\varepsilon\right)^2\right] + A_I(\psi)}{\left(n^{-1}\E\left[\xi^{\top}W(\psi)(X)\xi\right]\right)^2 + B_I(\psi)} - \frac{n^{-1}\E\left[\left(\xi^{\top}W(\psi)(X)\varepsilon\right)^2\right]}{\left(n^{-1}\E\left[\xi^{\top}W(\psi)(X)\xi\right]\right)^2}\right|.
\end{align*}
First we claim that $\sup_{\psi\in\Psi}\left|\mathcal{A}_I(\psi)\right|=o_P(1)$. This follows by the uniform law of large numbers \citep{jennrich}, which may be applied as $\E\left[\sup_{\psi\in\Psi}\left(\xi^{\top}W(\psi)(X)\varepsilon\right)^2\right] \leq C_\varepsilon^{\frac{1}{2}}C_\xi^{\frac{1}{2}}<\infty$ due to Assumption~\ref{ass:l4_moments}. It can similarly be shown that $\sup_{\psi\in\Psi}\left|\mathcal{B}_I(\psi)\right|=o_P(1)$.
Now,
\begin{align*}
    &\E\left[\sup_{\psi\in\Psi}\bigg|\frac{1}{N}\sum_{i=1}^Ib_i(\psi)\bigg|\right]
    \\
    \leq& \frac{1}{N}\sum_{i=1}^I\left(2\E\left[\sup_{\psi\in\Psi}\left|\mathcal{E}_i^{(\xi)T}W_i(\psi)\xi_i\right|\right] + \E\left[\sup_{\psi\in\Psi}\mathcal{E}_i^{(\xi)T}W_i(\psi)\mathcal{E}_i^{(\xi)}\right]\right) \\
    \leq& 2\E\left[\frac{1}{N}\sum_{i=1}^I\|\mathcal{E}_i^{(\xi)}\|_2\|\xi_i\|_2\right] + \E\left[\frac{1}{N}\sum_{i=1}^I\|\mathcal{E}_i^{(\xi)}\|_2^2\right]
    \\
    \leq& 2\E\left[\frac{1}{N}\sum_{i=1}^I\|\mathcal{E}_i^{(\xi)}\|_2^2\right]^{\frac{1}{2}}\E\left[\frac{1}{N}\sum_{i=1}^I\|\xi_i\|_2^2\right]^{\frac{1}{2}} + \E\left[\frac{1}{N}\sum_{i=1}^I\|\mathcal{E}_i^{(\xi)}\|_2^2\right]
    \\
    =& o(1),
\end{align*}
and so by Markov's inequality, $\sup_{\psi\in\Psi}\left|N^{-1}\sum_{i=1}^Ib_i(\psi)\right|=o_P(1)$.
Further,
\begin{align*}
    &\sup_{\psi\in\Psi}\left|G_I(\psi)\right|
    \\
    \leq&
    \sup_{\psi\in\Psi}\frac{2}{N}\sum_{i=1}^I\left|\xi_i^{\top}W_i(\psi)\varepsilon_i\right|\cdot\left|a_i^{(1)}(\psi)\right| + \sup_{\psi\in\Psi}\frac{2}{N}\sum_{i=1}^I\left|\xi_i^{\top}W_i(\psi)\varepsilon_i\right|\cdot\left|a_i^{(2)}(\psi)\right|
    \\
    & \quad +
    \sup_{\psi\in\Psi}\frac{2}{N}\sum_{i=1}^I\left(a_i^{(1)}(\psi)\right)^2 + \sup_{\psi\in\Psi}\frac{2}{N}\sum_{i=1}^I\left(a_i^{(2)}(\psi)\right)^2
    \\
    \leq& 2\cdot\underbrace{\frac{1}{N}\sum_{i=1}^I\norm{\xi_i}\norm{\varepsilon_i}\norm{\mathcal{E}_i^{(\xi)}}\norm{\mathcal{E}_i^{(\varepsilon)}}}_{=: G_I^{(\text{\RN{1}})}} + 2\cdot\underbrace{\frac{1}{N}\sum_{i=1}^I\left(\norm{\xi_i}\norm{\varepsilon_i}\left\{\norm{\mathcal{E}_i^{(\xi)}}\norm{\varepsilon_i} + \norm{\mathcal{E}_i^{(\varepsilon)}}\norm{\xi_i}\right\}\right)}_{=: G_I^{(\text{\RN{2}})}} 
    \\
    &\quad +
    2\cdot\underbrace{\frac{1}{N}\sum_{i=1}^I\norm{\mathcal{E}_i^{(\xi)}}^2\norm{\mathcal{E}_i^{(\varepsilon)}}^2}_{=: G_I^{(\text{\RN{3}})}} + 4\cdot\underbrace{\frac{1}{N}\sum_{i=1}^I\left(\norm{\mathcal{E}_i^{(\xi)}}^2\norm{\varepsilon_i}^2 + \norm{\mathcal{E}_i^{(\varepsilon)}}^2\norm{\xi_i}^2\right)}_{=: G_I^{(\text{\RN{4}})}}
\end{align*}
by the triangle equality. We proceed by showing that each term on the right hand side is $o_P(1)$. In the case of Assumption~\ref{ass:lp4} we have
\begin{align*}
    \E\left[G_I^{(\text{\RN{1}})}\right]
        &\leq
        \E\left[\frac{1}{N}\sum_{i=1}^I\|\xi_i\|_2^2\|\varepsilon_i\|_2^2\right]^{\frac{1}{2}}\E\left[\frac{1}{N}\sum_{i=1}^I\|\mathcal{E}_i^{(\xi)}\|_2^2\|\mathcal{E}_i^{(\varepsilon)}\|_2^2\right]^{\frac{1}{2}} \\
        &\leq
        C_{\varepsilon}^{\frac{1}{4}}C_{\xi}^{\frac{1}{4}}\E\left[\frac{1}{N}\sum_{i=1}^I\|\mathcal{E}_i^{(\xi)}\|_2^4\right]^{\frac{1}{4}}\E\left[\frac{1}{N}\sum_{i=1}^I\|\mathcal{E}_i^{(\varepsilon)}\|_2^4\right]^{\frac{1}{4}}
        = o(1),
\end{align*}
by the Cauchy--Schwarz inequality (twice). We can similarly show (again by multiple uses of the Cauchy--Schwarz inequality) that $\E[G_I^{(\text{\RN{2}})}], \E[G_I^{(\text{\RN{3}})}], \E[G_I^{(\text{\RN{4}})}] = o(1)$. It then follows by Markov's inequality (alongside Lemma~\ref{lem:unif_add}) that $G_I^{(\text{\RN{1}})},G_I^{(\text{\RN{2}})},G_I^{(\text{\RN{3}})},G_I^{(\text{\RN{4}})}=o_P(1)$, and so $\sup_{\psi\in\Psi}\left|G_I(\psi)\right|=o_P(1)$. If alternatively Assumption~\ref{ass:lp2} holds then we consider each term separately:

    \paragraph{Term $G_I^{(\text{\RN{1}})}$: } 
    \begin{align*}
        \E\left[\big(G_I^{(\text{\RN{1}})}\big)^{\frac{2}{3}}\right]
        &\leq
        2^{\frac{2}{3}}N^{\frac{1}{3}}\E\left[\frac{1}{N}\sum_{i=1}^I\|\xi_i\|_2^{\frac{2}{3}}\|\varepsilon_i\|_2^{\frac{2}{3}}\|\mathcal{E}_i^{(\xi)}\|_2^{\frac{2}{3}}\|\mathcal{E}_i^{(\varepsilon)}\|_2^{\frac{2}{3}}\right]
        \\
        &\leq  N^{\frac{1}{3}}\E\left[\frac{1}{N}\sum_{i=1}^I\|\varepsilon_i\|_2^2\|\xi_i\|_2^2\right]^{\frac{1}{3}}\E\left[\frac{1}{N}\sum_{i=1}^I\|\mathcal{E}_i^{(\varepsilon)}\|_2\|\mathcal{E}_i^{(\xi)}\|_2\right]^{\frac{2}{3}}
        \\
        &\leq N^{\frac{1}{3}}\E\left[\frac{1}{N}\sum_{i=1}^I\|\mathcal{E}_i^{(\varepsilon)}\|_2^2\right]^{\frac{1}{3}}\E\left[\frac{1}{N}\sum_{i=1}^I\|\mathcal{E}_i^{(\xi)}\|_2^2\right]^{\frac{1}{3}}\E\left[\frac{1}{N}\sum_{i=1}^I\|\varepsilon_i\|_2^4\right]^{\frac{1}{6}}\E\left[\frac{1}{N}\sum_{i=1}^I\|\xi_i\|_2^4\right]^{\frac{1}{6}}
        \\
        & = o(1),
    \end{align*}
    the first inequality follows as $(\sum_ix_i)^\eta\leq\sum_ix_i^\eta$ for any $\eta\in(0,1], x_i\geq0$, the second follows by H\"older's inequality, and the third by the Cauchy--Schwarz inequality.
    
    \paragraph{Term $G_I^{(\text{\RN{2}})}$: }
    \begin{align*}
        \E\left[\big(G_I^{(\text{\RN{2}})}\big)^{\frac{4}{5}}\right]
        &\leq 
        N^{\frac{1}{5}}\E\left[\frac{1}{N}\sum_{i=1}^I\|\xi_i\|_2^{\frac{4}{5}}\|\varepsilon\|_2^{\frac{8}{5}}\|\mathcal{E}_i^{(\xi)}\|_2^{\frac{4}{5}}\right]
        +
        N^{\frac{1}{5}}\E\left[\frac{1}{N}\sum_{i=1}^I\|\xi_i\|_2^{\frac{8}{5}}\|\varepsilon\|_2^{\frac{4}{5}}\|\mathcal{E}_i^{(\varepsilon)}\|_2^{\frac{4}{5}}\right]
        \\
        &\leq
        N^{\frac{1}{5}}\E\left[\frac{1}{N}\sum_{i=1}^I\|\xi_i\|_2^{\frac{4}{3}}\|\varepsilon_i\|_2^{\frac{8}{3}}\right]^{\frac{3}{5}}\E\left[\frac{1}{N}\sum_{i=1}^I\|\mathcal{E}_i^{(\xi)}\|_2^2\right]^{\frac{2}{5}}
        \\
        & +
        N^{\frac{1}{5}}\E\left[\frac{1}{N}\sum_{i=1}^I\|\xi_i\|_2^{\frac{8}{3}}\|\varepsilon_i\|_2^{\frac{4}{3}}\right]^{\frac{3}{5}}\E\left[\frac{1}{N}\sum_{i=1}^I\|\mathcal{E}_i^{(\varepsilon)}\|_2^2\right]^{\frac{2}{5}}
        \\
        &\leq
        N^{\frac{1}{5}}\E\left[\frac{1}{N}\sum_{i=1}^I\|\mathcal{E}_i^{(\xi)}\|_2^2\right]^{\frac{2}{5}}\E\left[\frac{1}{N}\sum_{i=1}^I\|\mathcal{E}_i^{(\xi)}\|_2^2\right]^{\frac{1}{3}}\E\left[\frac{1}{N}\sum_{i=1}^I\|\xi_i\|_2^4\right]^{\frac{1}{5}}\E\left[\frac{1}{N}\sum_{i=1}^I\|\varepsilon_i\|_2^4\right]^{\frac{2}{5}}
        \\
        & + N^{\frac{1}{5}}\E\left[\frac{1}{N}\sum_{i=1}^I\|\mathcal{E}_i^{(\varepsilon)}\|_2^2\right]^{\frac{2}{5}}\E\left[\frac{1}{N}\sum_{i=1}^I\|\mathcal{E}_i^{(\xi)}\|_2^2\right]^{\frac{1}{3}}\E\left[\frac{1}{N}\sum_{i=1}^I\|\xi_i\|_2^4\right]^{\frac{2}{5}}\E\left[\frac{1}{N}\sum_{i=1}^I\|\varepsilon_i\|_2^4\right]^{\frac{1}{5}}
        \\
        &=o(1),
    \end{align*}
    where the first inequality follows as $(\sum_ix_i)^\eta\leq\sum_ix_i^\eta$ for any $\eta\in(0,1], x_i\geq0$, the second follows by the Cauchy--Schwarz inequality, and the third by H\"older's inequality.
    
    \paragraph{Term $G_I^{(\text{\RN{3}})}$: }
    \begin{equation*}
        \E\left[\big(G_I^{(\text{\RN{3}})}\big)^{\frac{1}{2}}\right]
        \leq
        N^{\frac{1}{2}}
        \E\left[\frac{1}{N}\sum_{i=1}^I\|\mathcal{E}_i^{(\xi)}\|_2^2\right]^{\frac{1}{2}}\E\left[\frac{1}{N}\sum_{i=1}^I\|\mathcal{E}_i^{(\varepsilon)}\|_2^2\right]^{\frac{1}{2}}
        =o(1),
    \end{equation*}
    by $(\sum_ix_i)^\eta\leq\sum_ix_i^\eta$ for any $\eta\in(0,1], x_i\geq0$ and using the Cauchy--Schwarz inequality.
    
    \paragraph{Term $G_I^{(\text{\RN{4}})}$: }
    \begin{align*}
        \E\left[\big(G_I^{(\text{\RN{4}})}\big)^{\frac{2}{3}}\right]
        &\leq
        N^{\frac{1}{3}}\E\left[\frac{1}{N}\sum_{i=1}^I\|\mathcal{E}_i^{(\xi)}\|_2^{\frac{4}{3}}\|\varepsilon_i\|_2^{\frac{4}{3}}\right]
        +
        N^{\frac{1}{3}}\E\left[\frac{1}{N}\sum_{i=1}^I\|\mathcal{E}_i^{(\varepsilon)}\|_2^{\frac{4}{3}}\|\xi_i\|_2^{\frac{4}{3}}\right]
        \\
        &\leq
        N^{\frac{1}{3}}\E\left[\frac{1}{N}\sum_{i=1}^I\|\mathcal{E}_i^{(\xi)}\|_2^2\right]^{\frac{2}{3}}\E\left[\frac{1}{N}\sum_{i=1}^I\|\varepsilon_i\|_2^4\right]^{\frac{1}{3}}
        \\
        &\quad +
        N^{\frac{1}{3}}\E\left[\frac{1}{N}\sum_{i=1}^I\|\mathcal{E}_i^{(\varepsilon)}\|_2^2\right]^{\frac{2}{3}}\E\left[\frac{1}{N}\sum_{i=1}^I\|\xi_i\|_2^4\right]^{\frac{1}{3}}
    \end{align*}
By applying Markov's inequality to each of the above terms we have $G_I^{\text{\RN{1}}},G_I^{\text{\RN{2}}},G_I^{\text{\RN{3}}},G_I^{\text{\RN{4}}}=o_P(1)$, and therefore with Lemma~\ref{lem:unif_add} we have $\sup_{\psi\in\Psi}\left|G_I(\psi)\right|=o_P(1)$.

Therefore, combining the results above with the triangle inequality
\begin{equation}\label{eq:ABoP(1)}
    \sup_{\psi\in\Psi}\left|A_I(\psi)\right| = o_P(1) \quad\text{and}\quad \sup_{\psi\in\Psi}\left|B_I(\psi)\right| = o_P(1).
\end{equation}
Now, by the triangle inequality
\begin{align*}
    \sup_{\psi\in\Psi}\left| \hat{L}_{\SL,I}(\psi) - L_{\SL}(\psi) \right| 
    &\leq
    \left(\sup_{\psi\in\Psi}\left|\frac{1}{\left(\E\left[\xi^{\top}W(\psi)(X)\xi\right] + B_I(\psi)\right)^2}-\frac{1}{\left(\E\left[\xi^{\top}W(\psi)(X)\xi\right]\right)^2}\right|\right)\cdot \\
    &\hspace{-2em}\cdot\left(\sup_{\psi\in\Psi}\E\left[\left(\xi^{\top}W(\psi)(X)\varepsilon\right)^2\right] + \sup_{\psi\in\Psi}\left|A_I(\psi)\right|\right) + \frac{\sup_{\psi\in\Psi}\left|A_I(\psi)\right|}{\inf_{\psi\in\Psi}\left(\E\left[\xi^{\top}W(\psi)(X)\xi\right]\right)^2}.
\end{align*}
We have already shown that  $\E\left[\sup_{\psi\in\Psi}\left(\xi^{\top}W(\psi)(X)\varepsilon\right)^2\right]<\infty$ and it is similarly straightforward to show that $\E\left[\inf_{\psi\in\Psi}\xi^{\top}W(\psi)(X)\xi\right]>0$. Therefore, it remains to show that
\begin{equation}\label{eq:denomso(1)}
    \sup_{\psi\in\Psi}\left|\frac{1}{\left(\E\left[\xi^{\top}W(\psi)(X)\xi\right] + B_I(\psi)\right)^2}-\frac{1}{\left(\E\left[\xi^{\top}W(\psi)(X)\xi\right]\right)^2}\right| = o_P(1).
\end{equation}
To show this, it will be helpful to define
\begin{align*}
    c &:= \E\bigg[\inf_{\psi\in\Psi}\xi^{\top}W(\psi)(X)\xi\bigg]>0,
    \\
    C &:= \E\bigg[\sup_{\psi\in\Psi}\xi^{\top}W(\psi)(X)\xi\bigg]<\infty
\end{align*}
Then, further defining $B_I^{\sup} := \sup_{\psi\in\Psi}\left|B_I(\psi)\right|$ it follows that
\begin{align*}
    \PP_P\Bigg(\sup_{\psi\in\Psi}\Bigg|&\frac{1}{\left(\E\left[\xi^{\top}W(\psi)(X)\xi\right] + B_I(\psi)\right)^2} -\frac{1}{\left(\E\left[\xi^{\top}W(\psi)(X)\xi\right]\right)^2}\Bigg| > \epsilon \Bigg) \\
    &\leq
    \PP_P\left(\frac{2cB_I^{\sup}-\left(B_I^{\sup}\right)^2}{C^4\left(1+c^{-1}B_I^{\sup}\right)^2} > \epsilon\right) \\
    &\leq
    \PP_P\left(\frac{2cB_I^{\sup}-\left(B_I^{\sup}\right)^2}{C^4\left(1+c^{-1}B_I^{\sup}\right)^2} > \tilde{\epsilon}\right)
    \tag*{\text{where $\tilde{\epsilon}:=\min\left(\epsilon,\frac{1}{6}C^{-4}c^2\right)$}}
    \\
    &\leq
    \PP_P\left(B_I^{\sup}>\lambda\right)
    \tag*{\text{where $\lambda := \left(1+c^{-2}C^4\tilde{\epsilon}\right)^{-1}\left( c-c^{-1}C^4\tilde{\epsilon}-\sqrt{c^2-3C^4\tilde{\epsilon}} \right) > 0$}}
    \\
        &=\PP_P\left(\sup_{\psi\in\Psi}\left|B_I(\psi)\right|>\lambda\right) = o(1) \tag*{\text{by \eqref{eq:ABoP(1)}.}}
\end{align*}
The penultimate step follows by solving the quadratic inequality in $B_I^{\sup}$. This proves the claim~\eqref{eq:denomso(1)} and hence~\eqref{eq:uniflln}, completing the proof.

\section{Further details on sandwich boosting}\label{appsec:workingcorrelations}

\subsection{Working correlation structures that exhibit computational benefits}

We outline three common working correlation structures that allow for computationally fast sandwich boosting, at computational order $O(N)$ as opposed to $O\big(\sum_i n_i^3\big)$. 

\paragraph{Equicorrelated:}
The equicorrelated working correlation structure (and its scaled inverse) is given by
\begin{align*}
    C_{\theta}(X_i) &= \left(\ind_{\{j=k\}} + \ind_{\{j\neq k\}} \rho_\theta\right)_{(j,k)\in[n_i]^2},
    \\
    C^{-1}_\theta(X_i) &= \left(\ind_{\{j=k\}} - \frac{\theta}{1 + \theta n_i}\right)_{(j,k)\in[n_i]^2},
\end{align*}
where the constant correlation $\rho_\theta = \theta/(1+\theta)$ is given in terms of the parameter $\theta\in[0,\infty)$. This working correlation is particularly popular for repeated measures data, with this class of weights in the homoscedastic case (setting $\sigma^*\equiv s^*\equiv 1$) the same class as used when fitting an intercept only mixed effects model.

\paragraph{Autoregressive, $\text{AR(1)}$:}
The autoregressive structure is given by the working correlation matrix (and its scaled inverse) of an $\text{AR}(1)$ process
\begin{align*}
    C_{\theta}(X_i) &= \left(\theta^{|j-k|}\right)_{(j,k)\in[n_i]^2},
    \\
    C^{-1}_\theta(X_i) &=
    \begin{cases}
    \left( \ind_{\{j=k\}} + \theta^2\ind_{\{1<j=k<n_i\}} - \theta \ind_{\{|j-k|=1\}} \right)_{(j,k)\in[n_i]^2} & \text{if } n_i > 1
    \\
    (1-\theta^2) &\text{if } n_i = 1
    \end{cases}.
\end{align*}
This working correlation is therefore of particular relevance for longitudinal datasets where the correlation between observations is bounded by a function that decays exponentially with separation.

\paragraph{Hierarchical / Nested:}
The nested working correlation models hierarchical groupings as follows. Consider the $i$th group (of size $n_i$) to be further partitioned into $M$ sub-groups of sizes $(n_{i1},\ldots,n_{iM})$ (such that $\sum_{m=1}^Mn_{im}=n_i$). The nested group structure is then given by the block matrix
\begin{equation*}
    C_\theta(X_i) =         \begin{pmatrix}
        E_{i,11} & \cdots & E_{i,1M} \\
        \vdots & \ddots & \vdots \\
        E_{i,M1} & \cdots & E_{i,MM} \\
        \end{pmatrix},
        \quad
        E_{i,mm'}=\begin{cases}
        \rho_1 J_{n_{im}\times n_{im}} + (1-\rho_1)I_{n_{im}} &\text{if }m=m' \vspace{1mm}\\
        \rho_2 J_{n_{im}\times n_{im'}} &\text{if }m\neq m',
        \end{cases}
\end{equation*}
for the vector $\rho=(\rho_1,\rho_2)\in[0,1)\times[0,1)$ satisfying $\rho_2\leq\rho_1$ (suppressing the dependence on $\theta$ in $\rho_\theta$), where $J_{a\times b}$ denotes an $a\times b$ dimensional matrix of ones and $I_a$ denotes the $a\times a$ dimensional identity matrix. This nested correlation can be reparametrised in terms of $\theta=(\theta_1,\theta_2):=\left(\frac{\rho_1-\rho_2}{1-\rho_1},\frac{\rho_2}{1-\rho_1}\right)\in[0,\infty)\times[0,\infty)$ as
\begin{equation*}
\begin{gathered}
    C_\theta^{-1}(X_i) = (1+\theta_1+\theta_2)I_{n_i} - \mathcal{B}_i(\theta), \\
    \mathcal{B}_i(\theta) =
        \begin{pmatrix}
        \mathcal{A}_{i,11}(\theta)J_{n_{i1}\times n_{i1}} & \mathcal{A}_{i,12}(\theta)J_{n_{i1}\times n_{i2}} & \cdots & \mathcal{A}_{i,1M}(\theta)J_{n_{i1}\times n_{iM}}\\
        \mathcal{A}_{i,21}(\theta)J_{n_{i2}\times n_{i1}} & \mathcal{A}_{i,22}(\theta)J_{n_{i2}\times n_{i2}} & \cdots & \mathcal{A}_{i,2M}(\theta)J_{n_{i2}\times n_{iM}}\\
        \vdots & \vdots & \ddots & \vdots \\
        \mathcal{A}_{i,M1}(\theta)J_{n_{iM}\times n_{i1}} & \mathcal{A}_{i,M2}(\theta)J_{n_{iM}\times n_{i2}} & \cdots & \mathcal{A}_{i,MM}(\theta)J_{n_{iM}\times n_{iM}}\\
        \end{pmatrix},\\
        \mathcal{A}_{i,mm'}(\theta) = \frac{\theta_1}{1+\theta_1n_{im}}\ind_{\{m=m'\}}+\frac{\theta_2}{1+\theta_2\sum_{l=1}^{M}\frac{n_{il}}{1+\theta_1n_{il}}}\cdot\frac{1}{\left(1+\theta_1n_{im}\right)\left(1+\theta_1n_{im'}\right)} \quad(m,m'\in[M]), \\
\end{gathered}
\end{equation*}
calculated by the Woodbury matrix identity.

\subsection{Calculation of s-scores for computationally efficient sandwich boosting}\label{appsec:boosting-algorithms}

Note that, as discussed in Section~\ref{sec:sandwich_loss}, the scores for sandwich boosting can be calculated in $O\big(\sum_i n_i^3\big)$ computations for a generic (inverse) working correlation $C_\theta^{-1}(\cdot)$, but the three working correlations introduced are special cases that can be calculated in $O(N)$ operations. We first introduce a general algorithm to calculate scores for a generic working correlation parametrised by $\theta$, and then present this theory applied to the special cases of the working correlations introduced (where we have a reduction in computational order).

\begin{algorithm}[ht]
 \KwIn{Index set $\mathcal{I}$ for boosting training; Working correlation structure $\mathcal{C}_{i}$~\eqref{eq:C-cal}; Estimates of (grouped) errors $(\hat{\xi}_i,\hat{\varepsilon}_i,X_i)_{i\in\mathcal{I}}$; Estimates of $(s,\theta)$ at which scores to be calculated.}

Calculate $A^{(1)}$, $A^{(2)}$: 
\\
\For{$i\in\mathcal{I}$} {
    $A^{(1)}_i 
    := \hat{\xi}_i^{\top}W_i\hat{\xi}_i = \sum_{k=1}^{n_i}\sum_{k'=1}^{n_i}\mathcal{C}_{ikk'}(\theta)s(X_{ik})s(X_{ik'})\hat{\xi}_{ik}\hat{\xi}_{ik'},$ \\
    $A^{(2)}_i := \hat{\xi}_i^{\top}W_i\hat{\varepsilon}_i = \sum_{k=1}^{n_i}\sum_{k'=1}^{n_i}\mathcal{C}_{ikk'}(\theta)s(X_{ik})s(X_{ik'})\mathcal{C}_{ikk'}(\theta)\hat{\xi}_{ik}\hat{\varepsilon}_{ik'}.$
        }

Calculate $A^{(3)}, A^{(4)}$:
\\
\For{$i\in\mathcal{I}$ and $j\in[n_i]$} {
    $A^{(3)}_{ij} := 2\sum_{k=1}^{n_i}\mathcal{C}_{ijk}(\theta)s(X_{ik})\hat{\xi}_{ij}\hat{\xi}_{ik},$ \\
    $A^{(4)}_{ij} := \sum_{k=1}^{n_i}\mathcal{C}_{ijk}(\theta)s(X_{ik})\left(\hat{\xi}_{ij}\hat{\varepsilon}_{ik} + \hat{\xi}_{ik}\hat{\varepsilon}_{ij}\right).$
        }

Calculate $A^{(1)}_{\text{sum}} := \sum_{i\in\mathcal{I}}A^{(1)}_i$ and $A^{(2)}_{\text{sum.sq}} := \sum_{i\in\mathcal{I}}A^{(2)^2}_i$.
\\
Calculate $s$-scores:
\\
\For{$i\in\mathcal{I}$ and $j\in[n_i]$} {
$
    U_{\text{SL}}^{(s)}(s,\theta)(X_{ij}) = -2\left(A^{(1)}_{\text{sum}}\right)^{-3}\left[A^{(2)}_{\text{sum.sq}}A^{(3)}_{ij} - A^{(1)}_{\text{sum}}A^{(2)}_iA^{(4)}_{ij}\right].
$
        }

Calculate the $\theta$-score:
\\
$    U_{\text{SL}}^{(\theta)}(s,\theta) =-2\left(A^{(1)}_{\text{sum}}\right)^{-3}\Bigg[A^{(2)}_{\text{sum.sq}}\left(\sum_{i\in\mathcal{I}}\sum_{k=1}^{n_i}\sum_{k'=1}^{n_i}\frac{\partial\mathcal{C}_{ikk'}(\theta)}{\partial\theta}s(X_{ik})s(X_{ik'})\hat{\xi}_{ik}\hat{\xi}_{ik'}\right)$ \\ \qquad $ - A^{(1)}_{\text{sum}}\left(\sum_{i\in\mathcal{I}}A_i^{(2)}\sum_{k=1}^{n_i}\sum_{k'=1}^{n_i}\left(\frac{\partial\mathcal{C}_{ikk'}(\theta)}{\partial\theta}s(X_{ik})s(X_{ik'})\hat{\xi}_{ik}\hat{\varepsilon}_{ik'}\right)\right)\Bigg].
$

\KwOut{The set of $N$ $s$-scores $\left(X_{ij},\, U_{\SL}^{(s)}(s,\theta)(X_{ij})\right)_{i\in\mathcal{I},j\in[n_i]}$ and $\theta$-score $U_{\SL}^{(\theta)}(s,\theta)$.}
 \caption{Computationally lean sandwich boosting score calculation for an arbitrary working correlation}
 \label{alg:scores-arbitrary}
\end{algorithm}

Therefore, the order of computation to calculate the score functions for sandwich boosting is $O\big(\sum_i n_i^3\big)$ for a general correlation matrix structure whose inverse $C_\theta^{-1}(\cdot)$ is not known (and so needs to be computed analytically). In the special case of Toeplitz correlation structures, or where the inverse correlation structure $C_\theta^{-1}(\cdot)$ is known, this order of computation can be reduced to $O\big(\sum_i n_i^2\big)$. Note that this computational order can be further reduced to $O(N)$ in the three special cases of correlation structures posed in this paper. We give an example of these, implementing Algorithm~\ref{alg:scores-arbitrary} to the equicorrelation structure, with scores calculated in $O(N)$ operations as in Algorithm~\ref{alg:scores-equicorr}. For notational simplicity for these algorithms we introduce the notation
\begin{equation}\label{eq:C-cal}
    \mathcal{C}_{ijk}(\theta) := \left\{C_\theta^{-1}(X_i)\right\}_{jk},
\end{equation}
for the inverse working correlation structure.

\begin{algorithm}[ht]
 \KwIn{Index set $\mathcal{I}$ for boosting training; Estimates of (grouped) errors $(\hat{\xi}_i,\hat{\varepsilon}_i,X_i)_{i\in\mathcal{I}}$; Estimates of $(s,\theta)$ at which scores to be calculated.}

Calculate $A^{(1)},A^{(2)},B^{(1)},B^{(2)}$:
\\
\For{$i\in\mathcal{I}$} {
$    B^{(1)}_i := \sum_{k=1}^{n_i}s(X_{ik})\hat{\xi}_{ik},$
\\
$    B^{(2)}_i := \sum_{k=1}^{n_i}s(X_{ik})\hat{\varepsilon}_{ik},$
\\
$    A^{(1)}_i 
    := \hat{\xi}_i^{\top}W_i\hat{\xi}_i = \sum_{k=1}^{n_i}s^2(X_{ik})\hat{\xi}_{ik}^2-\frac{\theta}{1+\theta n_i}\left(\sum_{k=1}^{n_i}s(X_{ik})\hat{\xi}_{ik}\right)^2,$
\\
$A^{(2)}_i := \hat{\xi}_i^{\top}W_i\hat{\varepsilon}_i = \sum_{k=1}^{n_i}s^2(X_{ik})\hat{\xi}_{ik}\hat{\varepsilon}_{ik}-\frac{\theta}{1+\theta n_i}\left(\sum_{k=1}^{n_i}s(X_{ik})\hat{\xi}_{ik}\right)\Big(\sum_{k=1}^{n_i}s(X_{ik})\hat{\varepsilon}_{ik}\Big).$
        }

Calculate $A^{(3)}$ and $A^{(4)}$:
\\
\For{$i\in\mathcal{I}$ and $j\in[n_i]$} {
$    A^{(3)}_{ij} := 2s(X_{ij})\hat{\xi}_{ij}^2-\frac{2\theta}{1+\theta n_i}B^{(1)}_i\hat{\xi}_{ij},$
\\
$    A^{(4)}_{ij} := 2s(X_{ij})\hat{\xi}_{ij}\hat{\varepsilon}_{ij} - \frac{\theta}{1+\theta n_i}\left(B^{(2)}_i\hat{\xi}_{ij}+B^{(1)}_i\hat{\varepsilon}_{ij}\right).
$
        }

Calculate $A^{(1)}_{\text{sum}} := \sum_{i=1}^IA^{(1)}_i$ and $A^{(2)}_{\text{sum.sq}} := \sum_{i=1}^IA^{(2)^2}_i$.

Calculate $s$-scores:
\\
\For{$i\in\mathcal{I}$ and $j\in[n_i]$} {
$
    U_{\text{SL}}^{(s)}(s,\theta)(X_{ij}) = -2\left(A^{(1)}_{\text{sum}}\right)^{-3}\left[A^{(2)}_{\text{sum.sq}}A^{(3)}_{ij} - A^{(1)}_{\text{sum}}A^{(2)}_iA^{(4)}_{ij}\right].
$
        }

Calculate the $\theta$-score:
\\
$
    U_{\text{SL}}^{(\theta)}(s,\theta) = -2\left(A^{(1)}_{\text{sum}}\right)^{-3}\left[A^{(2)}_{\text{sum.sq}}\left(\sum_{i=1}^I\frac{\left(B^{(1)}_i\right)^2}{\left(1+\theta n_i\right)^2}\right) - A^{(1)}_{\text{sum}}\left(\sum_{i=1}^I\frac{A^{(2)}_iB^{(1)}_iB^{(2)}_i}{\left(1+\theta n_i\right)^2}\right)\right].
$

\KwOut{The set of $N$ $s$-scores $\left(X_{ij},\, U_{\SL}^{(s)}(s,\theta)(X_{ij})\right)_{i\in\mathcal{I},j\in[n_i]}$ and $\theta$-score $U_{\SL}^{(\theta)}(s,\theta)$.}
 \caption{Computationally lean sandwich boosting score calculation for the equicorrelated working correlation}
 \label{alg:scores-equicorr}
\end{algorithm}

Similar algorithms exist for the autoregressive $\text{AR}(1)$ and nested groups working correlations (and are also working correlation cases that achieve the faster $O(N)$ computational rates when calculating scores  for performing sandwich boosting).

\subsection{Theoretical results under working correlation structures}

The uniform asymptotic convergence results of Section~\ref{sec:theory} works on the general working covariance model parametrised by $(\sigma,\rho)$ functions and / or parameters. For computational speed and stability for sandwich boosting we consider the weight class reparametrised in terms of $(s,\theta)$. For completeness we reformulate the assumptions on the deterministic limits $(\sigma^*,\rho^*)$ of Assumption~\ref{ass:nuisance_rates} in terms of the deterministic limits $(s^*,\theta^*)$, with a particular focus on the three prototypical working correlation parametrisations considered in Section~\ref{appsec:workingcorrelations}.

\begin{corollary}\label{thm:(s,theta)-all}
    Consider the setting of Theorem~\ref{thm:main} where the weight class $\mathcal{W}$ takes one of the three forms outlined in Section~\ref{appsec:workingcorrelations} in terms of $(s,\theta)$. Further, suppose there exists the deterministic function $s^*:\R^d\to\R$ and vector $\theta^*$ whose estimators in place of Assumptions \ref{A3.1}-\ref{A3.3} in Assumption~\ref{ass:weight} satisfy
    \begin{enumerate}
        \item[(\mylabel{A3.1'}{A3.1'})] $\mathcal{R}_{s^*} := \underset{j\in[n]}{\max}\,\E_P\left[\left(c^*\hat{s}^{(1)}(X_{j})-s^*(X_{j})\right)^2\,\Big|\,\hat{s}^{(1)}\right] = o_{\mathcal{P}}(1)$ \\ where $c^* := \underset{c>0}{\normalfont\arginf}\,\underset{j\in[n]}{\max}\,\E_P\left[\left(c\hat{s}^{(1)}(X_{j})-s^*(X_{j})\right)^2\right]$,
        \item[(\mylabel{A3.2'}{A3.2'})] $\mathcal{R}_{\theta^*} 
 := \underset{k\in[K]}{\max}\big\|\hat{\theta}^{(k)}-\theta^*\big\|_2 = o_{\mathcal{P}}(1)$,
        \item[(\mylabel{A3.3'}{A3.3'})] For all $k\in[K]$, $\hat{s}^{(k)}$ are bounded above and below uniformly by some positive constants and, for the relevant working correlation:
            \begin{itemize}
                \item {\bf Equicorrelation:} \(\hat{\theta}^{(k)}\) is bounded above and below uniformly by constants in \([0,\infty)\),
                \item {\bf Longitudinal:} \(\hat{\theta}^{(k)}\) is bounded above and below uniformly by constants in \([0,1)\),
                \item {\bf Hierarchical / Nested:} \(\hat{\theta}^{(k)}\) is component-wise bounded above and below uniformly by constants in \([0,\infty)\).
            \end{itemize}
    \end{enumerate}
    Then the results of Theorem~\ref{thm:main} hold in this alternative setup.
\end{corollary}

A proof of this corollary is given in Section~\ref{appsec:(s,theta)-proof}.

\section{Proof of results in Section~\ref{sec:theory}}\label{appsec:apptheory}
In this section we prove a stronger version of Theorem~\ref{thm:main}, extending it to non-identically grouped data of varying group size, and allowing group sizes to diverge at sufficiently slow rates with $N$. To allow for non-identical groups we work on a sequence of families of probability distributions $(\mathcal{P}_I)_{I\in\mathbb{N}}$ that the grouped observations follow (satisfying Assumption~\ref{ass:mathcal(P)} with us suppressing dependence on $i$), and re-define the nuisance function error rates given in Assumptions~\ref{ass:nuisance_rates} and~\ref{ass:weight} with
\begin{equation*}
    \mathcal{R}_l := \max_{k\in[K]}\max_{i\in\mathcal{I}_k}\max_{j\in[n_i]}\,\E_P\left[\left(\hat{l}^{(k)}(X_{ij})-l_0(X_{ij})\right)^2\,\bigg|\,\hat{l}^{(k)}\right],
\end{equation*}
and similarly for $\mathcal{R}_m, \mathcal{R}_{\sigma}$ and $\mathcal{R}_{\rho}$. This notation then recovers that in Assumption~\ref{ass:nuisance_rates} when the groups are identically distributed.
We also introduce an additional parameter $\tau$ of the working correlation structure and re-define
\[
 V:=\left(\frac{1}{N}\sum_{i=1}^I\E_{P}\left[\xi_i^{\top}W_i^*\xi_i\right]\right)^{-2}\left(\frac{1}{N}\sum_{i=1}^I\E_{P}\left[\left(\varepsilon_i^{\top}W_i^*\xi_i\right)^2\right]\right),
\]
with $W_i^* := W(\sigma^*, \rho^*)(X_i)$.

\begin{assumptionext}[Additional assumption on weight function]
Additionally to the Assumptions (\ref{A3.1})-(\ref{A3.4}) of Assumption~\ref{ass:weight} suppose we also make the assumption
    \begin{itemize}
        \item[(\mylabel{A3.5}{A3.5})] The working correlation function $\rho^*$ satisfies  
        $\underset{k\in[K]}{\max}\underset{i\in\mathcal{I}_k}{\max}\underset{j\in[n_i]}{\max}\sum_{j'\neq j}^{n_i}\rho^*(X_{ij},X_{ij'})\lesssim n_i^\tau $
        for some $\tau\in[0,1]$ almost surely.
    \end{itemize}
\end{assumptionext}

We can now specify the asymptotic regimes on the maximum group size $n_{\max} := \max_{i\in[I]}n_i$ for which Theorem~\ref{thm:main} holds, in terms of the nuisance error rates $(\mathcal{R}_m,\mathcal{R}_l,\mathcal{R}_{\sigma},\mathcal{R}_{\rho})$, the working correlation metadata $(\alpha,\kappa,\gamma,\tau)$, and the true data generation mechanism metadata $(\delta,\mu_\Sigma,\mu_\Omega)$.

\begin{assumption}[Maximum Group Size]\label{ass:n_max}
The maximum group size \(n_{\max} := \max_{i\in[I]}n_i\) of the grouped observations \(\left\{S_i=(Y_i,D_i,X_i)\right\}_{i\in[I]}\) satisfies
\begin{multline*}\label{eq:n_max}
n_{\max}=o\Bigg(N^{\frac{\delta}{4(1+\kappa)+(2+\kappa)\delta}\wedge\frac{1}{3+2\kappa}}
\wedge \mathcal{R}_m^{-\left(\frac{1}{\kappa+\mu_{\Omega}}\wedge\frac{1}{2(\kappa+\alpha)}\wedge\frac{1}{\gamma}\right)}\wedge\left(N\mathcal{R}_m^{-1}\right)^{\frac{1}{2\gamma+\mu_{\Omega}}}\wedge\mathcal{R}_l^{-\frac{1}{\kappa+\mu_{\Omega}}} \\
\wedge \mathcal{R}_{\sigma}^{-\left(\frac{1}{2\kappa+2\mu_{\Sigma}+\tau+1}\wedge\frac{1}{2\gamma+\tau+1}\right)} \wedge \mathcal{R}_{\rho}^{-\left(\frac{1}{2(\kappa+\mu_{\Sigma}+1)}\wedge\frac{1}{2(\gamma+1)}\right)} \wedge \left(N^{-1}\mathcal{R}_m^{-1}\left(\mathcal{R}_l\vee\mathcal{R}_m\right)^{-1}\right)^{\frac{1}{\kappa}}
\Bigg).
\end{multline*}
\end{assumption}

This assumption reduces to the rates given in Section~\ref{sec:theory} by setting $(\mu_\Sigma,\mu_\Omega,\gamma,\tau)=(1,1,1,1)$ in the equicorrelated working correlation setting, and $(\mu_\Sigma,\mu_\Omega,\alpha,\kappa,\gamma,\tau)=(0,1,0,0,0,0)$ in the autoregressive $\text{AR}(1)$ working correlation setting.

We introduce the following additional notation relevant for the proofs in this section. For a distribution $P$ governing the distribution of a random vector \(A\in\R^d\), we write $\norm{A}_{P,p}:= (\E_P \|A\|_2^p)^{1/p}$, the subscript $P$ in the expectation indicating the dependence on $P$; we will similarly write $\pr_P$. We extend this definition to random matrices $A$ where then $\norm{A}_{P,p}:= (\E_P \|A\|_{\op}^p)^{1/p}$,  where \(\norm{\cdot}_{\op}\) is the operator norm.

\subsection{Proof of Theorem~\ref{thm:main}}

The proof of Theorem~\ref{thm:main} amounts to proving the two claims:
\begin{enumerate}[label=(\roman*)]
    \item $\lim\limits_{I\to\infty}\sup\limits_{P\in\mathcal{P}_I}\sup\limits_{t\in\R}\left|\PP_{P}\left(\sqrt{N/V}\big(\hat{\beta}-\beta\big)\leq t\right)-\Phi\left(t\right)\right|=0;$
    \item $\lim\limits_{I\to\infty}\sup\limits_{P\in\mathcal{P}_I}\sup\limits_{t\in\R}\left|\PP_{P}\left(\sqrt{N/\hat{V}}\big(\hat{\beta}-\beta\big)\leq t\right)-\Phi\left(t\right)\right|=0.$
\end{enumerate}
We will prove each in turn.

\begin{proof}[Proof of Theorem~\ref{thm:main} (i)]
Consider separating the nuisance functions into `fixed effects' functions $\eta:=(l,m)$ and weight functions $\nu:=(\sigma,\rho)$ (analogous to the working covariance in the GEE setup or random effects in the ML setup). Define the terms
\begin{equation*}
\begin{split}
    \varphi\big(S_i;\tilde{\beta},\eta,\nu\big) &:= \big(D_i-m(X_i)\big)^{\top}W_i(\nu)\big(Y_i-l(X_i)-\left(D_i-m(X_i)\right)\tilde{\beta}\big), \\
    \phi\big(S_i;\eta,\nu\big) &:= \big(D_i-m(X_i)\big)^{\top}W_i(\nu)\big(D_i-m(X_i)\big),
\end{split}
\end{equation*}
where $W_i(\nu) := W(\nu)(X_i)$.
Further, recall the notation that we denote the $k^{\text{th}}$ fold estimators for the `fixed effects' $\eta_0:=(l_0,m_0)$ as $\hat{\eta}^{(k)}:=(\hat{l}^{(k)},\hat{m}^{(k)})$, and for the weights the estimators $\hat{\nu}^{(k)}:=(\hat{\sigma}^{(k)},\hat{\rho}^{(k)})$ that converge to a deterministic limit $\nu^*:=(\sigma^*,\rho^*)$, each with corresponding weights $\hat{W}_i=W(\hat{\nu}^{(k)})(X_i)$ (with $k$-dependence implicit through $i$) and $W_i^*=W(\nu^*)(X_i)$. By scaling we may assume without loss of generality that in Assumption~\ref{A3.1} we have $c^*=1$.

Then $\hat{\beta}$ in Algorithm~\ref{alg:main} can be expanded as
\begin{align}\label{eq:betahat_decomp}
    N^{\frac{1}{2}}\big(\hat{\beta}-\beta\big) =& \left(N^{-1}\sum_{k=1}^K\sum_{i\in\mathcal{I}_k}\phi\big(S_i;\hat{\eta}^{(k)},\hat{\nu}^{(k)}\big)\right)^{-1} \left(N^{-\frac{1}{2}}\sum_{k=1}^K\sum_{i\in\mathcal{I}_k}\varphi\big(S_i;\beta,\hat{\eta}^{(k)},\hat{\nu}^{(k)}\big)\right)\nonumber \\
    =& \left(\frac{N^{-1}n_{\max}^{\kappa}\sum_{i=1}^I\E_P\left[\varphi^2(S_i;\beta,\eta_0,\nu^*)\right]}{\left(N^{-1}\sum_{i=1}^I\E_P\left[\phi(S_i;\eta_0,\nu^*)\right]\right)^2}\right)^{\frac{1}{2}}
    \left(\frac{N^{-1}n_{\max}^{\gamma}\sum_{i=1}^I\phi\left(S_i;\eta_0,\nu^*\right) + M_2}{N^{-1}n_{\max}^{\gamma}\sum_{i=1}^I\E_P\left[\phi(S_i;\eta_0,\nu^*)\right]}\right)^{-1} \nonumber\\
        &\qquad\cdot\left(\frac{N^{-\frac{1}{2}}n_{\max}^{\frac{\kappa}{2}}\sum_{i=1}^I\varphi(S_i;\beta,\eta_0,\nu^*)+M_1}{\left(N^{-1}n_{\max}^{\kappa}\sum_{i=1}^I\E_P\left[\varphi^2(S_i;\beta,\eta_0,\nu^*)\right]\right)^{\frac{1}{2}}}\right),
\end{align}
where
\begin{align*}
	M_1 &:= N^{-\frac{1}{2}}n_{\max}^{\frac{\kappa}{2}}\sum_{k=1}^K\sum_{i\in\mathcal{I}_k}\left[\varphi\big(S_i;\beta,\hat{\eta}^{(k)},\hat{\nu}^{(k)}\big)-\varphi\left(S_i;\beta,\eta_0,\nu^*\right)\right], \\
	M_2 &:= N^{-1}n_{\max}^{\gamma}\sum_{k=1}^K\sum_{i\in\mathcal{I}_k}\left[\phi\big(S_i;\hat{\eta}^{(k)},\hat{\nu}^{(k)}\big) - \phi\left(S_i;\eta_0,\nu^*\right) \right] .
\end{align*}

We claim that by Lemma~\ref{lem:unif_clt} we have
{\small
\begin{equation}\label{eq:clt}
    \lim_{I\to\infty}\sup_{P\in\mathcal{P}_I}\sup_{t\in\R}\left|\PP_{P}\left(\left(\sum_{i=1}^I\E_{P}\left[\varphi^2(S_i;\beta,\eta_0,\nu^*)\right]\right)^{-\frac{1}{2}}\left(\sum_{i=1}^I\varphi(S_i;\beta,\eta_0,\nu^*)\right)\leq t\right)-\Phi(t)\right| = 0.
\end{equation}
}%
We must verify the condition for Lemma~\ref{lem:unif_clt} to hold; we will show that
\begin{equation}\label{eq:lyapunov}
    \sup_{P\in\mathcal{P}_I}\frac{\sum_{i=1}^I\E_{P}\left[\left|\varphi(S_i;\beta,\eta_0,\nu^*)\right|^{2+2\Delta}\right]}{\left(\sum_{i=1}^I\E_{P}\left[\varphi^2(S_i;\beta,\eta_0,\nu^*)\right]\right)^{1+\Delta}} = o(1),
\end{equation}
holds for \(\Delta=\frac{\delta}{4}\). First, note that for an arbitrary \(P\in\mathcal{P}_I\),
\begin{align*}
    \sum_{i=1}^I\E_{P}\left[\left|\varphi\left(S_i;\beta,\eta_0,\nu^*\right)\right|^{2+2\Delta}\right]
    =& \sum_{i=1}^I\E_{P}\left[\left|\varepsilon_i^{\top}W_i^*\xi_i\right|^{2+2\Delta}\right]
    \\
    \leq& \sum_{i=1}^I\E_{P}\left[\norm{W_i^*}^{2+2\Delta}_{\op}\norm{\xi_i}_2^{2+2\Delta}\norm{\varepsilon_i}_2^{2+2\Delta}\right] \\
    \lesssim& \sum_{i=1}^I\norm{\xi_i}_{P,4+4\Delta}^{2+2\Delta}\norm{\varepsilon_i}_{P,4+4\Delta}^{2+2\Delta}
    \lesssim \sum_{i=1}^In_i^{2+2\Delta}
    \leq Nn_{\max}^{1+2\Delta} , 
    \\
    \sum_{i=1}^I\E_{P}\left[\varphi^2(S_i;\beta,\eta_0,\nu^*)\right]
    =& \sum_{i=1}^I\E_{P}\left[\tr\left(\Sigma_iW_i^*\xi_i\xi_i^{\top}W_i^*\right)\right]
    \\
    \geq&
    \sum_{i=1}^I\E_{P}\left[\Lambda_{\min}\left(W_i^{*\frac{1}{2}}\Sigma_iW_i^{*\frac{1}{2}}\right)\norm{W_i^*}_{\op}\tr\left(\xi_i\xi_i^{\top}\right)\right] 
    \\
    \gtrsim &
    \sum_{i=1}^In_i^{-\kappa}\E_{P}\left[\tr\left(\xi_i\xi_i^{\top}\right)\right] 
    = \sum_{i=1}^In_i^{-\kappa}\E_{P}\left[\tr\left(\Omega_i\right)\right]
    \\
    \geq &
    \sum_{i=1}^In_i^{1-\kappa}\E_{P}\left[\Lambda_{\min}\left(\Omega_i\right)\right]
    \gtrsim \sum_{i=1}^In_i^{1-\kappa} \geq Nn_{\max}^{-\kappa},
\end{align*}
where the bounding of the numerator follows by the Cauchy--Schwarz inequality and Assumptions~\ref{A1.1}~and~\ref{A3.3}, and the bounding of the denominator follows by Assumptions~\ref{A1.2}~and~\ref{A3.3}. Therefore,
\begin{gather*}
    \sup_{P\in\mathcal{P}_I}\frac{\sum_{i=1}^I\E_{P}\left[\left|\varphi(S_i;\beta,\eta_0,\nu^*)\right|^{2+2\Delta}\right]}{\left(\sum_{i=1}^I\E_{P}\left[\varphi^2(S_i;\beta,\eta_0,\nu^*)\right]\right)^{1+\Delta}}
    \leq N^{-\Delta}n_{\max}^{(1+\kappa)+(2+\kappa)\Delta} = o(1)
\end{gather*}
by Assumption~\ref{ass:n_max}, with $\Delta=\frac{\delta}{4}$. Therefore, the condition~\eqref{eq:lyapunov} holds, and so the claim~\eqref{eq:clt} is proved.

Further, \(M_1=o_{\mathcal{P}}(1)\) by Lemma~\ref{lem:M1} and (as shown already) for arbitrary $P\in\mathcal{P}_I$
\begin{gather*}
    N^{-1}n_{\max}^{\kappa}\sum_{i=1}^I\E_P\left[\varphi^2(S_i;\beta,\eta_0,\nu^*)\right]\gtrsim 1
\end{gather*}
therefore, by Lemma~\ref{lem:unif_slutsky} in conjunction with \eqref{eq:clt} we have that 
\begin{equation}\label{eq:beta_decomp_numerator}
    \lim_{I\to\infty}\sup_{P\in\mathcal{P}_I}\sup_{t\in\R}\left|\PP_P\left(\frac{N^{-1}n_{\max}^{\frac{\kappa}{2}}\sum_{i=1}^I\varphi(S_i;\beta,\eta_0,\nu^*)+M_1}{\left(N^{-1}n_{\max}^{\kappa}\sum_{i=1}^I\E_P\left[\varphi^2(S_i;\beta,\eta_0,\nu^*)\right]\right)^{\frac{1}{2}}}\leq t\right)-\Phi\left(t\right)\right|=0.
\end{equation}

Next, we claim that
\begin{equation}\label{eq:beta_decomp_denominator}
    \frac{N^{-1}n_{\max}^{\gamma}\sum_{i=1}^I\phi\left(S_i;\eta_0,\nu^*\right) + M_2}{N^{-1}n_{\max}^{\gamma}\sum_{i=1}^I\E_P\left[\phi(S_i;\eta_0,\nu^*)\right]} = 1 + o_{\mathcal{P}}(1),
\end{equation}
which we proceed to show. We claim that
\begin{equation*}
  N^{-1}n_{\max}^{\gamma}\sum_{i=1}^I\phi\left(S_i;\eta_0,\nu^*\right) = N^{-1}n_{\max}^{\gamma}\sum_{i=1}^I\E_{P}\left[\phi\left(S_i;\eta_0,\nu^*\right)\right] + o_{\mathcal{P}}(1).
\end{equation*}
This follows by Lemma~\ref{lem:unif_ciiu} alongside
\begin{equation*}
\begin{split}
    &\E_P\left[\bigg(N^{-1}n_{\max}^{\gamma}\sum_{i=1}^I\left(\phi\left(S_i;\eta_0,\nu^*\right) - \E_P\left[\phi\left(S_i;\eta_0,\nu^*\right)\right]\right)\bigg)^2 \right] \\
    \leq& N^{-2}n_{\max}^{2\gamma}\sum_{i=1}^I\E_P\left[\left(\phi\left(S_i;\eta_0,\nu^*\right) - \E_{P}\left[\phi\left(S_i;\eta_0,\nu^*\right)\right]\right)^2\right] \\
    \leq& N^{-2}n_{\max}^{2\gamma}\sum_{i=1}^I\E_P\left[\left(\phi\left(S_i;\eta_0,\nu^*\right)\right)^2\right]
    \leq N^{-2}n_{\max}^{2\gamma}\sum_{i=1}^I\E_{P}\left[\norm{W_i^*}_{\op}^2\norm{\xi_i}_2^4\right]
    \\
    \lesssim& N^{-2}n_{\max}^{2\gamma}\sum_{i=1}^In_i^2
    \leq N^{-1}n_{\max}^{1+2\gamma} = o(1),
\end{split}
\end{equation*}
by the Cauchy--Schwarz inequality and Assumptions~\ref{A1.1},~\ref{A3.3}~and~\ref{ass:n_max}. Because also \(M_2=o_{\mathcal{P}}(1)\) by Lemma~\ref{lem:M2}, as well as that
\begin{equation*}
\begin{split}
    N^{-1}n_{\max}^{\gamma}\sum_{i=1}^I\E_{P}\left[\phi(S_i;\eta_0,\nu^*)\right] 
    &= N^{-1}n_{\max}^{\gamma}\sum_{i=1}^I\E_{P}\left[\tr\left(W_i^*\xi_i\xi_i^{\top}\right)\right]
    \\
    &\geq N^{-1}n_{\max}^{\gamma}\sum_{i=1}^I\E_{P}\left[\Lambda_{\min}\left(W_i^*\right)\tr\left(\xi_i\xi_i^{\top}\right)\right]
    \\
    &\gtrsim N^{-1}n_{\max}^{\gamma}\sum_{i=1}^In_i^{-\gamma}\E_{P}\left[\tr\left(\Omega_i\right)\right]
    \\
    &\gtrsim N^{-1}n_{\max}^{\gamma}\sum_{i=1}^In_i^{1-\gamma}\geq 1 ,
\end{split}
\end{equation*}
we have, by Lemma~\ref{lem:unif_add}, proved the claim~\eqref{eq:beta_decomp_denominator}.

Finally, combining \eqref{eq:beta_decomp_numerator}~with~\eqref{eq:beta_decomp_denominator} with Lemma~\ref{lem:unif_slutsky} completes the proof;
\begin{equation*}
    \lim_{I\to\infty}\sup_{P\in\mathcal{P}_I}\sup_{t\in\R}\left|\PP_P\left(N^{\frac{1}{2}}{V}^{-\frac{1}{2}}\big(\hat{\beta}-\beta\big)\leq t\right)-\Phi\left(t\right)\right|=0.
\end{equation*}
\end{proof}

\begin{proof}[Proof of Theorem~\ref{thm:main} (ii)]
Note that
\begin{equation*}
    \hat{V} = \bigg(N^{-1}\sum_{k=1}^K\sum_{i\in\mathcal{I}_k}\phi\big(S_i;\hat{\eta}^{(k)},\hat{\nu}^{(k)}\big)\bigg)^{-2}\bigg(N^{-1}\sum_{k=1}^K\sum_{i\in\mathcal{I}_k}\varphi^2\big(S_i;\hat{\beta},\hat{\eta}^{(k)},\hat{\nu}^{(k)}\big)\bigg),
\end{equation*}
and so we can decompose
\begin{align}\label{eq:CI_beta_decomp}
    N^{\frac{1}{2}}\hat{V}^{-\frac{1}{2}}\big(\hat{\beta}-\beta\big) =& \left(\frac{N^{-1}\sum_{k=1}^K\sum_{i\in\mathcal{I}_k}\varphi^2\big(S_i;\hat{\beta},\hat{\eta}^{(k)},\hat{\nu}^{(k)}\big)}{N^{-1}\sum_{i=1}^I\E_{P}\left[\varphi^2(S_i;\beta,\eta_0,\nu^*)\right]}\right)^{-\frac{1}{2}}\cdot  \nonumber
    \\
    &\qquad\cdot\left(\sum_{i=1}^I\E_{P}\left[\varphi^2(S_i;\beta,\eta_0,\nu^*)\right]\right)^{-\frac{1}{2}}\left(\sum_{k=1}^K\sum_{i\in\mathcal{I}_k}\varphi\big(S_i;\beta,\hat{\eta}^{(k)},\hat{\nu}^{(k)}\big)\right)
    \\
    =& \left(\frac{N^{-1}n_{\max}^{\kappa}\sum_{i=1}^I\E_P\left[\varphi^2(S_i;\beta,\eta_0,\nu^*)\right] + M_3 + M_4}{N^{-1}n_{\max}^{\kappa}\sum_{i=1}^I\E_{P}\left[\varphi^2(S_i;\beta,\eta_0,\nu^*)\right]}\right)^{-\frac{1}{2}} \cdot
    \\
    &\qquad\cdot
    \left(\frac{N^{-\frac{1}{2}}n_{\max}^{\frac{\kappa}{2}}\sum_{i=1}^I\varphi(S_i;\beta,\eta_0,\nu^*)+M_1}{\left(N^{-1}n_{\max}^{\kappa}\sum_{i=1}^I\E_P\left[\varphi^2(S_i;\beta,\eta_0,\nu^*)\right]\right)^{\frac{1}{2}}}\right),
\end{align}
where
\begin{equation*}
    \begin{split}
    M_3 & := N^{-1}n_{\max}^{\kappa}\sum_{k=1}^K\sum_{i\in\mathcal{I}_k}\left[\varphi^2\big(S_i;\hat{\beta},\hat{\eta}^{(k)},\hat{\nu}^{(k)}\big)-\varphi^2\left(S_i;\beta,\eta_0,\nu^*\right)\right], \\
    M_4 & := N^{-1}n_{\max}^{\kappa}\sum_{i=1}^I\left[\varphi^2(S_i;\beta,\eta_0,\nu^*)-\E_{P}\left[\varphi^2(S_i;\beta,\eta_0,\nu^*)\right]\right].
    \end{split}
\end{equation*}

Note that the latter term in this decomposition is considered in the proof of Theorem~\ref{thm:main}; specifically the result~\eqref{eq:beta_decomp_numerator}. Therefore it remains for us to show that
\begin{equation}\label{eq:beta_CI_decomp_denominator}
    \frac{N^{-1}n_{\max}^{\kappa}\sum_{i=1}^I\E_P\left[\varphi^2(S_i;\beta,\eta_0,\nu^*)\right] + M_3 + M_4}{N^{-1}n_{\max}^{\kappa}\sum_{i=1}^I\E_{P}\left[\varphi^2(S_i;\beta,\eta_0,\nu^*)\right]} = 1 + o_{\mathcal{P}}(1).
\end{equation}
From Lemmas~\ref{lem:M3}~and~\ref{lem:M4} we have that \(M_3\) and \(M_4\) are both \(o_{\mathcal{P}}(1)\). Further (as shown in the proof of Theorem~\ref{thm:main}),
\begin{gather*}
    N^{-1}n_{\max}^{\kappa}\sum_{i=1}^I\E_P\left[\varphi^2(S_i;\beta,\eta_0,\nu^*)\right]\gtrsim 1,
\end{gather*}
and so combining these results with Lemmas~\ref{lem:unif_add}~and~\ref{lem:unif_times} yields the result~\eqref{eq:beta_CI_decomp_denominator}. Finally, the results~\eqref{eq:beta_decomp_numerator}~and~\eqref{eq:beta_CI_decomp_denominator} in conjunction with Lemma~\ref{lem:unif_slutsky} completes the proof.

\end{proof}

\subsection{Auxiliary results for the proof of Theorem~\ref{thm:main}: bounding terms}

\begin{lemma}\label{lem:hatW-W}
Under the setup and notation of Theorem~\ref{thm:main}, we have that
\begin{equation*}
    \big\|\hat{W}_i-W_i^*\big\|_{P|S_{\mathcal{I}_k^c},2} \lesssim n_i^{\frac{1+\tau}{2}}\mathcal{R}_{\sigma}^{\frac{1}{2}}+n_i\mathcal{R}_{\rho}^{\frac{1}{2}}.
\end{equation*}
\end{lemma}

\begin{proof}
We define
\begin{equation*}
    \Sigma_i(\sigma,\rho) := \Big(\big(\ind_{\{j=k\}} + \ind_{\{j\neq k\}}\rho(X_{ij},X_{ik})\big)\sigma(X_{ij})\sigma(X_{ik})\Big)_{(j,k)\in[n_i]^2},
\end{equation*}
and further $\Sigma_i^* := \Sigma_i(\sigma^*,\rho^*)$ and $\hat{\Sigma}_i := \Sigma_i(\hat{\sigma}^{(k)},\hat{\rho}^{(k)})$ (suppressing dependence on $k\in[K]$).

Note that, as \(\hat{W}_i-W_i^*\asymp \hat{\Sigma}^{-1}_i-\Sigma^{*^{-1}}_i\),
\begin{equation*}
    \big\|\hat{W}_i-W_i^*\big\|_{\op} \lesssim \big\|\hat{\Sigma}^{-1}_i\big(\Sigma^*_i-\hat{\Sigma}_i\big)\Sigma^{*^{-1}}_i\big\|_{\op} \leq \big\|\Sigma^{*^{-1}}_i\big\|_{\op}\big\|\hat{\Sigma}^{-1}_i\big\|_{\op}\big\|\hat{\Sigma}_i-\Sigma^*_i\big\|_{\op} \lesssim \big\|\hat{\Sigma}_i-\Sigma^*_i\big\|_{\op},
\end{equation*}
and further,
\begin{align*}
    \big\|\hat{\Sigma}_i-\Sigma^*_i\big\|_{\op}^2 &\leq \big\|\hat{\Sigma}_i-\Sigma^*_i\big\|_{F}^2 \\
        &= \sum_{j=1}^{n_i}\Big(\hat{\sigma}^{(k)}(X_{ij})^2-\sigma^*(X_{ij})^2\Big)^2 \\
        &\quad + \sum_{j=1}^{n_i}\sum_{j'=1,j'\neq j}^{n_i}\Big(\hat{\rho}^{(k)}(X_{ij},X_{ij'})\hat{\sigma}^{(k)}(X_{ij})\hat{\sigma}^{(k)}(X_{ij'})-\rho^*(X_{ij},X_{ij'})\sigma^*(X_{ij})\sigma^*(X_{ij'})\Big)^2 \\
        &\lesssim \sum_{j=1}^{n_i}\Big(\hat{\sigma}^{(k)}(X_{ij})-\sigma^*(X_{ij})\Big)^2\bigg(1+\sum_{j'=1, j'\neq j}^{n_i}\hat{\rho}^{(k)}(X_{ij},X_{ij'})^2\bigg)
        \\
        &\quad + \sum_{j=1}^{n_i}\sum_{j'=1, j'\neq j}^{n_i}\big(\hat{\rho}^{(k)}(X_{ij},X_{ij'})-\rho^*(X_{ij},X_{ij'})\big)^2,
\end{align*}
where $\norm{\cdot}_F$ denotes the Frobenius norm, and so
\begin{align*}
    \E_{P|S_{\mathcal{I}_k^c}}\left[\big\|\hat{\Sigma}_i-\Sigma^*_i\big\|_{\op}^2\right] &\lesssim
    n_i^{1+\tau}\mathcal{R}_{\sigma} + n_i^2\mathcal{R}_{\rho},
\end{align*}
and hence
\begin{equation*}
    \big\|\hat{W}_i-W_i^*\big\|_{P|S_{\mathcal{I}_k^c},2} = 
    \left(\E_{P|S_{\mathcal{I}_k^c}}\left[\big\|\hat{W}_i-W_i^*\big\|_{\op}^2\right]\right)^{\frac{1}{2}} \lesssim n_i^{\frac{1+\tau}{2}}\mathcal{R}_{\sigma}^{\frac{1}{2}}+n_i\mathcal{R}_{\rho}^{\frac{1}{2}}.
\end{equation*}
completing the proof.
\end{proof}

\begin{lemma}\label{lem:M1}
Under the setup of Theorem~\ref{thm:main},
\begin{equation*}
    M_1 := N^{-\frac{1}{2}}n_{\max}^{\frac{\kappa}{2}}\sum_{k=1}^K\sum_{i\in\mathcal{I}_k}\left[\varphi\big(S_i;\beta,\hat{\eta}^{(k)},\hat{\nu}^{(k)}\big)-\varphi\left(S_i;\beta,\eta_0,\nu^*\right)\right] = o_{\mathcal{P}}(1).
\end{equation*}
\end{lemma}
\begin{proof}
For \(k\in[K]\) and \(i\in\mathcal{I}_k\), define the \((\hat{l}^{(k)},\hat{m}^{(k)})\) errors
    \begin{equation*}
    \begin{gathered}
        \mathcal{E}_m^{(k)}(X_i) := \hat{m}^{(k)}\left(X_i\right)-m\left(X_i\right), \qquad \mathcal{E}_l^{(k)}(X_i) := \hat{l}^{(k)}\left(X_i\right)-l\left(X_i\right),
        \\
        \mathcal{E}_{l-\beta m}^{(k)}(X_i) := \mathcal{E}_l^{(k)}(X_i) - \beta\mathcal{E}_m^{(k)}(X_i).
    \end{gathered}
    \end{equation*}
The term \(M_1\) can then be decomposed as \(M_1=\sum_{k=1}^KM_{1,k}\), where
\begin{multline*}
    M_{1,k} := \underbrace{N^{-\frac{1}{2}}n_{\max}^{\frac{\kappa}{2}}\sum_{i\in\mathcal{I}_k}\mathcal{E}_{l-\beta m}^{(k)}(X_i)^{\top}\hat{W}_i\mathcal{E}_m^{(k)}(X_i)}_{\RN{1}}
    + \underbrace{N^{-\frac{1}{2}}n_{\max}^{\frac{\kappa}{2}}\sum_{i\in\mathcal{I}_k}\mathcal{E}_{l-\beta m}^{(k)}(X_i)^{\top}\hat{W}_i\xi_i}_{\RN{2}} \\
    + \underbrace{N^{-\frac{1}{2}}n_{\max}^{\frac{\kappa}{2}}\sum_{i\in\mathcal{I}_k}\mathcal{E}_m^{(k)}(X_i)^{\top}W_i^*\varepsilon_i}_{\RN{3}}
    + \underbrace{N^{-\frac{1}{2}}n_{\max}^{\frac{\kappa}{2}}\sum_{i\in\mathcal{I}_k}\left(\xi_i+\mathcal{E}_m^{(k)}(X_i)\right)^{\top}\left(\hat{W}_i-W_i^*\right)\varepsilon_i}_{\RN{4}}.
\end{multline*}
We show that, for an arbitrary \(k\in[K]\), each of these four terms are \(o_{\mathcal{P}}(1)\) in turn:

    \paragraph{Term \RN{1}:} Note that for an arbitrary \(P\in\mathcal{P}_I\),
        \begin{align*}
            &\E_{P|S_{\mathcal{I}_k^c}}\left[\bigg|N^{-\frac{1}{2}}n_{\max}^{\frac{\kappa}{2}}\sum_{i\in\mathcal{I}_k}\mathcal{E}_{l-\beta m}^{(k)}(X_i)^{\top}\hat{W}_i\mathcal{E}_m^{(k)}(X_i)\bigg|\right]
            \\
            \leq& N^{-\frac{1}{2}}n_{\max}^{\frac{\kappa}{2}}\E_{P|S_{\mathcal{I}_k^c}}\left[\sum_{i\in\mathcal{I}_k}\left|\mathcal{E}_{l-\beta m}^{(k)}(X_i)^{\top}\hat{W}_i\mathcal{E}_m^{(k)}(X_i)\right|\right]
            \\
            \leq& N^{-\frac{1}{2}}n_{\max}^{\frac{\kappa}{2}}\sum_{i\in\mathcal{I}_k}\E_{P|S_{\mathcal{I}_k^c}}\left[\mathcal{E}_{l-\beta m}^{(k)}(X_i)^{\top}\hat{W}_i\mathcal{E}_{l-\beta m}^{(k)}(X_i)\right]^{\frac{1}{2}}\E_{P|S_{\mathcal{I}_k^c}}\left[\mathcal{E}_m^{(k)}(X_i)^{\top}\hat{W}_i\mathcal{E}_m^{(k)}(X_i)\right]^{\frac{1}{2}}
            \\
            \lesssim& N^{-\frac{1}{2}}n_{\max}^{\frac{\kappa}{2}}\sum_{i\in\mathcal{I}_k}\big\|\mathcal{E}_{l-\beta m}^{(k)}(X_i)\big\|_{P|S_{\mathcal{I}_k^c},2}\big\|\mathcal{E}_m^{(k)}(X_i)\big\|_{P|S_{\mathcal{I}_k^c},2}
            \\
            \leq& N^{-\frac{1}{2}}n_{\max}^{\frac{\kappa}{2}}\sum_{i\in\mathcal{I}_k}n_i\mathcal{R}_{m}^{\frac{1}{2}}\left(\mathcal{R}_{l}+\mathcal{R}_{m}\right)^{\frac{1}{2}} = N^{\frac{1}{2}}n_{\max}^{\frac{\kappa}{2}}\mathcal{R}_{m}^{\frac{1}{2}}\left(\mathcal{R}_{l}+\mathcal{R}_m\right)^{\frac{1}{2}} = o_{\mathcal{P}}(1),
        \end{align*}
    by Assumption~\ref{ass:n_max}. The result that this term \(\RN{1}\) is \(o_{\mathcal{P}}(1)\) follows directly by Lemma~\ref{lem:unif_ciiu}.
    
    \paragraph{Term \RN{2}:} First note that \(\E_{P|S_{\mathcal{I}_k^c}}\left[\mathcal{E}_{l-\beta m}^{(k)}(X_i)^{\top}\hat{W}_i\xi_i\right]=0\) for any \(P\in\mathcal{P}_I\) and \(i\in\mathcal{I}_k\). Then we have
    \begin{align*}
        &\E_{P|S_{\mathcal{I}_k^c}}\left[\bigg(N^{-\frac{1}{2}}n_{\max}^{\frac{\kappa}{2}}\sum_{i\in\mathcal{I}_k}\mathcal{E}_{l-\beta m}^{(k)}(X_i)^{\top}\hat{W}_i\xi_i\bigg)^2\right]\\
        =& N^{-1}n_{\max}^{\kappa}\sum_{i\in\mathcal{I}_k}\E_{P|S_{\mathcal{I}_k^c}}\left[\left(\mathcal{E}_{l-\beta m}^{(k)}(X_i)^{\top}\hat{W}_i\xi_i\right)^2\right]\\
        =& N^{-1}n_{\max}^{\kappa}\sum_{i\in\mathcal{I}_k}\E_{P|S_{\mathcal{I}_k^c}}\left[\mathcal{E}_{l-\beta m}^{(k)}(X_i)^{\top}\hat{W}_i\Omega_i\hat{W}_i\mathcal{E}_{l-\beta m}^{(k)}(X_i)\right] \\
        \leq& N^{-1}n_{\max}^{\kappa}\sum_{i\in\mathcal{I}_k}\E_{P|S_{\mathcal{I}_k^c}}\left[\Lambda_{\max}(\Omega_i)\norm{\hat{W}_i}_{\op}^2\big\|\mathcal{E}_{l-\beta m}^{(k)}(X_i)\big\|_{P|S_{\mathcal{I}_k^c},2}^2\right] \\
        \lesssim& N^{-1}n_{\max}^{\kappa+\mu_{\Omega}}\left(\mathcal{R}_{l}+\mathcal{R}_{m}\right)
        = o_{\mathcal{P}}(1),
    \end{align*}
    by Assumption~\ref{ass:n_max}. The result that this term \(\RN{2}\) is \(o_{\mathcal{P}}(1)\) follows directly by Lemma~\ref{lem:unif_ciiu}.
    
    \paragraph{Term \RN{3}:} First note that \(\E_{P|S_{\mathcal{I}_k^c}}\left[\mathcal{E}^{m}_k(X_i)^{\top}\hat{W}_i\varepsilon_i\right]=0\) for any \(P\in\mathcal{P}_I\) and \(i\in\mathcal{I}_k\). Then
    \begin{align*}
        &\E_{P|S_{\mathcal{I}_k^c}}\left[\bigg(N^{-\frac{1}{2}}n_{\max}^{\frac{\kappa}{2}}\sum_{i\in\mathcal{I}_k}\mathcal{E}^{m}_k(X_i)^{\top}W_i^*\varepsilon_i\bigg)^2\right]
        \\
        =& N^{-1}n_{\max}^{\kappa}\sum_{i\in\mathcal{I}_k}\E_{P|S_{\mathcal{I}_k^c}}\left[\left(\mathcal{E}^{m}_k(X_i)^{\top}W_i^*\varepsilon_i\right)^2\right]
        \\
        =& N^{-1}n_{\max}^{\kappa}\sum_{i\in\mathcal{I}_k}\E_{P|S_{\mathcal{I}_k^c}}\left[\mathcal{E}_m^{(k)}(X_i)^{\top}W_i^*\Sigma_iW_i^*\mathcal{E}_m^{(k)}(X_i)\right] 
        \\
        \leq& N^{-1}n_{\max}^{\kappa}\sum_{i\in\mathcal{I}_k}\E_{P|S_{\mathcal{I}_k^c}}\left[\Lambda_{\max}(W_i^{*\frac{1}{2}}\Sigma_iW_i^{*\frac{1}{2}})\norm{W_i^*}_{\op}\big\|\mathcal{E}_m^{(k)}(X_i)\big\|_2^2\right] 
        \\
        \lesssim& N^{-1}n_{\max}^{\kappa}\sum_{i\in\mathcal{I}_k}n_i^{1+\alpha}\mathcal{R}_{m}
        \leq n_{\max}^{\kappa+\alpha}\mathcal{R}_{m}
        = o_{\mathcal{P}}(1),
    \end{align*}
    by Assumption~\ref{ass:n_max}. The result that this term \(\RN{2}\) is \(o_{\mathcal{P}}(1)\) follows directly by Lemma~\ref{lem:unif_ciiu}.
    
    \paragraph{Term \RN{4}:} Similarly, note that \(\E_{P|S_{\mathcal{I}_k^c}}\Big[\Big(\xi_i+\mathcal{E}_m^{(k)}(X_i)\Big)^{\top}\big(\hat{W}_i-W_i^*\big)\varepsilon_i\Big]=0\) for any $P\in\mathcal{P}_I$ and $i\in\mathcal{I}_k$. Then
    \begin{align*}
        &\E_{P|S_{\mathcal{I}_k^c}}\bigg[\Big(\Big(\xi_i+\mathcal{E}_m^{(k)}(X_i)\Big)^{\top}\big(\hat{W}_i-W_i^*\big)\varepsilon_i\Big)^2\bigg]
        \\
        =& \E_{P|S_{\mathcal{I}_k^c}}\left[\left(\xi_i+\mathcal{E}_m^{(k)}(X_i)\right)^{\top}\big(\hat{W}_i-W_i^*\big)\Sigma_i\big(\hat{W}_i-W_i^*\big)\left(\xi_i+\mathcal{E}_m^{(k)}(X_i)\right)\right]
        \\
        \leq& \E_{P|S_{\mathcal{I}_k^c}}\left[\Lambda_{\max}(\Sigma_i)\big\|\hat{W}_i-W_i^*\big\|_{\op}^2\left(\norm{\xi_i}_2^2+\big\|\mathcal{E}_m^{(k)}(X_i)\big\|_2^2\right)\right]
        \\
        \leq& \E_{P|S_{\mathcal{I}_k^c}}\left[\Lambda_{\max}(\Sigma_i)\big\|\hat{W}_i-W_i^*\big\|_{\op}\left(\norm{\xi_i}_2^2+\big\|\mathcal{E}_m^{(k)}(X_i)\big\|_2^2\right)\right]
        \\
        \lesssim& n_i^{\mu_{\Sigma}}\big\|\hat{W}_i-W_i^*\big\|_{P|S_{\mathcal{I}_k^c},2}\Big(\norm{\xi_i}_{P|S_{\mathcal{I}_k^c},4}^2+\big\|\mathcal{E}_m^{(k)}(X_i)\big\|_{P|S_{\mathcal{I}_k^c},4}^2\Big)
        \\
        \lesssim& n_i^{1+\mu_{\Sigma}}\Big(n_i^{\frac{1+\tau}{2}}\mathcal{R}_{\sigma}^{\frac{1}{2}}+n_i\mathcal{R}_{\rho}^{\frac{1}{2}}\Big)
    \end{align*}
    with the second inequality holding because \(\big\|\hat{W}_i-W_i^*\big\|_{\op}\leq\big\|\hat{W}_i\big\|_{\op}+\big\|W_i^*\big\|_{\op}\lesssim 1\) almost surely, and the fourth inequality obtained by invoking Lemma~\ref{lem:hatW-W}. Therefore
    \begin{align*}
        &\E_{P|S_{\mathcal{I}_k^c}}\bigg[\Big(N^{-\frac{1}{2}}n_{\max}^{\frac{\kappa}{2}}\sum_{i\in\mathcal{I}_k}\left(\xi_i+\mathcal{E}_m^{(k)}(X_i)\right)^{\top}\big(\hat{W}_i-W_i^*\big)\varepsilon_i\Big)^2\bigg]
        \\
        \leq& N^{-1}n_{\max}^{\kappa}\sum_{i\in\mathcal{I}_k}\E_{P|S_{\mathcal{I}_k^c}}\bigg[\Big(\left(\xi_i+\mathcal{E}_m^{(k)}(X_i)\right)^{\top}\big(\hat{W}_i-W_i^*\big)\varepsilon_i\Big)^2\bigg]
        \\
        \lesssim& n_{\max}^{\kappa+\mu_{\Sigma}}\Big(n_{\max}^{\frac{1+\tau}{2}}\mathcal{R}_{\sigma}^{\frac{1}{2}}+n_i\mathcal{R}_{\rho}^{\frac{1}{2}}\Big)
        = o_{\mathcal{P}}(1),
    \end{align*}
    by Assumption~\ref{ass:n_max}. The result that this term \(\RN{4}\) is \(o_{\mathcal{P}}(1)\) follows directly by Lemma~\ref{lem:unif_ciiu}.

\medskip\noindent
All four terms are therefore \(o_{\mathcal{P}}(1)\) and so applying Lemma~\ref{lem:unif_add} completes the proof.
\end{proof}

\begin{lemma}\label{lem:M2}
Under the setup of Theorem~\ref{thm:main},
\begin{equation*}
    M_2 := N^{-1}n_{\max}^{\gamma}\sum_{k=1}^K\sum_{i\in\mathcal{I}_k}\left[\phi\big(S_i;\hat{\eta}^{(k)},\hat{\nu}^{(k)}\big) - \phi\left(S_i;\eta_0,\nu^*\right) \right] = o_{\mathcal{P}}(1).
\end{equation*}
\end{lemma}
\begin{proof}
We adopt the same notation as in Lemma~\ref{lem:M1}. Then we can decompose \(M_2\) into \(M_2=\sum_{k=1}^KM_{2,k}\), where
\begin{multline*}
    M_{2,k} := \underbrace{N^{-1}n_{\max}^{\gamma}\sum_{i\in\mathcal{I}_k}\mathcal{E}_m^{(k)}(X_i)\hat{W}_i\mathcal{E}_m^{(k)}(X_i)}_{\RN{5}} - 2\underbrace{N^{-1}n_{\max}^{\gamma}\sum_{i\in\mathcal{I}_k}\xi_i^{\top}\hat{W}_i\mathcal{E}_m^{(k)}(X_i)}_{\RN{6}} 
    \\
    + \underbrace{N^{-1}n_{\max}^{\gamma}\sum_{i\in\mathcal{I}_k}\xi_i^{\top}\left(\hat{W}_i-W_i^*\right)\xi_i}_{\RN{7}}.
\end{multline*}
For an arbitrary \(k\in[K]\) we will show that each of these three terms are each \(o_{\mathcal{P}}(1)\):

\paragraph{Term \RN{5}:} We have that
\begin{align*}
    \E_{P|S_{\mathcal{I}_k^c}}\bigg[N^{-1}n_{\max}^{\gamma}\sum_{i\in\mathcal{I}_k}\mathcal{E}_m^{(k)}(X_i)^{\top}\hat{W}_i\mathcal{E}_m^{(k)}(X_i)\bigg] 
    \lesssim& N^{-1}n_{\max}^{\gamma}\sum_{i\in\mathcal{I}_k^c}\big\|\mathcal{E}_m^{(k)}(X_i)\big\|_{P|S_{\mathcal{I}_k^c},2}^2 
    \\
    \lesssim& n_{\max}^{\gamma}\mathcal{R}_{m}=o_{\mathcal{P}}(1),
\end{align*}
and so (by Lemma~\ref{lem:unif_ciiu}) this term is \(o_{\mathcal{P}}(1)\).

\paragraph{Term \RN{6}:} 
\begin{align*}
    &\E_{P|S_{\mathcal{I}_k^c}}\bigg[\Big(N^{-1}n_{\max}^{\gamma}\sum_{i\in\mathcal{I}_k}\mathcal{E}_m^{(k)}(X_i)^{\top}\hat{W}_i\xi\Big)^2\bigg]
    \\
    =& N^{-2}n_{\max}^{2\gamma}\sum_{i\in\mathcal{I}_k}\E_{P|S_{\mathcal{I}_k^c}}\left[\mathcal{E}_m^{(k)}(X_i)^{\top}\hat{W}_i\Omega_i\hat{W}_i\mathcal{E}_m^{(k)}(X_i)^{\top}\right]
    \\
    \lesssim& N^{-2}n_{\max}^{2\gamma}\sum_{i\in\mathcal{I}_k}n_i^{1+\mu_{\Omega}}\mathcal{R}_{m} \leq N^{-1}n_{\max}^{2\gamma+\mu_{\Omega}}\mathcal{R}_{m}=o_{\mathcal{P}}(1),
\end{align*}
and so (again by Lemma~\ref{lem:unif_ciiu}) this term is \(o_{\mathcal{P}}(1)\).

\paragraph{Term \RN{7}:} 
\begin{align*}
    &\E_{P|S_{\mathcal{I}_k^c}}\bigg[\Big|N^{-1}n_{\max}^{\gamma}\sum_{i\in\mathcal{I}_k}\xi_i^{\top}(\hat{W}_i-W_i^*)\xi\Big|\bigg] 
    \\
    \leq& N^{-1}n_{\max}^{\gamma}\sum_{i\in\mathcal{I}_k}\E_{P|S_{\mathcal{I}_k^c}}\left[\big\|\hat{W}_i-W_i^*\big\|_{\op}\norm{\xi_i}_2^2\right] 
    \\
    \lesssim& n_{\max}\Big(n_{\max}^{\frac{1+\tau}{2}}\mathcal{R}_{\sigma}^{\frac{1}{2}}+n_{\max}\mathcal{R}_{\rho}^{\frac{1}{2}}\Big)=o_{\mathcal{P}}(1),
\end{align*}
and so similarly this term is \(o_{\mathcal{P}}(1)\).

\medskip\noindent
Combining these results with Lemma~\ref{lem:unif_add} proves the claim.
\end{proof}

\begin{lemma}\label{lem:M3}
Under the setup of Theorem~\ref{thm:main},
\begin{equation*}
    M_3 := N^{-1}n_{\max}^{\kappa}\sum_{k=1}^K\sum_{i\in\mathcal{I}_k}\left[\varphi^2\big(S_i;\hat{\beta},\hat{\eta}^{(k)},\hat{\nu}^{(k)}\big)-\varphi^2\left(S_i;\beta,\eta_0,\nu^*\right)\right] = o_{\mathcal{P}}(1).
\end{equation*}
\end{lemma}
\begin{proof}
First, decompose the \(M_3\) term as \(M_3=\sum_{k=1}^KM_{3,k}\), where
\begin{align*}
    M_{3,k} &:= N^{-1}n_{\max}^{\kappa}\sum_{i\in\mathcal{I}_k}\left[\varphi^2(S_i;\hat{\beta},\hat{\eta}^{(k)},\hat{\nu}^{(k)})-\varphi^2(S_i;\beta,\eta_0,\nu^*)\right] \\
    &= N^{-1}n_{\max}^{\kappa}\sum_{i\in\mathcal{I}_k}\left(\varphi(S_i;\hat{\beta},\hat{\eta}^{(k)},\hat{\nu}^{(k)})-\varphi(S_i;\beta,\eta_0,\nu^*)\right)^2
    \\
    &\quad + N^{-1}n_{\max}^{\kappa}\sum_{i\in\mathcal{I}_k}\varphi(S_i;\beta,\eta_0,\nu^*)\left(\varphi(S_i;\hat{\beta},\hat{\eta}^{(k)},\hat{\nu}^{(k)})-\varphi(S_i;\beta,\eta_0,\nu^*)\right)
    \\
    &\leq 2N^{-1}n_{\max}^{\kappa}\sum_{i\in\mathcal{I}_k}\left(\varphi(S_i;\beta,\hat{\eta}^{(k)},\hat{\nu}^{(k)})-\varphi(S_i;\beta,\eta_0,\nu^*)\right)^2
    \\
    &\quad + 2N^{-1}n_{\max}^{\kappa}\sum_{i\in\mathcal{I}_k}\left(\varphi(S_i;\hat{\beta},\hat{\eta}^{(k)},\hat{\nu}^{(k)})-\varphi(S_i;\beta,\hat{\eta}^{(k)},\hat{\nu}^{(k)})\right)^2
    \\
    &\quad + N^{-1}n_{\max}^{\kappa}\sum_{i\in\mathcal{I}_k}\varphi(S_i;\beta,\eta_0,\nu^*)\left(\varphi(S_i;\beta,\hat{\eta}^{(k)},\hat{\nu}^{(k)})-\varphi(S_i;\beta,\eta_0,\nu^*)\right)
    \\
    &\quad + N^{-1}n_{\max}^{\kappa}\sum_{i\in\mathcal{I}_k}\varphi(S_i;\beta,\eta_0,\nu^*)\left(\varphi(S_i;\hat{\beta},\hat{\eta}^{(k)},\hat{\nu}^{(k)})-\varphi(S_i;\beta,\hat{\eta}^{(k)},\hat{\nu}^{(k)})\right)
    \\
    &= \underbrace{2N^{-1}n_{\max}^{\kappa}\sum_{i\in\mathcal{I}_k}\left(\varphi(S_i;\beta,\hat{\eta}^{(k)},\hat{\nu}^{(k)})-\varphi(S_i;\beta,\eta_0,\nu^*)\right)^2}_{\RN{8}} 
    \\
    &\quad + \underbrace{N^{-1}n_{\max}^{\kappa}\sum_{i\in\mathcal{I}_k}\varphi(S_i;\beta,\eta_0,\nu^*)\left(\varphi(S_i;\beta,\hat{\eta}^{(k)},\hat{\nu}^{(k)})-\varphi(S_i;\beta,\eta_0,\nu^*)\right)}_{\RN{9}} 
    \\
    &\quad + \underbrace{2N^{-1}n_{\max}^{\kappa}\sum_{i\in\mathcal{I}_k}\phi^2(S_i;\hat{\eta}^{(k)},\hat{\nu}^{(k)})\big(\hat{\beta}-\beta\big)^2}_{\RN{10}}
    \\
    &\quad +  \underbrace{N^{-1}n_{\max}^{\kappa}\sum_{i\in\mathcal{I}_k}\varphi(S_i;\beta,\eta_0,\nu^*)\phi(S_i;\hat{\eta}^{(k)},\hat{\nu}^{(k)})\big(\hat{\beta}-\beta\big)}_{\RN{11}}
\end{align*}
We claim that each term is \(o_{\mathcal{P}}(1)\), and prove this for each term separately:

    \paragraph{Term \RN{8}:} We decompose this term into the four subsequent terms
    \begin{align*}
        & 2N^{-1}n_{\max}^{\kappa}\sum_{i\in\mathcal{I}_k}\left(\varphi(S_i;\beta,\hat{\eta}^{(k)},\hat{\nu}^{(k)})-\varphi(S_i;\beta,\eta_0,\nu^*)\right)^2 
        \\
        \lesssim& \underbrace{N^{-1}n_{\max}^{\kappa}\sum_{i\in\mathcal{I}_k}\left|\mathcal{E}_{l-\beta m}^{(k)}(X_i)^{\top}\hat{W}_i\mathcal{E}_m^{(k)}(X_i)\right|^2}_{\RN{8}\text{a}}
        + \underbrace{N^{-1}n_{\max}^{\kappa}\sum_{i\in\mathcal{I}_k}\left(\mathcal{E}_{l-\beta m}^{(k)}(X_i)^{\top}\hat{W}_i\xi_i\right)^2}_{\RN{8}\text{b}} \\
        &\hspace{-0.5em} + \underbrace{N^{-1}n_{\max}^{\kappa}\sum_{i\in\mathcal{I}_k}\left(\mathcal{E}_m^{(k)}(X_i)^{\top}W_i^*\varepsilon_i\right)^2}_{\RN{8}\text{c}}
        + \underbrace{N^{-1}n_{\max}^{\kappa}\sum_{i\in\mathcal{I}_k}\left(\left(\xi_i+\mathcal{E}_m^{(k)}(X_i)\right)^{\top}\big(\hat{W}_i-W_i^*\big)\varepsilon_i\right)^2}_{\RN{8}\text{d}}.
    \end{align*}
    Note that
    \begin{align*}
        &\E_{P|S_{\mathcal{I}_k^c}}\left[\bigg(N^{-1}n_{\max}^{\kappa}\sum_{i\in\mathcal{I}_k}\left|\mathcal{E}_{l-\beta m}^{(k)}(X_i)^{\top}\hat{W}_i\mathcal{E}_m^{(k)}(X_i)\right|^2\bigg)^{\frac{1}{2}}\right] 
        \\
        \leq& N^{-\frac{1}{2}}n_{\max}^{\frac{\kappa}{2}}\sum_{i\in\mathcal{I}_k}\E_{P|S_{\mathcal{I}_k^c}}\left[\left|\mathcal{E}_{l-\beta m}^{(k)}(X_i)^{\top}\hat{W}_i\mathcal{E}_m^{(k)}(X_i)\right|\right] = o_{\mathcal{P}}(1)
    \end{align*}
    as shown when dealing with Term \(\RN{1}\) in the decomposition for Lemma~\ref{lem:M1}, and so the square root of \(\RN{8}\text{a}\) and hence \(\RN{8}\text{a}\) are both \(o_{\mathcal{P}}(1)\).
    Further, the expectation under the law of \(P|S_{\mathcal{I}_k^c}\) of the terms \(\RN{8}\text{b}\), \(\RN{8}\text{c}\) and \(\RN{8}\text{d}\) are each shown in Lemma~\ref{lem:M1} to be \(o_{\mathcal{P}}(1)\), and so the terms \(\RN{8}\text{b}\), \(\RN{8}\text{c}\) and \(\RN{8}\text{d}\) are all \(o_{\mathcal{P}}(1)\), thus (again by Lemma~\ref{lem:unif_ciiu}) the term~\(\RN{5}\) is \(o_{\mathcal{P}}(1)\).
    
    \paragraph{Term \RN{9}:} Note that this term can be further decomposed into the four terms
    \begin{align*}
    &N^{-1}n_{\max}^{\kappa}\sum_{i\in\mathcal{I}_k}\varphi(S_i;\beta,\eta_0,\nu^*)\left(\varphi(S_i;\beta,\hat{\eta}^{(k)},\hat{\nu}^{(k)})-\varphi(S_i;\beta,\eta_0,\nu^*)\right)
    \\
    =& \underbrace{N^{-1}n_{\max}^{\kappa}\sum_{i\in\mathcal{I}_k}\big(\xi_i^{\top}W_i^*\varepsilon_i\big)\left(\mathcal{E}_{l-\beta m}^{(k)}(X_i)^{\top}\hat{W}_i\mathcal{E}_m^{(k)}(X_i)\right)}_{\RN{9}\text{a}}
    \\
    &\quad + \underbrace{N^{-1}n_{\max}^{\kappa}\sum_{i\in\mathcal{I}_k}\big(\xi_i^{\top}W_i^*\varepsilon_i\big)\left(\mathcal{E}_{l-\beta m}^{(k)}(X_i)^{\top}\hat{W}_i^*\xi_i\right)}_{\RN{9}\text{b}} 
    + \underbrace{N^{-1}n_{\max}^{\kappa}\sum_{i\in\mathcal{I}_k}\big(\xi_i^{\top}W_i^*\varepsilon_i\big)\left(\mathcal{E}_m^{(k)}(X_i)^{\top}W_i^*\varepsilon_i\right)}_{\RN{9}\text{c}}
    \\
    &\quad + \underbrace{N^{-1}n_{\max}^{\kappa}\sum_{i\in\mathcal{I}_k}\Big(\big(\xi_i+\mathcal{E}_m^{(k)}(X_i)\big)^{\top}W_i^*\varepsilon_i\Big)\left(\xi_i^{\top}\big(\hat{W}_i-W_i^*\big)\varepsilon_i\right)}_{\RN{9}\text{d}}.
    \end{align*}
    We claim that each of these four terms are \(o_{\mathcal{P}}(1)\); we will show each in turn. For the first term \(\RN{9}\text{a}\) note that \(\E_{P|S_{\mathcal{I}_k^c}}\left[\left(\xi_i^{\top}W_i^*\varepsilon_i\right)\left(\mathcal{E}_{l-\beta m}^{(k)}(X_i)^{\top}\hat{W}_i\mathcal{E}_m^{(k)}(X_i)\right)\right]=0\) for all $P\in\mathcal{P}_I$ and $i\in\mathcal{I}_k$. We consider the cases of \(\delta\in[0,4)\) and \(\delta\geq 4\) separately. If \(\delta\in[0,4)\) then by the von Bahr--Esseen inequality
    \begin{align*}
        &\E_P\bigg[\Big|N^{-1}n_{\max}^{\kappa}\sum_{i\in\mathcal{I}_k}\big(\xi_i^{\top}W_i^*\varepsilon_i\big)\left(\mathcal{E}_{l-\beta m}^{(k)}(X_i)^{\top}\hat{W}_i\mathcal{E}_m^{(k)}(X_i)\right)\Big|^{1+\frac{\delta}{4}}\bigg] \\
        \leq&\left(2-I^{-1}\right)N^{-1-\frac{\delta}{4}}n_{\max}^{\kappa\left(1+\frac{\delta}{4}\right)}\sum_{i\in\mathcal{I}_k}\E_P\left[\left|\big(\xi_i^{\top}W_i\varepsilon_i\big)\left(\mathcal{E}_{l-\beta m}^{(k)}(X_i)^{\top}\hat{W}_i\mathcal{E}_m^{(k)}(X_i)\right)\right|^{1+\frac{\delta}{4}}\right] \\
        \lesssim& N^{-1-\frac{\delta}{4}}n_{\max}^{\kappa\left(1+\frac{\delta}{4}\right)}\sum_{i\in\mathcal{I}_k}\E_P\left[\norm{\xi_i}_2^{1+\frac{\delta}{4}}\norm{\varepsilon_i}_2^{1+\frac{\delta}{4}}\big\|\mathcal{E}_{l-\beta m}^{(k)}(X_i)\big\|_2^{1+\frac{\delta}{4}}\big\|\mathcal{E}_m^{(k)}(X_i)\big\|_2^{1+\frac{\delta}{4}}\right] \\
        \lesssim& N^{-1-\frac{\delta}{4}}n_{\max}^{\kappa\left(1+\frac{\delta}{4}\right)}\sum_{i\in\mathcal{I}_k}\norm{\xi_i}_{P,4+\delta}^{1+\frac{\delta}{4}}\norm{\varepsilon_i}_{P,4+\delta}^{1+\frac{\delta}{4}}\big\|\mathcal{E}_{l-\beta m}^{(k)}(X_i)\big\|_{P|S_{\mathcal{I}_k^c},4+\delta}^{1+\frac{\delta}{4}}\big\|\mathcal{E}_m^{(k)}(X_i)\big\|_{P|S_{\mathcal{I}_k^c},4+\delta}^{1+\frac{\delta}{4}} \\
        \lesssim& N^{-1-\frac{\delta}{4}}n_{\max}^{\kappa\left(1+\frac{\delta}{4}\right)}\sum_{i\in\mathcal{I}_k}n_i^{2+\frac{\delta}{2}} \leq N^{-\frac{\delta}{4}}n_{\max}^{(1+\kappa)+(2+\kappa)\frac{\delta}{4}},
    \end{align*}
    for any $P\in\mathcal{P}_I$, and thus
    \begin{align*}
        \sup_{P\in\mathcal{P}_I}\E_P\bigg[\Big|N^{-1}n_{\max}^{\kappa}\sum_{i\in\mathcal{I}_k}\big(\xi_i^{\top}W_i^*\varepsilon_i\big)\left(\mathcal{E}_{l-\beta m}^{(k)}(X_i)^{\top}\hat{W}_i\mathcal{E}_m^{(k)}(X_i)\right)\Big|^{1+\frac{\delta}{4}}\bigg] &\lesssim N^{-\frac{\delta}{4}}n_{\max}^{(1+\kappa)+(2+\kappa)\frac{\delta}{4}}
        \\
        &= o(1),
    \end{align*}
    by Assumption~\ref{ass:n_max}. 
    If on the other hand \(\delta\geq 4\) then note that
    \begin{align*}
        &\E_P\bigg[\Big(N^{-1}n_{\max}^{\kappa}\sum_{i\in\mathcal{I}_k}\big(\xi_i^{\top}W_i^*\varepsilon_i\big)\left(\mathcal{E}_{l-\beta m}^{(k)}(X_i)^{\top}\hat{W}_i\mathcal{E}_m^{(k)}(X_i)\right)\Big)^2\bigg] \\
        =&N^{-2}n_{\max}^{2\kappa}\sum_{i\in\mathcal{I}_k}\E_P\left[\left|\big(\xi_i^{\top}W_i^*\varepsilon_i\big)\left(\mathcal{E}_{l-\beta m}^{(k)}(X_i)^{\top}\hat{W}_i\mathcal{E}_m^{(k)}(X_i)\right)\right|^2\right] \\
        \lesssim& N^{-2}n_{\max}^{2\kappa}\sum_{i\in\mathcal{I}_k}\E_P\left[\norm{\xi_i}_2^2\norm{\varepsilon_i}_2^2\big\|\mathcal{E}_{l-\beta m}^{(k)}(X_i)\big\|_2^2\big\|\mathcal{E}_m^{(k)}(X_i)\big\|_2^2\right] \\
        \lesssim& N^{-2}n_{\max}^{2\kappa}\sum_{i\in\mathcal{I}_k}\norm{\xi_i}_{P,8}^2\norm{\varepsilon_i}_{P,8}^2\big\|\mathcal{E}_{l-\beta m}^{(k)}(X_i)\big\|_{P|S_{\mathcal{I}_k^c},8}^2\big\|\mathcal{E}_m^{(k)}(X_i)\big\|_{P|S_{\mathcal{I}_k^c},8}^2 \\
        \lesssim& N^{-2}n_{\max}^{2\kappa}\sum_{i\in\mathcal{I}_k}n_i^4 \leq N^{-1}n_{\max}^{3+2\kappa},
    \end{align*}
    for all $P\in\mathcal{P}_I$, and thus
    \begin{equation*}
        \sup_{P\in\mathcal{P}_I}\E_P\bigg[\Big(N^{-1}n_{\max}^{\kappa}\sum_{i\in\mathcal{I}_k}\big(\xi_i^{\top}W_i^*\varepsilon_i\big)\left(\mathcal{E}_{l-\beta m}^{(k)}(X_i)^{\top}\hat{W}_i\mathcal{E}_m^{(k)}(X_i)\right)\Big)^2\bigg]
        \lesssim
        N^{-1}n_{\max}^{3+2\kappa}
        =
        o(1),
    \end{equation*}
    by Assumption~\ref{ass:n_max}. Hence, by Lemma~\ref{lem:unif_ciiu}, the term \(\RN{9}\text{a}\) is \(o_{\mathcal{P}}(1)\).
    
    For the term \(\RN{9}\text{b}\), note that as \(\E_{P|S_{\mathcal{I}_k^c}}\left[\left(\xi_i^{\top}W_i\varepsilon_i\right)\left(\mathcal{E}_{l-\beta m}^{(k)}(X_i)^{\top}\hat{W}_i\xi_i\right)\right]=0\) for any $P\in\mathcal{P}_I$ and $i\in\mathcal{I}_k$, and so the same methodology as for the term \(\RN{9}\text{a}\) can be used for \(\RN{9}\text{b}\), showing that this term is also \(o_{\mathcal{P}}(1)\).
    
    For the term \(\RN{9}\text{c}\), note that
    \begin{equation*}
    \begin{split}
        &\E_{P|S_{\mathcal{I}_k^c}}\bigg[\Big|N^{-1}n_{\max}^{\kappa}\sum_{i\in\mathcal{I}_k}\big(\xi_i^{\top}W_i^*\varepsilon_i\big)\left(\mathcal{E}_m^{(k)}(X_i)^{\top}W_i^*\varepsilon_i\right)\Big|\bigg] \\
        \leq& N^{-1}n_{\max}^{\kappa}\sum_{i\in\mathcal{I}_k}\E_{P|S_{\mathcal{I}_k^c}}\left[\left|\big(\xi_i^{\top}W_i^*\varepsilon_i\big)\left(\mathcal{E}_m^{(k)}(X_i)^{\top}W_i^*\varepsilon_i\right)\right|\right] \\
        \leq&
        N^{-1}n_{\max}^{\kappa}\sum_{i\in\mathcal{I}_k}\E_{P|S_{\mathcal{I}_k^c}}\left[\big(\xi_i^{\top}W_i^*\varepsilon_i\big)^2\right]^{\frac{1}{2}}E_{P|S_{\mathcal{I}_k^c}}\left[\left(\mathcal{E}_m^{(k)}(X_i)^{\top}W_i^*\varepsilon_i\right)^2\right]^{\frac{1}{2}} \\
        =& N^{-1}n_{\max}^{\kappa}\sum_{i\in\mathcal{I}_k}\E_{P|S_{\mathcal{I}_k^c}}\left[\xi_i^{\top}W_i^*\Sigma_iW_i^*\xi_i\right]^{\frac{1}{2}}E_{P|S_{\mathcal{I}_k^c}}\left[\mathcal{E}_m^{(k)}(X_i)^{\top}W_i^*\Sigma_iW_i^*\mathcal{E}_m^{(k)}(X_i)\right]^{\frac{1}{2}} \\
        \leq& N^{-1}n_{\max}^{\kappa}\sum_{i\in\mathcal{I}_k}n_i^{1+\alpha}\mathcal{R}_{m}^{\frac{1}{2}}
        \leq n_{\max}^{\kappa+\alpha}\mathcal{R}_{m}^{\frac{1}{2}}
        = o_{\mathcal{P}}(1),
    \end{split}
    \end{equation*}
    
    Similarly, for the term \(\RN{9}\text{d}\),
    \begin{equation*}
    \begin{split}
        &\E_{P|S_{\mathcal{I}_k^c}}\bigg[\Big|N^{-1}n_{\max}^{\kappa}\sum_{i\in\mathcal{I}_k}\big(\xi_i^{\top}W^*_i\varepsilon_i\big)\Big(\Big(\xi_i+\mathcal{E}_m^{(k)}(X_i)\Big)^{\top}\big(\hat{W}_i-W^*_i\big)\varepsilon_i\Big)\Big|\bigg] \\
        \leq& N^{-1}n_{\max}^{\kappa}\sum_{i\in\mathcal{I}_k}\E_{P|S_{\mathcal{I}_k^c}}\bigg[\Big|\big(\xi_i^{\top}W^*_i\varepsilon_i\big)\Big(\Big(\xi_i+\mathcal{E}_m^{(k)}(X_i)\Big)^{\top}\big(\hat{W}_i-W^*_i\big)\varepsilon_i\Big)\Big|\bigg] \\
        \leq&
        N^{-1}n_{\max}^{\kappa}\sum_{i\in\mathcal{I}_k}\E_{P|S_{\mathcal{I}_k^c}}\Big[\big(\xi_i^{\top}W^*_i\varepsilon_i\big)^2\Big]^{\frac{1}{2}}E_{P|S_{\mathcal{I}_k^c}}\bigg[\Big(\Big(\xi_i+\mathcal{E}_m^{(k)}(X_i)\Big)^{\top}\big(\hat{W}_i-W^*_i\big)\varepsilon_i\Big)^2\bigg]^{\frac{1}{2}} \\
        \lesssim& n_{\max}^{\kappa+\frac{\alpha}{2}+\frac{\mu_{\Sigma}}{2}}\Big(n_{\max}^{\frac{1+\tau}{2}}\mathcal{R}_{\sigma}^{\frac{1}{2}}+n_{\max}\mathcal{R}_{\rho}^{\frac{1}{2}}\Big)
        = o_{\mathcal{P}}(1),
    \end{split}
    \end{equation*}
    by Assumption~\ref{ass:n_max}. Hence, by Lemma~\ref{lem:unif_ciiu}, the term \(\RN{9}\text{d}\) is \(o_{\mathcal{P}}(1)\).
    
    \paragraph{Term \RN{10}:} Note that
    \begin{align*}
        &\E_P\bigg[\Big|2N^{-1}n_{\max}^{\kappa}\sum_{i\in\mathcal{I}_k}\phi^2(S_i;\hat{\eta}^{(k)},\hat{\nu}^{(k)})\big(\hat{\beta}-\beta\big)^2\Big|\bigg] \\
        \leq& 2N^{-1}n_{\max}^{\kappa}\sum_{i\in\mathcal{I}_k}\E_P\left[\phi^2(S_i;\hat{\eta}^{(k)},\hat{\nu}^{(k)})\big(\hat{\beta}-\beta\big)^2\right] \\
        \lesssim& N^{-1}n_{\max}^{\kappa}\sum_{i\in\mathcal{I}_k}\Big(\E_{P|S_{\mathcal{I}_k^c}}\left[\norm{\xi_i}_2^4\big(\hat{\beta}-\beta\big)^2\right] + \E_{P|S_{\mathcal{I}_k^c}}\left[\big\|\mathcal{E}_m^{(k)}(X_i)\big\|_2^4\big(\hat{\beta}-\beta\big)^2\right]\Big)  \\
        \leq& N^{-1}n_{\max}^{\kappa}\sum_{i\in\mathcal{I}_k}\Big(\norm{\xi_i}_{P,4+\delta}^4\big\|\hat{\beta}-\beta\big\|_{P,\frac{2(4+\delta)}{\delta}}^2 + \big\|\mathcal{E}_m^{(k)}(X_i)\big\|_{P,4+\delta}^4\big\|\hat{\beta}-\beta\big\|_{P,\frac{2(4+\delta)}{\delta}}^2\Big) \\
        \leq& N^{-2}n_{\max}^{\kappa}\sum_{i\in\mathcal{I}_k}n_i^2\bigg(\Big\|n_i^{-\frac{1}{2}}\xi_i\Big\|_{P,4+\delta}^4\norm{N^{\frac{1}{2}}\big(\hat{\beta}-\beta\big)}_{P,\frac{2(4+\delta)}{\delta}}^2 
        \\
        & \hspace{12em} + \Big\|n_i^{-\frac{1}{2}}\mathcal{E}_m^{(k)}(X_i)\Big\|_{P,4+\delta}^4\norm{N^{\frac{1}{2}}\big(\hat{\beta}-\beta\big)}_{P,\frac{2(4+\delta)}{\delta}}^2\bigg) \\
        \lesssim& N^{-2}n_{\max}^{\kappa}\sum_{i\in\mathcal{I}_k}n_i^2 \leq N^{-1}n_{\max}^{1+\kappa},
    \end{align*}
    by using the fact that any finite moment of a normal distribution is finite. Thus
    \begin{equation*}
        \sup_{P\in\mathcal{P}_I}\E_P\bigg[\Big|2N^{-1}n_{\max}^{\kappa}\sum_{i\in\mathcal{I}_k}\phi^2(S_i;\hat{\eta}^{(k)},\hat{\nu}^{(k)})\big(\hat{\beta}-\beta\big)^2\Big|\bigg]
        \lesssim
        N^{-1}n_{\max}^{1+\kappa}
        =
        o(1),
    \end{equation*}
    by Assumption~\ref{ass:n_max}. By Lemmas~\ref{lem:unif_add}~and~\ref{lem:unif_ciiu} we have that term \(\RN{10}\) is \(o_{\mathcal{P}}(1)\).
    
    \paragraph{Term \RN{11}:} Note that
    \begin{align*}
        &\E_P\bigg[\Big|N^{-1}n_{\max}^{\kappa}\sum_{i\in\mathcal{I}_k}\varphi(S_i;\beta,\eta_0,\nu^*)\phi(S_i;\hat{\eta}^{(k)},\hat{\nu}^{(k)})\big(\hat{\beta}-\beta\big)\Big|\bigg] \\
        \lesssim& N^{-1}n_{\max}^{\kappa}\sum_{i\in\mathcal{I}_k}\E_P\left[\norm{\xi_i}_2^2\left(\norm{\xi_i}_2+\big\|\mathcal{E}_m^{(k)}(X_i)\big\|_2\right)\norm{\varepsilon_i}_2\big|\hat{\beta}-\beta\big|\right] \\
        \leq& N^{-1}n_{\max}^{\kappa}\sum_{i\in\mathcal{I}_k}\Big(\norm{\xi_i}_{P,4+\delta}^3\norm{\varepsilon_i}_{P,4+\delta}\big\|\hat{\beta}-\beta\big\|_{P,\frac{4+\delta}{\delta}} \\ 
        &\hspace{6em} + \norm{\xi_i}_{P,4+\delta}^2\norm{\varepsilon_i}_{P,4+\delta}\big\|\mathcal{E}_m^{(k)}(X_i)\big\|_{P,4+\delta}\big\|\hat{\beta}-\beta\big\|_{P,\frac{4+\delta}{\delta}} \Big) \\
        \lesssim& N^{-\frac{3}{2}}n_{\max}^{\kappa}\sum_{i\in\mathcal{I}_k}n_i^2\big\|N^{\frac{1}{2}}\big(\hat{\beta}-\beta\big)\big\|_{P,\frac{4+\delta}{\delta}} \\
        \lesssim& N^{-\frac{3}{2}}n_{\max}^{\kappa}\sum_{i\in\mathcal{I}_k}n_i^2 \leq N^{-\frac{1}{2}}n_{\max}^{1+\kappa},
    \end{align*}
    by H\"older's inequality (twice), and thus
    \begin{equation*}
        \sup_{P\in\mathcal{P}_I}\E_P\bigg[\Big|N^{-1}n_{\max}^{\kappa}\sum_{i\in\mathcal{I}_k}\varphi(S_i;\beta,\eta_0,\nu^*)\phi(S_i;\hat{\eta}^{(k)},\hat{\nu}^{(k)})\big(\hat{\beta}-\beta\big)\Big|\bigg] 
        \lesssim
        N^{-\frac{1}{2}}n_{\max}^{1+\kappa} = o(1)
    \end{equation*}
    and Assumption~\ref{ass:n_max}. It then follows by Lemmas~\ref{lem:unif_add}~and~\ref{lem:unif_ciiu} that the term \(\RN{11}\) is \(o_{\mathcal{P}}(1)\). As all terms are \(o_{\mathcal{P}}(1)\), applying Lemma~\ref{lem:unif_add} completes the proof.

\end{proof}

\begin{lemma}\label{lem:M4}
Under the setup of Theorem~\ref{thm:main},
\begin{equation*}
    M_4 := N^{-1}n_{\max}^{\kappa}\sum_{i=1}^I\left[\varphi^2(S_i;\beta,\eta_0,\nu^*)-\E_{P}\left[\varphi^2(S_i;\beta,\eta_0,\nu^*)\right]\right] = o_{\mathcal{P}}(1).
\end{equation*}
\end{lemma}
\begin{proof}
We will consider the two cases \(\delta\in(0,4)\) and \(\delta\geq 4\) separately. If \(\delta\in(0,4)\) then, by the von Bahr--Esseen inequality, for all $P\in\mathcal{P}_I$
\begin{align*}
    \E_{P}\left[\left|M_4\right|^{1+\frac{\delta}{4}}\right] \leq & N^{-\left(1+\frac{\delta}{4}\right)}n_{\max}^{\kappa\left(1+\frac{\delta}{4}\right)}\E_{P}\left[\bigg|\sum_{i=1}^I\varphi^2\left(S_i;\beta,\eta_0,\nu^*\right)-\E_{P}\left[\varphi^2\left(S_i;\beta,\eta_0,\nu^*\right)\right]\bigg|^{1+\frac{\delta}{4}}\right] \\
    \leq& 2N^{-\left(1+\frac{\delta}{4}\right)}n_{\max}^{\kappa\left(1+\frac{\delta}{4}\right)}\sum_{i=1}^I\E_{P}\left[\left|\varphi^2\left(S_i;\beta,\eta_0,\nu^*\right)-\E_{P}\left[\varphi^2\left(S_i;\beta,\eta_0,\nu^*\right)\right]\right|^{1+\frac{\delta}{4}}\right] \\
    \leq& 2N^{-\left(1+\frac{\delta}{4}\right)}n_{\max}^{\kappa\left(1+\frac{\delta}{4}\right)}\sum_{i=1}^I\E_{P}\left[\left|\varphi\left(S_i;\beta,\eta_0,\nu^*\right)\right|^{2+\frac{\delta}{2}}\right] \\
    \leq& 2N^{-\left(1+\frac{\delta}{4}\right)}n_{\max}^{\kappa\left(1+\frac{\delta}{4}\right)}\sum_{i=1}^I\E_{P}\left[\norm{W_i^*}_{\op}^{2+\frac{\delta}{2}}\norm{\xi_i}_2^{2+\frac{\delta}{2}}\norm{\varepsilon_i}_2^{2+\frac{\delta}{2}}\right] \\
    \lesssim& N^{-\left(1+\frac{\delta}{4}\right)}n_{\max}^{\kappa\left(1+\frac{\delta}{4}\right)}\sum_{i=1}^I\norm{\xi_i}_{P,4+\delta}^{2+\frac{\delta}{2}}\norm{\varepsilon_i}_{P,4+\delta}^{2+\frac{\delta}{2}} \\
    \lesssim& N^{-\left(1+\frac{\delta}{4}\right)}n_{\max}^{\kappa\left(1+\frac{\delta}{4}\right)}\sum_{i=1}^In_i^{2+\frac{\delta}{2}} \leq N^{-\frac{\delta}{4}}n_{\max}^{(1+\kappa)+(2+\kappa)\frac{\delta}{4}},
\end{align*}
and so
\begin{equation*}
    \sup_{P\in\mathcal{P}_I}\E_{P}\left[\left|M_4\right|^{1+\frac{\delta}{4}}\right]
    \lesssim
    N^{-\frac{\delta}{4}}n_{\max}^{(1+\kappa)+(2+\kappa)\frac{\delta}{4}} = o(1),
\end{equation*}
by Assumption~\ref{ass:n_max}.

If \(\delta\geq 4\) then, again for all $P\in\mathcal{P}_I$,
\begin{align*}
    \E_{P}\left[M_4^2\right] \leq& N^{-2}n_{\max}^{2\kappa}\sum_{i=1}^I\E_{P}\left[\varphi^4\left(S_i;\beta,\eta_0,\nu^*\right)\right] \\
    \leq& N^{-2}n_{\max}^{2\kappa}\sum_{i=1}^I\E_{P}\left[\left|\xi_i^{\top}W_i^*\varepsilon_i\right|^4\right] \\
    \leq& N^{-2}n_{\max}^{2\kappa}\sum_{i=1}^I\E_{P}\left[\norm{W_i^*}_{\op}^4\norm{\xi_i}_2^4\norm{\varepsilon_i}_2^4\right] \\
    \lesssim& N^{-2}n_{\max}^{2\kappa}\sum_{i=1}^I\norm{\xi_i}_{P,8}^4\norm{\varepsilon_i}_{P,8}^4 \\
    \lesssim& N^{-2}n_{\max}^{2\kappa}\sum_{i=1}^In_i^4 \leq N^{-1}n_{\max}^{3+2\kappa},
\end{align*}
and so
\begin{equation*}
    \sup_{P\in\mathcal{P}_I}\E_{P}\left[M_4^2\right]
    \lesssim N^{-1}n_{\max}^{3+2\kappa} = o(1),
\end{equation*}
by Assumptions~\ref{A1.1}~and~\ref{A3.3}.

Thus, in each case we have by applying Markov's inequality that $M_4=o_{\mathcal{P}}(1)$.
\end{proof}

\subsection{Auxiliary uniform convergence results}

For the uniform convergence results of Theorem~\ref{thm:main} we require the following lemmata. Given a sequence of families of probability distributions \(\left(\mathcal{P}_I\right)_{I\in\mathbb{N}}\) we use the shorthand \(\left(A_I\right)_{I\in\mathbb{N}}\) to denote a sequence of random variables \(\left(A_{P,I}\right)_{P\in\mathcal{P}_I,I\in\mathbb{N}}\).

\begin{lemma}\label{lem:unif_add}
Let \(\left(A_I\right)_{I\in\mathbb{N}}\) and \(\left(B_I\right)_{I\in\mathbb{N}}\) be sequences of real-valued random variables. If \(A_I=o_{\mathcal{P}}(1)\) and \(B_I=o_{\mathcal{P}}(1)\) then \(A_I+B_I=o_{\mathcal{P}}(1)\).
\end{lemma}

\begin{lemma}\label{lem:unif_leq}
Let \(\left(A_I\right)_{I\in\mathbb{N}}\) and \(\left(B_I\right)_{I\in\mathbb{N}}\) be sequences of real-valued random variables. Suppose that \(A_I=o_{\mathcal{P}}(1)\) and \(B_I\leq A_I\) for all \(I\in\mathbb{N}\) uniformly over all \(P\in\mathcal{P}_I\). Then \(B_I=o_{\mathcal{P}}(1)\).
\end{lemma}

\begin{lemma}\label{lem:unif_times}
Let \(\left(A_I\right)_{I\in\mathbb{N}}\) and \(\left(B_I\right)_{I\in\mathbb{N}}\) be sequences of real-valued random variables. If \(A_I = o_\mathcal{P}(1)\) and
\(B_I = O_{\mathcal{P}}(1)\) then \(A_IB_I=o_{\mathcal{P}}(1)\).
\end{lemma}

\begin{lemma}\label{lem:unif_ciiu}
Let \(\left(A_I\right)_{I\in\mathbb{N}}\) and \(\left(B_I\right)_{I\in\mathbb{N}}\) be a sequences of real-valued random variables. If there exists some \(\alpha\geq 1\) such that \(\E_{P}(\left|A_I\right|^{\alpha}\given B_I)=o_{\mathcal{P}}(1)\), then \(A_I=o_{\mathcal{P}}(1)\).
\end{lemma}

For proofs of these four lemmas see, for example, \citet{lundborg2022projected}

\begin{lemma}\label{lem:unif_max}
    Let $(A_I^{(1)})_{I\in\mathbb{N}}, \ldots, (A_I^{(K)})_{I\in\mathbb{N}}$ be sequences of real-valued random variables satisfying $A_I^{(k)}=o_{\mathcal{P}}(1)$ for all $k\in[K]$, where $K\in\mathbb{N}$ is finite. Then $\max_{k\in[K]}A_I^{(k)} = o_{\mathcal{P}}(1)$.
\end{lemma}
\begin{proof}
    For any $\epsilon>0$
    \begin{equation*}
        \sup_{P\in\mathcal{P}}\PP_P\left(
        \left|\max_{k\in[K]}A_I^{(k)}\right| > \epsilon\right)
        \leq \sup_{P\in\mathcal{P}}\PP_P\left(\max_{k\in[K]}\big|A_I^{(k)}\big|>\epsilon\right)
        \leq\sum_{k=1}^K\sup_{P\in\mathcal{P}}\PP_P\left(\big|A_I^{(k)}\big|>\epsilon\right)
        = o(1),
    \end{equation*}
    completing the proof.
\end{proof}

\begin{lemma}\label{lem:unif_clt}
Let \(\left(A_{I,i}\right)_{I\in\mathbb{N},i\in[I]}\) be a triangular array of real-valued random variables satisfying:
\begin{itemize}
    \item \(A_{I,1},\ldots,A_{I,I}\) are independent;
    \item \(\E_{P}\left[A_{I,i}\right]=0\hspace{4mm} (\forall I\in\mathbb{N}, i\in[I])\);
    \item There exists some \(\Delta>0\) such that
    \[
    \lim_{I\to\infty}\sup_{P \in \mathcal{P}_I}\frac{\sum_{i=1}^I\E_{P}\left[\left|A_{I,i}\right|^{2+\Delta}\right]}{\left(\sum_{i=1}^I\E_{P}\left[A_{I,i}^2\right]\right)^{1+\frac{\Delta}{2}}} = 0.
    \]
\end{itemize}
Then \(S_I := \left(\sum_{i=1}^I\E_{P}\left[A_{I,i}^2\right]\right)^{-\frac{1}{2}}\left(\sum_{i=1}^IA_{I,i}\right)\) converges uniformly to \(N(0,1)\), i.e.
\begin{equation*}
    \lim_{I\to\infty}\sup_{P\in\mathcal{P}_I}\sup_{t\in\R}\left|\PP_{P}\left(S_I\leq t\right)-\Phi(t)\right|=0.
\end{equation*}
\end{lemma}
\begin{proof}
For each \(I\in\mathbb{N}\), let \(P\in\mathcal{P}_I\) satisfy
\begin{equation}\label{eq:clt_proof}
    \sup_{P\in\mathcal{P}_I}\sup_{t\in\R}\left|\PP_{P}\left(S_I\leq t\right)-\Phi(t)\right| \leq \sup_{t\in\R}\left|\PP_{P}\left(S_I\leq t\right)-\Phi(t)\right| + I^{-1}.
\end{equation}
Let \(T_{I,i} := \left(\sum_{i=1}^I\E_{P}\left[A_{I,i}^2\right]\right)^{-\frac{1}{2}}\left(A_{I,i}\right)\) for each \(i\in[I]\). Note that by construction \(T_{I,1},\ldots,T_{I,I}\) are independent and mean zero and \(\sum_{i=1}^I\E_{P}(T_{I,i}^2)=1\). By the Lindeberg--Feller central limit theorem for triangular arrays (Proposition 2.27, \citet{vandervaart}) we have
\begin{equation*}
    \lim_{I\to\infty}\sup_{t\in\R}\left|\PP_{P}\left(S_I\leq t\right)-\Phi(t)\right| = 0,
\end{equation*}
provided the condition for the Lindeberg--Feller theorem holds. This condition holds as, for all \(\epsilon>0\),
\begin{align*}
    \sum_{i=1}^I\E_{P}\left[\left|T_{I,i}\right|^2\ind_{\left\{\left|T_{I,i}\right|>\epsilon\right\}}\right] 
    &\leq \sum_{i=1}^I\left(\E_{P}\left[\left|T_{I,i}\right|^{2+\Delta}\right]\right)^{\frac{2}{2+\Delta}}\big(\PP_{P}\left(\left|T_{I,i}\right|>\epsilon\right)\big)^{\frac{\Delta}{2+\Delta}} \\
    &\leq \epsilon^{-\Delta}\sum_{i=1}^I\E_{P}\left[\left|T_{I,i}\right|^{2+\Delta}\right] \\
    &= \epsilon^{-\Delta}\cdot\frac{\sum_{i=1}^I\E_{P}\left[\left|A_{I,i}\right|^{2+\Delta}\right]}{\left|\sum_{i=1}^I\E_{P}[A_{I,i}^2]\right|^{1+\frac{\Delta}{2}}}
\end{align*}
by H\"older's inequality followed by Markov's inequality. Taking limits of \eqref{eq:clt_proof} therefore completes the proof.
\end{proof}

\begin{lemma}[\citet{shahpeters}, Lemma 20]\label{lem:unif_slutsky}
Let \(\left(A_I\right)_{I\in\mathbb{N}}\) and \(\left(B_I\right)_{I\in\mathbb{N}}\) be sequences of real-valued random variables. Suppose
\begin{equation*}
    \lim_{I\to\infty}\sup_{P\in\mathcal{P}_I}\sup_{t\in\R}\left|\PP_{P}\left(A_I\leq t\right)-\Phi(t)\right| = 0.
\end{equation*}
Then we have:
\begin{enumerate}
    \item[(a)] \[\text{If } B_I=o_{\mathcal{P}}(1) \text{ then } \lim_{I\to\infty}\sup_{P\in\mathcal{P}_I}\sup_{t\in\R}\left|\PP_{P}\left(A_I+B_I\leq t\right)-\Phi(t)\right| = 0;\]
    \item[(b)] \[\text{If } B_I=1+o_{\mathcal{P}}(1) \text{ then } \lim_{I\to\infty}\sup_{P\in\mathcal{P}_I}\sup_{t\in\R}\left|\PP_{P}\left(A_I/B_I\leq t\right)-\Phi(t)\right| = 0.\]
\end{enumerate}
\end{lemma}

\subsection{Proof of Corollary~\ref{thm:(s,theta)-all}}\label{appsec:(s,theta)-proof}
\begin{proof}[Proof of Theorem~\ref{thm:(s,theta)-all}]
We aim to show that Assumption~\ref{A3.1'} implies \ref{A3.1} and \ref{A3.2'} implies \ref{A3.2}. Note that for all \(k\in[K]\), \(i\in\mathcal{I}_k\), \(j\in[n_i]\),
    \begin{equation*}
        \E_{P}\left[\left(\hat{\sigma}^{(k)}(X_{ij})-\sigma^*(X_{ij})\right)^2\right] = \E_{P}\left[\frac{\left(\hat{s}^{(k)}(X_{ij})-s^*(X_{ij})\right)^2}{\hat{s}^{(k)}(X_{ij})^2s^*(X_{ij})^2}\right] \lesssim \E_{P}\left[\left(\hat{s}^{(k)}(X_{ij})-s^*(X_{ij})\right)^2\right],
    \end{equation*}
    and so \(\mathcal{R}_{\sigma}\lesssim\mathcal{R}_{s^*}\), hence Assumption~\ref{A3.1'} implies \ref{A3.1}. 

    To show that Assumption~\ref{A3.2'} implies 
    \ref{A3.2}, we consider each case separately:
    
    \paragraph{Equicorrelation Case: }
    Note that for all \(k\in[K]\),
    \begin{equation*}
        \left|\hat{\rho}^{(k)}-\rho^*\right| = \frac{\left|\hat{\theta}^{(k)}-\theta^*\right|}{|1+\theta^*|\cdot |1+\hat{\theta}^{(k)}|} \lesssim \left|\hat{\theta}^{(k)}-\theta^*\right| ,
    \end{equation*}
    and therefore \(\mathcal{R}_{\rho} \lesssim \mathcal{R}_{\theta^*}\).
    
\paragraph{Hierarchical / Nested Case: }
Note that for all \(k\in[K]\),
\begin{align*}
    \left|\hat{\rho}^{(k)}_1-\rho^*_1\right| &= \frac{\left|\hat{\theta}^{(k)}_1+\hat{\theta}^{(k)}_2-\theta^*_1-\theta^*_2\right|}{\left|1+\hat{\theta}^{(k)}_1+\hat{\theta}^{(k)}_2\right|\cdot\left|1+\theta^*_1+\theta^*_2\right|} \lesssim \left|\hat{\theta}^{(k)}_1-\theta^*_1\right|+\left|\hat{\theta}^{(k)}_2-\theta^*_2\right| \leq \sqrt{2}\norm{\hat{\theta}^{(k)}-\theta^*}_2, \\
    \left|\hat{\rho}^{(k)}_2-\rho^*_2\right| &= \frac{\left|\left(1+\theta^*_1\right)\left(\hat{\theta}_2-\theta^*_2\right)-\theta^*_2\left(\hat{\theta}^{(k)}_1-\theta^*_1\right)\right|}{\left|1+\hat{\theta}^{(k)}_1+\hat{\theta}^{(k)}_2\right|\cdot\left|1+\theta^*_1+\theta^*_2\right|} \lesssim \sqrt{2}\norm{\hat{\theta}^{(k)}-\theta^*}_2,
\end{align*}
and therefore \(\mathcal{R}_{\rho} \lesssim \mathcal{R}_{\theta^*}\).

\paragraph{Longitudinal $\text{AR}(1)$ Case: }
Note that for all \(k\in[K]\), \(i\in\mathcal{I}_k\), \(j\in[n_i]\), \(j'\in[n_i]\backslash\left\{j\right\}\),
\begin{align*}
    &\left|\hat{\rho}^{(k)}(X_{ij},X_{ij'})-\rho^*(X_{ij},X_{ij'})\right| = \left|\hat{\theta}^{(k)^{\left|j-j'\right|}}-\theta^{*^{\left|j-j'\right|}}\right| \\ 
    \leq& \left|\hat{\theta}^{(k)}-\theta\right|\cdot\left|\sum_{l=1}^{\left|j-j'\right|-1}\hat{\theta}^{(k)^l}\theta^{*^{\left|j-j'\right|-1-l}}\right| \leq\left|\hat{\theta}^{(k)}-\theta^*\right|\cdot\left|j-j'\right|\cdot\left(\hat{\theta}^{(k)}\vee\theta^*\right)^{\left|j-j'\right|-1} \\
    \leq& \left|\hat{\theta}^{(k)}-\theta^*\right|\cdot\left|j-j'\right|\left(1-c\right)^{\left|j-j'\right|-1}
    \leq \left(\log\left(\frac{1}{1-c}\right)\right)^{-1}(1-c)^{\frac{1}{\log\left(\frac{1}{1-c}\right)}-1}\left|\hat{\theta}^{(k)}-\theta^*\right| \\
    \lesssim& \left|\hat{\theta}^{(k)}-\theta^*\right|,
\end{align*}
for some \(c>0\),
with the last inequality following by elementary calculus (optimising the function \(N\mapsto N(1-c)^{N-1}\)). Therefore \(\mathcal{R}_{\rho} \lesssim \mathcal{R}_{\theta^*}\).

Hence, for each working correlation, Assumption~\ref{A3.2'} implies \ref{A3.2}.

Finally, we must show that for each working correlation \(\Lambda_{\max}(W_i^*)\lesssim 1\), \(\Lambda_{\min}(W_i^*)\gtrsim n^{-\gamma}\) for all \(i\in[I]\). We consider each separately.

\paragraph{Equicorrelation Case:}
    Consider the decomposition of the matrix \(W_i^*\) as the sum of two matrices \(W_i=W_i^{*(A)}+W_i^{*(B)}\)
    \begin{align*}
        \big(W_i^{*(A)}\big)_{jj'} &= \delta_{jj'}s^*(X_{ij})s^*(X_{ij'}) 
        \\
        \big(W_i^{*(B)}\big)_{jj'} &= -\frac{\theta}{1+\theta n_i}s^*(X_{ij})s^*(X_{ij'}),
    \end{align*}
    for $j,j'\in[n_i]$
    Then, we have that
    \begin{equation*}
        \Lambda_{\max}(W_i^*) \leq \max_{j\in[n_i]}s^*(X_{ij})^2 + 0 \lesssim 1,
    \end{equation*}
    and further that, by Weyl's inequality,
    \begin{multline*}
        \Lambda_{\min}(W_i^*) \geq \min_{j\in[n_i]}s^*(X_{ij})^2-\frac{\theta^*}{1+\theta^* n_i}\norm{s^*(X_i)}_2^2 \\ \geq \left(1-\frac{\theta^* n_i}{1+\theta^* n_i}\right)\min_{j\in[n_i]}s^*(X_{ij})^2 \gtrsim \frac{1}{1+\theta^* n_i}\geq\frac{1}{1+\theta^*}\cdot\frac{1}{n_i} \gtrsim n_i^{-1}.
    \end{multline*}
    
\paragraph{Hierarchical / Nested Case:}
First consider the decomposition of \(W_i^*=W_i^{*(A)}+W_i^{*(B)}\) where
\begin{align*}
    \big(W_i^{*(A)}\big)_{jj'} &= (1+\theta^*_1+\theta^*_2)\delta_{jj'}s^*(X_{ij})s^*(X_{ij'}),
    \\
    \big(W_i^{*(B)}\big)_{jj'} &= -\mathcal{B}_{i,jj'}s^*(X_{ij})s^*(X_{ij'}),
\end{align*}
for $j,j'\in[n_i]$ and \(\mathcal{B}\) as defined in Section~\ref{sec:boosting}. Then
\begin{equation*}
    \Lambda_{\max}(W_i^*)\leq \max_{j\in[n_i]}s^*(X_{ij})^2 + 0 \lesssim 1
\end{equation*}
Further, by considering \(W_i^{*^{-1}}\) parametrised by \((\rho^*_1,\rho^*_2)=\left(\frac{\theta^*_1+\theta^*_2}{1+\theta^*_1+\theta^*_2},\frac{\theta^*_2}{1+\theta^*_1+\theta^*_2}\right)\), we can similarly show that \(\Lambda_{\max}(W_i^{*^{-1}})\lesssim n_i\), and therefore \(\Lambda_{\min}(W_i^*)\gtrsim n_i^{-1}\), completing the proof.

\paragraph{Longitudinal $\text{AR}(1)$ Case:}
By the Gershgorin circle theorem alongside the triangle inequality,
\begin{align*}
    \Lambda_{\max}(W_i^*) \leq & (1+2\theta^*)\max_{j\in[n_i]}s^*(X_{ij})^2 \lesssim 1,
    \\
    \Lambda_{\max}(W_i^{*^{-1}}) \leq & \left(\min_{j\in[n_i]}s^*(X_{ij})^2\right)^{-1} + 2\left(\min_{j\in[n_i]}s^*(X_{ij})^2\right)^{-1}\sum_{j=1}^{n_i-1}\theta^j \\
    =& \left(1+\frac{2\theta^*(1-\theta^{*^{n_i-1}})}{1-\theta^*}\right)\left(\min_{j\in[n_i]}s^*(X_{ij})^2\right)^{-1} 
    \\
    \leq&\left(1+\frac{2}{1-\theta^*}\right)\left(\min_{j\in[n_i]}s^*(X_{ij})^2\right)^{-1} \lesssim 1,
\end{align*}
and so $\Lambda_{\min}(W_i^*)=\big(\Lambda_{\max}(W_i^{*^{-1}})\big)^{-1}\gtrsim 1$ as required.

\end{proof}

\section{Practical details of numerical results in Section~\ref{sec:numerical_results}}\label{appsec:numerical_results}

\subsection{Details of weight estimators studied in~\ref{sec:simulations}}\label{sec:loss_details}

The `heteroscedastic ML' estimator used in Section~\ref{sec:numerical_results} follows Algorithm~\ref{alg:main} but with weights constructed as follows. Adopting the notation in Algorithm~\ref{alg:main} with weights as in~\eqref{eq:sigma-rho-workingcov}, we construct
\begin{gather*}
    (\hat{\phi}_\ML,\hat{\rho}_\ML,\hat{\beta}_\ML) = \underset{(\phi,\rho,\beta)\in\R^5\times(-1,1)\times\R}{\argmin} l_\ML(\phi,\rho,\beta),
    \\
   l_\ML(\phi,\rho,\beta) := \sum_{i\in\mathcal{I}_k^c} \Big(-\log\det W(\phi,\rho)(X_i) + \big(\tilde{R}^Y_i-\beta\tilde{R}^D_i\big)^{\top}W(\phi,\rho)(X_i)\big(\tilde{R}^Y_i-\beta\tilde{R}^D_i\big)\Big),
    \\
    \left\{W(\phi,\rho)(X_i)\right\}^{-1} = D_{\sigma_\phi}(X_i)C_{\theta_\rho}D_{\sigma_\phi}(X_i),
    \qquad
    \sigma_\phi(x) = \exp\big(\phi^{\top}(1,x,x^2,x^3,x^4)\big),
\end{gather*}
with working correlations $(C_{\theta_\rho})_{jk}=\ind_{\left\{j=k\right\}} + \rho\cdot\ind_{\left\{j\neq k\right\}}$ for examples~\ref{sec:sim-complexity}~and~\ref{sec:sim-var} (equicorrelated) and $(C_{\theta_\rho})_{jk}=\rho^{|j-k|}$ for examples~\ref{sec:sim-misspec}~and~\ref{sec:sim-corr} ($\text{AR}(1)$). These are estimated using the \texttt{nlme} package~\citep{nlme-code}. Weights are then constructed as
\begin{equation*}
     \big\{\hat{W}_\ML^{(k)}(\cdot)\big\}^{-1} = D_{\hat{\sigma}_\ML}(\cdot)C_{\theta_{\hat{\rho}_\ML}}D_{\hat{\sigma}_\ML}(\cdot),
     \qquad
     \hat{\sigma}_{\ML}(x):=\exp\big(\hat{\phi}_\ML^{\top}(1,x,x^2,x^3,x^4)\big)
\end{equation*}
and used in Algorithm~\ref{alg:main} to construct the final $\hat{\beta}$. The `homoscedastic ML' estimator is constructed as above fixing $\hat{\phi}_\ML={\bf 0}$.

The `heteroscedastic GEE' estimator also follow the setup of Algorithm~\ref{alg:main} with weights constructed as follows. First we generate a generalised additive cubic spline estimator for the variance and estimator for the correlation parameter as
\begin{align*}
    \hat{\sigma}_{\GEE} &:= \underset{\sigma\in\Pi}{\argmin}\sum_{i\in\mathcal{I}_k^c}\sum_{j=1}^{n_i}\big(\tilde{\varepsilon}_{ij}^2-\sigma^2(X_{ij})\big)^2  +  \lambda J(\sigma^2), &
    J(f)=\int\left(f''(x)\right)^2dx, &
    \\
    \hat{\rho}_{\GEE} &:= \underset{\rho\in(-1,1)}{\argmin}\sum_{i\in\mathcal{I}_k^c}\sum_{\substack{j,k=1\\j\neq k}}^{n_i} \Big(\tilde{\varepsilon}_{ij}\tilde{\varepsilon}_{ik}-\hat{\sigma}_{\GEE}(X_{ij})\hat{\sigma}_{\GEE}(X_{ik})(C_{\theta_\rho})_{jk}\Big)^2, & &
\end{align*}
where $\Pi$ is the set of $\R_{>0}$ valued functions that have absolutely continuous first derivative, $\lambda$ is a tuning parameter that minimises the generalised cross-validation (GCV) score, and $C_{\theta_\rho}$ takes the equicorrelated or $\text{AR}(1)$ structure for each example as outlined above. In all the simulations in Section~\ref{sec:simulations} we fit $\hat{\sigma}_\GEE$ using the \texttt{gam} package~\citep{gam-code}, and $\hat{\rho}_\GEE$ using the \texttt{geepack} package~\citep{geepack}. Weights are then constructed as
\begin{equation*}
    \big\{\hat{W}_\GEE^{(k)}(\cdot)\big\}^{-1} = D_{\hat{\sigma}_\GEE}(\cdot)C_{\theta_{\hat{\rho}_\GEE}}D_{\hat{\sigma}_\GEE}(\cdot),
\end{equation*}
and used in Algorithm~\ref{alg:main} to construct $\hat{\beta}$. The `homoscedastic GEE' estimator is estimated in the same way except we set $\hat{\sigma}_\GEE\equiv 1$.

\subsection{Sandwich boosting with variable step size}

In Section~\ref{sec:boosting} we outlined the sandwich boosting algorithm~\ref{alg:boosting} in terms of a fixed step size $\lambda^{(s)}$, with an optimal number of boosting iterations $m_{\text{stop}}$ selected by cross-validation. In this setup $\lambda^{(s)}$ acts as a hyperparameter; if it is set to be too large one can be susceptible to under-fitting (being unable to attain a local minimum), and if set too small it can take an unreasonably large $m_{\text{stop}}$ to achieve convergence. An adaptation of Algorithm~\ref{alg:boosting} that may be helpful in practice would be to replace the $\lambda^{(s)}$ constant step size with a variable step size
\begin{align*}
    \lambda^{(s)}_m \approx \underset{\lambda\in\Lambda}{\argmin}\,\hat{L}_{\SL}(\hat{s}_m-\lambda\hat{u}_m, \hat{\theta}_m),
\end{align*}
over some closed interval $\Lambda$ of positive step sizes, attained by minimising the second order expansion of $\lambda\mapsto\hat{L}_{\SL}(\hat{s}_m-\lambda\hat{u}_m, \hat{\theta}_m)$. The $m$th iterative update of the $s$-function is then taken as
\begin{align*}
    \hat{s}_{m+1} = (\hat{s}_m -\mu^{(s)}\lambda^{(s)}_m\hat{u}_m) \vee \epsilon,
\end{align*}
for some constant shrinkage parameter $\mu^{(s)}\in(0,1)$, whose purpose is to avoid overfitting, and small constant $\epsilon >0$ which ensures that $\hat{s}_{m+1}$ is always strictly positive everywhere; in all our experiments, we set $\epsilon = 0.1$. With this setup boosting, can then be performed until a convergence criterion is achieved. This can eliminate the need to determine an approximately optimal $(\lambda^{(s)},m_{\text{stop}})$, as well as achieve faster convergence of sandwich boosting. We implement this variable step size scheme in some of the numerical examples studied in Section~\ref{sec:numerical_results}.

\subsection{Details of numerical results in Section~\ref{sec:numerical_results}}\label{appsec:numerical-details}

Sandwich boosting for the numerical results is setup as follows.

\smallskip\noindent\textbf{Increasing Model Complexity (Section~\ref{sec:sim-complexity}): } The variable step size adaptation of sandwich boosting is used, with $\Lambda=[0.001,10]$ and $\lambda^{(\theta)}=0.1$.

\smallskip\noindent\textbf{Increasing Covariance Misspecification (Section~\ref{sec:sim-misspec}): } A constant step size $\lambda^{(s)}=10$ is used, with $\lambda^{(\theta)}=0.1$ and $m_{\text{stop}}=200$ selected by cross-validation.

\smallskip\noindent\textbf{Correlation Misspecification (Section~\ref{sec:sim-corr}): } As this simulated example involves optimisation of the sandwich loss over a parameteric weight class in terms of a single $\text{AR}(1)$ parameter, there is no need to select a $(\lambda^{(s)},m_{\text{stop}})$.

\smallskip\noindent\textbf{Conditional Variance Misspecification (Section~\ref{sec:sim-var}): } A constant step size $\lambda^{(s)}=100$ is used, with $\lambda^{(\theta)}=0.1$ and $m_{\text{stop}}=500$ selected by cross-validation.

\smallskip\noindent\textbf{Real world data scenarios (Section~\ref{sec:real-data}): }
We carry out sandwich boosting for both with the variable step size adaptation within the interval $\Lambda=[0.01,1]$, with $s$-shrinkage parameter $\mu^{(s)}=0.1$ and $\theta$-step size $\lambda^{(\theta)}=0.1$.

When analysing the real world data discussed in Section~\ref{sec:real-data}, we alleviate the dependence of the cross-fitting method on the randomness of the sample splits by repeating the procedure of Algorithm~\ref{alg:main} $\mathcal{S}=50$ times for different random sample splits. Suppose for the repeat indexed by $s\in[\mathcal{S}]$ we have $\beta$-estimator $\hat{\beta}_s$ and asymptotic variance estimator $\hat{V}_s$. Then we aggregate the results of these repeats to produce a final estimate of $\beta$ and associated estimate of its variance via
\begin{equation*}
    \hat{\beta} := \underset{s\in[\mathcal{S}]}{\text{median}}\,\hat{\beta}_s,
    \qquad
    \hat{V} := \underset{s\in[\mathcal{S}]}{\text{median}}\,\big(\hat{V}_s + (\hat{\beta}-\hat{\beta}_s)^2\big);
\end{equation*}
see for example \citet{chern, emmenegger}.

\subsection{Choice of parameters for cross-fitting} \label{sec:cross-fitting_K}

As outlined in Section~\ref{sec:numerical_results}, for the simulations in Section~\ref{sec:simulations} we perform cross-fitting with $K=2$ folds and for $\mathcal{S}=1$ random sample splits. The choice of $(K,\mathcal{S})$ does not affect the asymptotic mean squared error of the resulting estimator (see Theorem~\ref{thm:main}), and thus for sufficiently large sample size the choice of cross-fitting parameters does not affect the resulting mean squared error. Figure~\ref{fig:K-and-S} (left panel) shows a simulation in Section~\ref{sec:simulations} (specifically Example~\ref{sec:sim-var} with group size $n_i\equiv 8$) for cross-fitting parameters $K\in\{2,5,10\}$. We see the choice of $K$ does not have a significant effect on the mean squared error of $\hat{\beta}$ at the sample sizes considered in Section~\ref{sec:simulations}, therefore justifying our choice of $(K,\mathcal{S})=(2,1)$ (although any other choices would be valid, at a further computational cost). We do however note that, while such a simple, computationally lean choice is justified for the simulations in Section~\ref{sec:simulations}, this does not necessarily hold for smaller sample sizes; see for example Figure~\ref{fig:K-and-S} (right). In such settings a larger choice of either (or indeed both) $K$ and $\mathcal{S}$ can improve the mean squared error, by reducing the variance in the nuisance function estimators, in addition to the variance induced by the sample split(s) used for cross-fitting. Therefore, when analysing the real data scenarios of Section~\ref{sec:real-data} we choose a larger choice of $(K,\mathcal{S})=(5,50)$ (see Section~\ref{appsec:numerical-details} for details).

\begin{figure}[ht]
    \centering
    \includegraphics[width=0.96\textwidth]{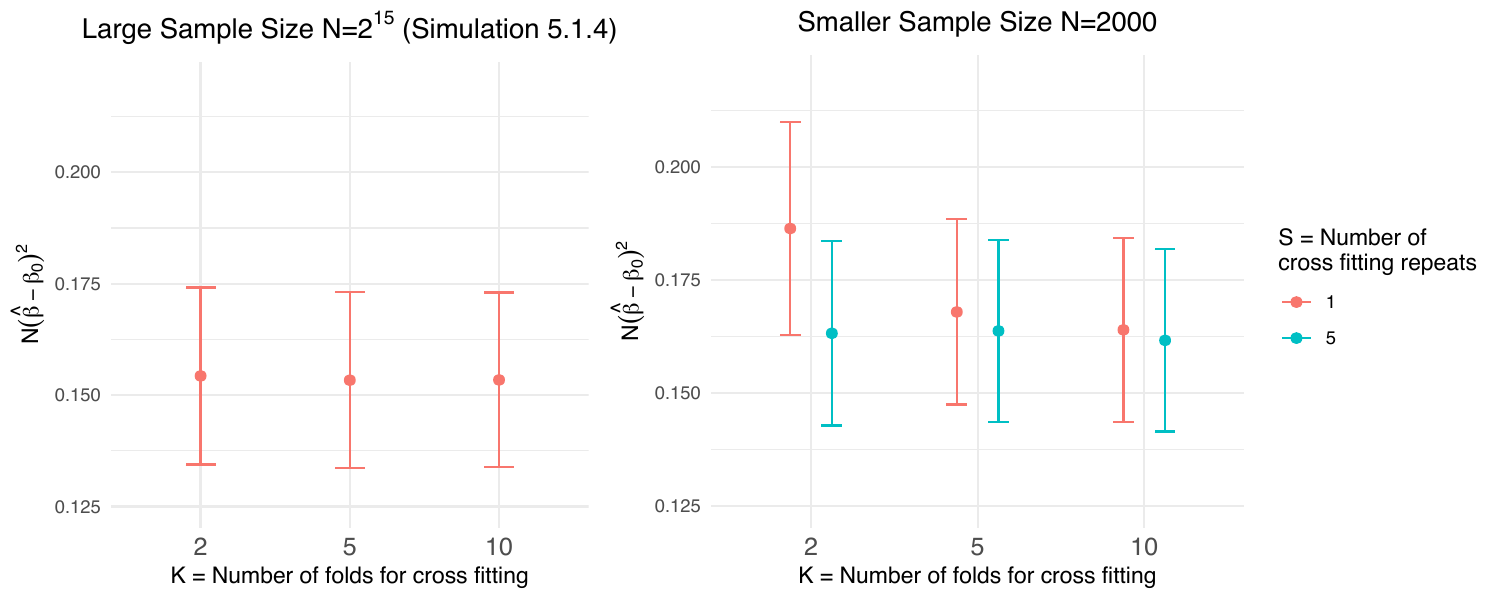}
    \captionsetup{list=false}
    \caption{Mean squared error of the sandwich boosted $\beta$-estimators for different cross-fitting parameter choices (the number of folds used for cross-fitting, $K\in\{2,5,10\}$, and the number of repeats over different sample splits, $\mathcal{S}\in\{1,5\}$). We present these for two simulations: (Left) Example~\ref{sec:sim-var}, with $I=2^{12}$ groups each of size $n_i\equiv 8$; and (Right) The same data generating mechanism as in~Example~\ref{sec:sim-var} but with $I=500$ groups each of size $n_i\equiv4$. For all settings nominal $95\%$ confidence intervals for the mean squared errors (over 500 simulations) are plotted.}
    \label{fig:K-and-S}
\end{figure}

\section{Additional numerical results}\label{sec:add_results}
Tables~\ref{tab:coverage-1}--\ref{tab:coverage-4} present the coverage probabilities of nominal 95\% confidence intervals based on the construction in Algorithm~\ref{alg:main} for the simulation settings in Section~\ref{sec:simulations}.

\begin{table}[htbp]
    \centering
    \begin{tabular}{cccccccc}
        \toprule
        $\lambda$ & \begin{tabular}{c}Sandwich\\Boosting\end{tabular} & Het ML & Hom ML & Het GEE & Hom GEE & Unweighted & Oracle \\
        \midrule
        0.1 & 0.952 & 0.954 & 0.950 & 0.950 & 0.950 & 0.946 & 0.950 \\
        0.3 & 0.956 & 0.950 & 0.956 & 0.950 & 0.956 & 0.946 & 0.950 \\
        0.5 & 0.944 & 0.952 & 0.946 & 0.952 & 0.946 & 0.952 & 0.950 \\
        0.7 & 0.950 & 0.950 & 0.948 & 0.948 & 0.948 & 0.956 & 0.948 \\
        0.9 & 0.950 & 0.958 & 0.948 & 0.950 & 0.948 & 0.954 & 0.948 \\
        1.1 & 0.956 & 0.952 & 0.950 & 0.958 & 0.950 & 0.956 & 0.958 \\
        1.3 & 0.952 & 0.946 & 0.952 & 0.962 & 0.952 & 0.950 & 0.954 \\
        1.5 & 0.958 & 0.942 & 0.948 & 0.958 & 0.948 & 0.946 & 0.962 \\
        1.7 & 0.962 & 0.946 & 0.946 & 0.970 & 0.946 & 0.948 & 0.962 \\
        1.9 & 0.958 & 0.942 & 0.944 & 0.970 & 0.944 & 0.944 & 0.966 \\
        2.1 & 0.962 & 0.938 & 0.942 & 0.954 & 0.942 & 0.948 & 0.954 \\
        2.3 & 0.954 & 0.936 & 0.946 & 0.964 & 0.946 & 0.950 & 0.956 \\
        2.5 & 0.946 & 0.942 & 0.952 & 0.948 & 0.952 & 0.946 & 0.952 \\
        \bottomrule
    \end{tabular}
    \caption{Coverage of nominal $95\%$ confidence intervals $\hat{C}(0.05)$ for $\beta$ in Example~\ref{sec:sim-complexity}.}\label{tab:coverage-1}
\end{table}

\begin{table}[htbp]
    \centering
    \begin{tabular}{ccccccccc}
        \toprule
         $\eta$ & \begin{tabular}{c}Sandwich\\Boosting\end{tabular} & Het ML & Hom ML & Het GEE & Hom GEE & Unweighted \\
        \midrule
        10 & 0.982 & 0.974 & 0.974 & 0.974 & 0.974 & 0.958 \\
        20 & 0.976 & 0.978 & 0.984 & 0.952 & 0.984 & 0.952 \\
        30 & 0.982 & 0.962 & 0.970 & 0.948 & 0.972 & 0.946 \\
        40 & 0.984 & 0.960 & 0.976 & 0.930 & 0.974 & 0.944 \\
        50 & 0.972 & 0.946 & 0.968 & 0.910 & 0.966 & 0.940 \\
        60 & 0.972 & 0.942 & 0.958 & 0.910 & 0.960 & 0.932 \\
        70 & 0.972 & 0.932 & 0.958 & 0.892 & 0.958 & 0.934 \\
        80 & 0.970 & 0.934 & 0.962 & 0.906 & 0.960 & 0.952 \\
        90 & 0.970 & 0.908 & 0.964 & 0.884 & 0.964 & 0.956 \\
        100 & 0.968 & 0.914 & 0.960 & 0.868 & 0.958 & 0.932 \\
        \bottomrule
    \end{tabular}
    \caption{Coverage of nominal $95\%$ confidence intervals $\hat{C}(0.05)$ for $\beta$ in Example~\ref{sec:sim-misspec}.}\label{tab:coverage-2}
\end{table}

\begin{table}[htbp]
    \centering
    \begin{tabular}{cccccc}
        \toprule
        \begin{tabular}{c}Group\\Size\end{tabular} & Hom SL & Hom ML & Hom GEE & Unweighted \\
        \midrule
        \(2^1\) & 0.948 & 0.948 & 0.948 & 0.946 \\
        \(2^2\) & 0.946 & 0.946 & 0.948 & 0.944 \\
        \(2^3\) & 0.950 & 0.946 & 0.936 & 0.942 \\
        \(2^4\) & 0.942 & 0.960 & 0.934 & 0.958 \\
        \(2^5\) & 0.952 & 0.956 & 0.936 & 0.948 \\
        \(2^6\) & 0.950 & 0.946 & 0.936 & 0.950 \\
        \(2^7\) & 0.944 & 0.942 & 0.932 & 0.946 \\
        \(2^8\) & 0.942 & 0.944 & 0.938 & 0.938 \\
        \bottomrule
    \end{tabular}
    \caption{Coverage of nominal $95\%$ confidence intervals $\hat{C}(0.05)$ for $\beta$ in Example~\ref{sec:sim-corr}.}\label{tab:coverage-3}
\end{table}

\begin{table}[htbp]
    \centering
    \begin{tabular}{ccccccccc}
        \toprule
        \begin{tabular}{c}Group\\Size\end{tabular} & \begin{tabular}{c}Sandwich\\Boosting\end{tabular} & Het ML & Hom ML & Het GEE & Hom GEE & Unweighted \\
        \midrule
        \(2^1\) & 0.938 & 0.936 & 0.932 & 0.948 & 0.932 & 0.926 \\
        \(2^2\) & 0.938 & 0.938 & 0.940 & 0.954 & 0.940 & 0.922 \\
        \(2^3\) & 0.938 & 0.936 & 0.938 & 0.956 & 0.938 & 0.926 \\
        \(2^4\) & 0.938 & 0.940 & 0.944 & 0.954 & 0.944 & 0.924 \\
        \(2^5\) & 0.938 & 0.940 & 0.936 & 0.966 & 0.936 & 0.934 \\
        \(2^6\) & 0.942 & 0.942 & 0.948 & 0.968 & 0.948 & 0.928 \\
        \(2^7\) & 0.946 & 0.948 & 0.950 & 0.966 & 0.950 & 0.916 \\
        \(2^8\) & 0.952 & 0.952 & 0.952 & 0.968 & 0.952 & 0.926 \\
        \bottomrule
    \end{tabular}
    \caption{Coverage of nominal $95\%$ confidence intervals $\hat{C}(0.05)$ for $\beta$ in Example~\ref{sec:sim-var}.}\label{tab:coverage-4}
\end{table}

\section{Further details on the generalised sandwich loss of Section~\ref{sec:discussion}}\label{appsec:het_treatment_effect}

Consider the setup of Section~\ref{sec:discussion} where we wish to estimate function $\beta(X)$ in model \eqref{eq:hetPLR}, a task for which by our assumption \eqref{eq:beta_basis}, it suffices to estimate coefficient vector $\mbb \phi$. Below we briefly outline how this may be performed. For intuition, suppose the regression functions $(l_0,m_0)$ ({defined as in the rest of the paper}) are known, and define $R_i^Y:=Y_i-l_0(X_i)$ and $R_i^D:=D_i-m_0(X_i)$. Then
\begin{equation*}
    R_i^Y = M_i \mbb{\phi} + \varepsilon_i,
    \qquad M_i := \big(\varphi_l(X_{ij})R_{ij}^D\big)_{j\in[n], l\in[L]}.
\end{equation*}
An estimate of $\mbb{\phi}$ can then be obtained via a potentially weighted multivariate linear regression of the $R^Y_i$ onto the $M_i$. In reality, estimates $(\hat{l},\hat{m})$ would be used in place of the unknown $(l_0,m_0)$ in a cross-fitting scheme similar to that in Section~\ref{sec:cross}. We may obtain an estimate $\hat{V}_{\mbb\phi}(W)$ as a function of the weight function, and then seek to minimise the generalised sandwich loss objective \eqref{eq:gen_sand} via sandwich boosting.

Algorithm~\ref{alg:het-scores-arbitrary} outlines how to calculate the $s$-scores for sandwich boosting in this setup for an arbitrary inverse working correlation $\mathcal{C}_{ijk}(\theta)$ as in~\eqref{eq:C-cal}. The order of computation of this score calculation can be reduced to $O(L^2N)$ when using one of the working correlation parameterisations in Section~\ref{appsec:workingcorrelations}; see Algorithm~\ref{alg:het-scores-equicorr} for the example of the equicorrelated working correlation.

\newpage

\begin{algorithm}[H]
 \KwIn{Index set $\mathcal{I}$ for boosting training; Working correlation structure $\mathcal{C}_{i}$~\eqref{eq:C-cal}; Estimates of (grouped) errors $(\hat{\xi}_i,\hat{\varepsilon}_i,X_i)_{i\in\mathcal{I}}$; Estimates of $(s,\theta)$ at which scores to be calculated; $\Phi\in\R^{L\times L}$ matrix (see~\eqref{eq:beta_function_MSE}).}

Calculate $A^{(1)}\in\R^{|\mathcal{I}|\times L\times L}$ and $A^{(2)}\in\R^{|\mathcal{I}|\times L\times L}$:
\\
\For{$i\in\mathcal{I}$} {
\For{$(l,l')\in[L]^2$} {
    $A^{(1)}_{ill'}
    := \sum_{k=1}^{n_i}\sum_{k'=1}^{n_i}\mathcal{C}_{ikk'}(\theta)s(X_{ik})s(X_{ik'})\hat{\xi}_{ik}\hat{\xi}_{ik'}\varphi_l(X_{ik})\varphi_{l'}(X_{ik'})$,
    
    $A^{(2)}_{ill'} := \sum_{k=1}^{n_i}\sum_{k'=1}^{n_i}\mathcal{C}_{ikk'}(\theta)s(X_{ik})s(X_{ik'})\hat{\xi}_{ik}\hat{\varepsilon}_{ik'}\varphi_l(X_{ik})\varphi_{l'}(X_{ik'})$.
    
}
}

Calculate $A^{(1)}_{\text{sum}} := \sum_{i\in\mathcal{I}}A^{(1)}_i$ and $A^{(2)}_{\text{sum.sq}} := \sum_{i\in\mathcal{I}}A^{(2)}_iA^{(2)T}_i$.

Calculate $A^{(3)}$ and $A^{(4)}$:

\For{$i\in\mathcal{I}$ and $j\in[n_i]$} {
\For{$(l,l')\in[L]^2$} {

    $A^{(3)}_{ijll'} := \sum_{k=1}^{n_i}\mathcal{C}_{ijk}(\theta)s(X_{ik})\hat{\xi}_{ij}\hat{\xi}_{ik}\Big(\varphi_l(X_{ij})\varphi_{l'}(X_{ik})+\varphi_{l'}(X_{ij})\varphi_l(X_{ik})\Big)$,
    
    $A^{(4)}_{ijll'} := \sum_{k=1}^{n_i}\mathcal{C}_{ijk}(\theta)s(X_{ik})\left(\hat{\xi}_{ij}\hat{\varepsilon}_{ik}\varphi_l(X_{ij})\varphi_{l'}(X_{ik}) + \hat{\xi}_{ik}\hat{\varepsilon}_{ij}\varphi_{l'}(X_{ij})\varphi_l(X_{ik})\right)$.
    
}
}
Calculate $s$-scores:

\For{$i\in\mathcal{I}$ and $j\in[n_i]$} {
    $U_{\text{SL}}^{(s)}(s,\theta)(X_{ij}) := 
    - \tr\bigg(\Phi\cdot\left(A_{\text{sum}}^{(1)}\right)^{-1}\Big[
    A_{ij}^{(3)}\left(A_{\text{sum}}^{(1)}\right)^{-1}A_{\text{sum.sq}}^{(2)} $ 
    
    $\qquad+ A_{\text{sum.sq}}^{(2)}\left(A_{\text{sum}}^{(1)}\right)^{-1}A_{ij}^{(3)}
    - A_{ij}^{(4)}A_{i}^{(2)} - A_{i}^{(2)}A_{ij}^{(4)}
    \Big]\left(A_{\text{sum}}^{(1)}\right)^{-1}\bigg)$.

}
Calculate $A^{(5)}_{\text{sum}}$ and $A^{(6)}$:

\For{$(l,l')\in[L]^2$} {
    $\big(A^{(5)}_{\text{sum}}\big)_{ll'} = \sum_{i\in\mathcal{I}}\sum_{k=1}^{n_i}\sum_{k'=1}^{n_i}\frac{\partial\mathcal{C}_{ikk'}}{\partial\theta}s(X_{ik})s(X_{ik'})\hat{\xi}_{ik}\hat{\xi}_{ik'}\varphi_l(X_{ik})\varphi_{l'}(X_{ik'})$,

\For{$i\in\mathcal{I}$} {

    $A^{(6)}_{ill'} := \sum_{k=1}^{n_i}\sum_{k'=1}^{n_i}\frac{\partial\mathcal{C}_{ikk'}}{\partial\theta}s(X_{ik})s(X_{ik'})\hat{\xi}_{ik}\hat{\varepsilon}_{ik'}\varphi_l(X_{ik})\varphi_{l'}(X_{ik'})$.

}
}
Calculate   $A^{(6)}_{\text{sum.sq.diff}} := \sum_{i\in\mathcal{I}} A_i^{(2)}A_i^{(6)} + A_i^{(6)}A_i^{(2)}$.

Calculate the $\theta$-score:

    $U_{\text{SL}}^{(\theta)}(s,\theta) =
    - \tr\bigg( \Phi\cdot
    \left(A_{\text{sum}}^{(1)}\right)^{-1}\Big[
    A_{\text{sum}}^{(5)}\left(A_{\text{sum}}^{(1)}\right)^{-1}A_{\text{sum.sq}}^{(2)}
    + A_{\text{sum.sq}}^{(2)}\left(A_{\text{sum}}^{(1)}\right)^{-1}A_{\text{sum}}^{(5)}$
    \\
    $\qquad - A_{\text{sum.sq.diff}}^{(6)}
    \Big]
    \left(A_{\text{sum}}^{(1)}\right)^{-1} \bigg)   $.

\KwOut{The set of $N$ $s$-scores $\left(X_{ij},\, U_{\SL}^{(s)}(s,\theta)(X_{ij})\right)_{i\in\mathcal{I},j\in[n_i]}$ and $\theta$-score $U_{\SL}^{(\theta)}(s,\theta)$.}
 \caption{Generalised sandwich boosting: Score calculation for arbitrary working correlation}
 \label{alg:het-scores-arbitrary}
\end{algorithm}

\begin{algorithm}[H]
 \KwIn{Index set $\mathcal{I}$ for boosting training; Estimates of (grouped) errors $(\hat{\xi}_i,\hat{\varepsilon}_i,X_i)_{i\in\mathcal{I}}$; Estimates of $(s,\theta)$ at which scores to be calculated; $\Phi\in\R^{L\times L}$ matrix (see~\eqref{eq:beta_function_MSE}).}

Calculate $A^{(1)}\in\R^{|\mathcal{I}|\times L\times L}$, $A^{(2)}\in\R^{|\mathcal{I}|\times L\times L}$, $B^{(1)}\in\R^{|\mathcal{I}|\times L}$ and $B^{(2)}\in\R^{|\mathcal{I}|\times L}$:

\For{$i\in\mathcal{I}$} {
\For{$l\in[L]$} {
$B_{il}^{(1)} := \sum_{k=1}^{n_i}s(X_{ik})\hat{\xi}_{ik}\varphi_l(X_{ik})$,

$B_{il}^{(2)} := \sum_{k=1}^{n_i}s(X_{ik})\hat{\varepsilon}_{ik}\varphi_l(X_{ik})$,
    
}

\For{$(l,l')\in[L]^2$} {
    $A^{(1)}_{ill'} := \sum_{k=1}^{n_i}s^2(X_{ik})\hat{\xi}_{ik}^2\varphi_l(X_{ik})\varphi_{l'}(X_{ik}) - \frac{\theta}{1+\theta n_i}B_{il}^{(1)} B^{(1)}_{il'}$,

    $A^{(2)}_{ill'} := \sum_{k=1}^{n_i}s^2(X_{ik})\hat{\xi}_{ik}\hat{\varepsilon}_{ik}\varphi_l(X_{ik})\varphi_{l'}(X_{ik}) - \frac{\theta}{1+\theta n_i}B_{il}^{(1)} B_{il'}^{(2)}$.
}
}

Calculate $A^{(1)}_{\text{sum}} := \sum_{i\in\mathcal{I}}A^{(1)}_i$ and $A^{(2)}_{\text{sum.sq}} := \sum_{i\in\mathcal{I}}A^{(2)}_iA^{(2)T}_i$.

Calculate $A^{(3)}$ and $A^{(4)}$:

\For{$i\in\mathcal{I}$ and $j\in[n_i]$} {
\For{$(l,l')\in[L]^2$} {

    $A^{(3)}_{ijll'} := 2s(X_{ij})\hat{\xi}_{ij}^2\varphi_l(X_{ij})\varphi_{l'}(X_{ij}) - \frac{\theta}{1+\theta n_i}\hat{\xi}_{ij}\left(B_{il}^{(1)}\varphi_{l'}(X_{ij}) + B^{(1)}_{il'}\varphi_l(X_{ij})\right)$,
    
    $A^{(4)}_{ijll'} := 2s(X_{ij})\hat{\xi}_{ij}\hat{\varepsilon}_{ij}\varphi_l(X_{ij})\varphi_{l'}(X_{ij}) - \frac{\theta}{1+\theta n_i}\left(B_{il}^{(1)}\hat{\varepsilon}_{ij}\varphi_{l'}(X_{ij}) + B_{il'}^{(2)}\hat{\xi}_{ij}\varphi_l(X_{ij})\right)$.

}
}

Calculate $s$-scores:
\\
\For{$i\in\mathcal{I}$ and $j\in[n_i]$} {

    $U_{\text{SL}}^{(s)}(s,\theta)(X_{ij}) := 
    - \tr\bigg(\Phi\cdot\left(A_{\text{sum}}^{(1)}\right)^{-1}\Big[
    A_{ij}^{(3)}\left(A_{\text{sum}}^{(1)}\right)^{-1}A_{\text{sum.sq}}^{(2)} $
    \\
    $ \qquad  + A_{\text{sum.sq}}^{(2)}\left(A_{\text{sum}}^{(1)}\right)^{-1}A_{ij}^{(3)} - A_{ij}^{(4)}A_{i}^{(2)} - A_{i}^{(2)}A_{ij}^{(4)}
    \Big]\left(A_{\text{sum}}^{(1)}\right)^{-1}\bigg).$

        }
Calculate the $\theta$-score:
\\
    $U_{\text{SL}}^{(\theta)}(s,\theta) = \tr\Bigg( \Phi\cdot \left(A_{\text{sum}}^{(1)}\right)^{-1}
    \Bigg[
    \left(\underset{i\in\mathcal{I}}{\sum}\frac{B_i^{(1)}B_i^{(1)T}}{\left(1+\theta n_i\right)^2}\right)\left(A_{\text{sum}}^{(1)}\right)^{-1}A_{\text{sum.sq}}^{(2)}$
    \\
    $ \qquad + A_{\text{sum.sq}}^{(2)}\left(A_{\text{sum}}^{(1)}\right)^{-1}\left(\underset{i\in\mathcal{I}}{\sum}\frac{B_i^{(1)}B_i^{(1)T}}{\left(1+\theta n_i\right)^2}\right)
    -\underset{i\in\mathcal{I}}{\sum}\frac{A_i^{(2)}B_i^{(1)}B_i^{(2)T} + B_i^{(2)}B_i^{(1)T}A_i^{(2)}}{\left(1+\theta n_i\right)^2}
    \Bigg]
    \left(A_{\text{sum}}^{(1)}\right)^{-1}\Bigg)$.

\KwOut{The set of $N$ $s$-scores $\left(X_{ij},\, U_{\SL}^{(s)}(s,\theta)(X_{ij})\right)_{i\in\mathcal{I},j\in[n_i]}$ and $\theta$-score $U_{\SL}^{(\theta)}(s,\theta)$.}
\caption{Generalised sandwich boosting: Score calculation for equicorrelated working correlation}
\label{alg:het-scores-equicorr}
\end{algorithm}

\end{cbunit}


\begin{thebibliography}{55}
	\providecommand{\natexlab}[1]{#1}
	\providecommand{\url}[1]{\texttt{#1}}
	\expandafter\ifx\csname urlstyle\endcsname\relax
	\providecommand{\doi}[1]{doi: #1}\else
	\providecommand{\doi}{doi: \begingroup \urlstyle{rm}\Url}\fi
	
	\bibitem[Bates et~al.(2015)Bates, M{\"a}chler, Bolker, and Walker]{lmer}
	D.~Bates, M.~M{\"a}chler, B.~Bolker, and S.~Walker.
	\newblock Fitting linear mixed-effects models using {lme4}.
	\newblock \emph{Journal of Statistical Software}, 67\penalty0 (1):\penalty0
	1--48, 2015.
	\newblock \doi{10.18637/jss.v067.i01}.
	
	\bibitem[Bojer and Meldgaard(2021)]{bojer2021kaggle}
	C.~S. Bojer and J.~P. Meldgaard.
	\newblock Kaggle forecasting competitions: An overlooked learning opportunity.
	\newblock \emph{International Journal of Forecasting}, 37\penalty0
	(2):\penalty0 587--603, 2021.
	\newblock ISSN 0169-2070.
	\newblock \doi{https://doi.org/10.1016/j.ijforecast.2020.07.007}.
	\newblock URL
	\url{https://www.sciencedirect.com/science/article/pii/S0169207020301114}.
	
	\bibitem[Box(1976)]{box}
	G.~E.~P. Box.
	\newblock Science and statistics.
	\newblock \emph{Journal of the American Statistical Association}, 71\penalty0
	(356):\penalty0 791--799, 1976.
	\newblock \doi{10.1080/01621459.1976.10480949}.
	
	\bibitem[Breiman(1999)]{breiman2}
	L.~Breiman.
	\newblock {Prediction Games and Arcing Algorithms}.
	\newblock \emph{Neural Computation}, 11\penalty0 (7):\penalty0 1493--1517, 10
	1999.
	\newblock \doi{10.1162/089976699300016106}.
	
	\bibitem[B{\"u}hlmann and Hothorn(2007)]{buhlmann1}
	P.~B{\"u}hlmann and T.~Hothorn.
	\newblock {Boosting Algorithms: Regularization, Prediction and Model Fitting}.
	\newblock \emph{Statistical Science}, 22\penalty0 (4):\penalty0 477 -- 505,
	2007.
	\newblock \doi{10.1214/07-STS242}.
	
	\bibitem[B{\"u}hlmann and Yu(2003)]{buhlmann2}
	P.~B{\"u}hlmann and B.~Yu.
	\newblock Boosting with the l2 loss.
	\newblock \emph{Journal of the American Statistical Association}, 98\penalty0
	(462):\penalty0 324--339, 2003.
	\newblock \doi{10.1198/016214503000125}.
	
	\bibitem[Bureau~of Labor~Statistics(2004)]{workingwomen}
	U.~D. o.~L. Bureau~of Labor~Statistics.
	\newblock National longitudinal survey of young women, 1968-1988 (rounds 1-15).
	\newblock Produced and distributed by the Center for Human Resource Research
	(CHRR), The Ohio State University, 2004.
	\newblock URL \url{https://www.stata-press.com/data/r10/nlswork.dta}.
	
	\bibitem[Carroll(1982)]{carroll}
	R.~J. Carroll.
	\newblock {Adapting for Heteroscedasticity in Linear Models}.
	\newblock \emph{The Annals of Statistics}, 10\penalty0 (4):\penalty0 1224 --
	1233, 1982.
	\newblock \doi{10.1214/aos/1176345987}.
	
	\bibitem[Chen and Guestrin(2016)]{xgboost}
	T.~Chen and C.~Guestrin.
	\newblock Xgboost: A scalable tree boosting system.
	\newblock In \emph{Proceedings of the 22nd ACM SIGKDD International Conference
		on Knowledge Discovery and Data Mining}, pages 785--794. Association for
	Computing Machinery, 2016.
	\newblock \doi{10.1145/2939672.2939785}.
	
	\bibitem[Chernozhukov et~al.(2018)Chernozhukov, Chetverikov, Demirer, Duflo,
	Hansen, Newey, and Robins]{chern}
	V.~Chernozhukov, D.~Chetverikov, M.~Demirer, E.~Duflo, C.~Hansen, W.~Newey, and
	J.~Robins.
	\newblock {Double/debiased machine learning for treatment and structural
		parameters}.
	\newblock \emph{The Econometrics Journal}, 21\penalty0 (1):\penalty0 C1--C68,
	2018.
	\newblock \doi{10.1111/ectj.12097}.
	
	\bibitem[Corbeil and Searle(1976)]{corbeil}
	R.~R. Corbeil and S.~R. Searle.
	\newblock Restricted maximum likelihood (reml) estimation of variance
	components in the mixed model.
	\newblock \emph{Technometrics}, 18\penalty0 (1):\penalty0 31--38, 1976.
	
	\bibitem[Crowder(1995)]{crowder}
	M.~Crowder.
	\newblock On the use of a working correlation matrix in using generalised
	linear models for repeated measures.
	\newblock \emph{Biometrika}, 82\penalty0 (2):\penalty0 407--410, 1995.
	
	\bibitem[Diggle et~al.(2013)Diggle, Heagerty, Liang, and Zeger]{diggle}
	P.~Diggle, P.~Heagerty, K.-Y. Liang, and S.~Zeger.
	\newblock \emph{Analysis of Longitudinal Data}, volume Second edition of
	\emph{Oxford Statistical Science Series}.
	\newblock OUP Oxford, Oxford, 2013.
	
	\bibitem[Emmenegger(2021)]{dmlalg}
	C.~Emmenegger.
	\newblock \emph{{dmlalg}: Double machine learning algorithms}, 2021.
	\newblock URL \url{https://CRAN.R-project.org/package=dmlalg}.
	\newblock R-package available on CRAN.
	
	\bibitem[Emmenegger and B{\"u}hlmann(2021)]{emmenegger2021regularizing}
	C.~Emmenegger and P.~B{\"u}hlmann.
	\newblock Regularizing double machine learning in partially linear endogenous
	models.
	\newblock \emph{Electronic Journal of Statistics}, 15\penalty0 (2):\penalty0
	6461--6543, 2021.
	
	\bibitem[Emmenegger and B{\"u}hlmann(2023)]{emmenegger}
	C.~Emmenegger and P.~B{\"u}hlmann.
	\newblock Plug-in machine learning for partially linear mixed-effects models
	with repeated measurements.
	\newblock \emph{Scandinavian Journal of Statistics}, 2023.
	\newblock \doi{10.1111/sjos.12639}.
	
	\bibitem[Fahrmeir and Tutz(2001)]{fahrmeir}
	L.~Fahrmeir and G.~Tutz.
	\newblock \emph{Multivariate Statistical Modelling Based on Generalized Linear
		Models}.
	\newblock Springer Series in Statistics. Springer, 2 edition, 2001.
	
	\bibitem[Freund and Schapire(1996)]{freund1996experiments}
	Y.~Freund and R.~E. Schapire.
	\newblock Experiments with a new boosting algorithm.
	\newblock In \emph{International Conference on Machine Learning}, pages
	148--156, 1996.
	
	\bibitem[Friedman et~al.(2000)Friedman, Hastie, and
	Tibshirani]{statisticalboosting}
	J.~Friedman, T.~Hastie, and R.~Tibshirani.
	\newblock Additive logistic regression: a statistical view of boosting (with
	discussion and a rejoinder by the authors).
	\newblock \emph{The Annals of Statistics}, 28\penalty0 (2):\penalty0 337 --
	407, 2000.
	\newblock \doi{10.1214/aos/1016218223}.
	
	\bibitem[Goldstein(1986)]{goldstein1986multilevel}
	H.~Goldstein.
	\newblock Multilevel mixed linear model analysis using iterative generalized
	least squares.
	\newblock \emph{Biometrika}, 73\penalty0 (1):\penalty0 43--56, 1986.
	
	\bibitem[Goldstein(1989)]{goldstein1989restricted}
	H.~Goldstein.
	\newblock Restricted unbiased iterative generalized least-squares estimation.
	\newblock \emph{Biometrika}, 76\penalty0 (3):\penalty0 622--623, 1989.
	
	\bibitem[Gourieroux and Monfort(1993)]{gourieroux}
	C.~Gourieroux and A.~Monfort.
	\newblock Pseudo-likelihood methods.
	\newblock In \emph{Econometrics}, volume~11 of \emph{Handbook of Statistics},
	pages 335--362. Elsevier, 1993.
	
	\bibitem[Gourieroux et~al.(1984)Gourieroux, Monfort, and Trognon]{gourieroux2}
	C.~Gourieroux, A.~Monfort, and A.~Trognon.
	\newblock Pseudo maximum likelihood methods: Theory.
	\newblock \emph{Econometrica}, 52\penalty0 (3):\penalty0 681--700, 1984.
	
	\bibitem[Halekoh et~al.(2006)Halekoh, H{\o}jsgaard, and Yan]{geepack}
	U.~Halekoh, S.~H{\o}jsgaard, and J.~Yan.
	\newblock The r package geepack for generalized estimating equations.
	\newblock \emph{Journal of Statistical Software}, 15/2:\penalty0 1--11, 2006.
	
	\bibitem[Hardin and Hilbe(2003)]{geehardin}
	J.~W. Hardin and J.~M. Hilbe.
	\newblock \emph{Generalized estimating equations}.
	\newblock Chapman and Hall, 2003.
	
	\bibitem[Hartley and Rao(1967)]{hartley}
	H.~O. Hartley and J.~N.~K. Rao.
	\newblock Maximum-likelihood estimation for the mixed analysis of variance
	model.
	\newblock \emph{Biometrika}, 54\penalty0 (1/2):\penalty0 93--108, 1967.
	
	\bibitem[Heagerty and Zeger(2000)]{heagerty}
	P.~J. Heagerty and S.~L. Zeger.
	\newblock Marginalized multilevel models and likelihood inference.
	\newblock \emph{Statistical Science}, 15\penalty0 (1):\penalty0 1--19, 2000.
	
	\bibitem[Huang et~al.(2007)Huang, Zhang, and Zhou]{huang2007efficient}
	J.~Z. Huang, L.~Zhang, and L.~Zhou.
	\newblock Efficient estimation in marginal partially linear models for
	longitudinal/clustered data using splines.
	\newblock \emph{Scandinavian Journal of Statistics}, 34\penalty0 (3):\penalty0
	451--477, 2007.
	
	\bibitem[Huber(1967)]{huber}
	P.~J. Huber.
	\newblock The behaviour of maximum likelihood estimates under nonstandard
	conditions.
	\newblock \emph{Proceedings of the Fifth Berkeley Symposium}, pages 221--223,
	1967.
	
	\bibitem[James M. Kilts~Center(Accessed: 2022)]{ojdataset}
	U.~o. C. B. S. o.~B. James M. Kilts~Center.
	\newblock Dominick's finer foods dataset.
	\newblock Available from Chicago Booth Research Data Center:
	\url{https://www.chicagobooth.edu/research/kilts/research-data/dominicks},
	Accessed: 2022.
	
	\bibitem[Kennedy(2022)]{kennedy2022semiparametric}
	E.~H. Kennedy.
	\newblock Semiparametric doubly robust targeted double machine learning: a
	review.
	\newblock \emph{arXiv preprint arXiv:2203.06469}, 2022.
	
	\bibitem[Li et~al.(2022)Li, Cai, and Li]{li2022inference}
	S.~Li, T.~T. Cai, and H.~Li.
	\newblock Inference for high-dimensional linear mixed-effects models: A
	quasi-likelihood approach.
	\newblock \emph{Journal of the American Statistical Association}, 117\penalty0
	(540):\penalty0 1835--1846, 2022.
	
	\bibitem[Li et~al.(2018)Li, Wang, Song, Wang, Zhou, and Zhu]{li2018doubly}
	Y.~Li, S.~Wang, P.~X.-K. Song, N.~Wang, L.~Zhou, and J.~Zhu.
	\newblock Doubly regularized estimation and selection in linear mixed-effects
	models for high-dimensional longitudinal data.
	\newblock \emph{Statistics and its Interface}, 11\penalty0 (4):\penalty0 721,
	2018.
	
	\bibitem[Liang and Zeger(1986)]{liangzeger}
	K.-Y. Liang and S.~L. Zeger.
	\newblock Longitudinal data analysis using generalized linear models.
	\newblock \emph{Biometrika}, 73\penalty0 (1):\penalty0 13--22, 1986.
	
	\bibitem[Liang et~al.(1992)Liang, Zeger, and Qaqish]{gee1}
	K.-Y. Liang, S.~L. Zeger, and B.~Qaqish.
	\newblock Multivariate regression analyses for categorical data.
	\newblock \emph{Journal of the Royal Statistical Society. Series B
		(Methodological)}, 54\penalty0 (1):\penalty0 3--40, 1992.
	
	\bibitem[Lumley(1996)]{lumley}
	T.~Lumley.
	\newblock Generalized estimating equations for ordinal data: A note on working
	correlation structures.
	\newblock \emph{Biometrics}, 52\penalty0 (1):\penalty0 354--361, 1996.
	
	\bibitem[Mason et~al.(1999)Mason, Baxter, Bartlett, and Frean]{mason}
	L.~Mason, J.~Baxter, P.~Bartlett, and M.~Frean.
	\newblock Boosting algorithms as gradient descent.
	\newblock In \emph{Advances in Neural Information Processing Systems},
	volume~12. MIT Press, 1999.
	
	\bibitem[McCullagh and Nelder(1989)]{mccullagh}
	P.~McCullagh and J.~A. Nelder.
	\newblock \emph{Generalized linear models}.
	\newblock Monographs on statistics and applied probability (Series) ; 37.
	Chapman and Hall, London, 2 edition, 1989.
	
	\bibitem[Park and Kang(2021)]{park2021more}
	C.~Park and H.~Kang.
	\newblock A more efficient, doubly robust, nonparametric estimator of treatment
	effects in multilevel studies.
	\newblock \emph{arXiv preprint arXiv:2110.07740}, 2021.
	
	\bibitem[Pinheiro et~al.(2022)Pinheiro, Bates, and {R Core Team}]{nlme-code}
	J.~Pinheiro, D.~Bates, and {R Core Team}.
	\newblock \emph{nlme: Linear and Nonlinear Mixed Effects Models}, 2022.
	\newblock URL \url{https://CRAN.R-project.org/package=nlme}.
	\newblock R package version 3.1-161.
	
	\bibitem[Pinheiro and Bates(2000)]{membook}
	J.~C. Pinheiro and D.~M. Bates.
	\newblock \emph{Mixed-Effects Models in S and S-PLUS}, volume~1 of
	\emph{Springer Statistics and Computing}.
	\newblock Springer, New York, 2000.
	
	\bibitem[Prentice and Zhao(1991)]{prentice}
	R.~L. Prentice and L.~P. Zhao.
	\newblock Estimating equations for parameters in means and covariances of
	multivariate discrete and continuous responses.
	\newblock \emph{Biometrics}, 47\penalty0 (3):\penalty0 825--839, 1991.
	
	\bibitem[Robins and Rotnitzky(1995)]{RR95}
	J.~M. Robins and A.~Rotnitzky.
	\newblock Semiparametric efficiency in multivariate regression models with
	missing data.
	\newblock \emph{Journal of the American Statistical Association}, 90\penalty0
	(429):\penalty0 122--129, 1995.
	
	\bibitem[Robins et~al.(1994)Robins, Rotnitzky, and Zhao]{RRZ94}
	J.~M. Robins, A.~Rotnitzky, and L.~P. Zhao.
	\newblock Estimation of regression coefficients when some regressors are not
	always observed.
	\newblock \emph{Journal of the American statistical Association}, 89\penalty0
	(427):\penalty0 846--866, 1994.
	
	\bibitem[Robinson(1987)]{robinson}
	P.~Robinson.
	\newblock Asymptotically efficient estimation in the presence of
	heteroskedasticity of unknown form.
	\newblock \emph{Econometrica}, 55\penalty0 (4):\penalty0 875--891, 1987.
	
	\bibitem[Royall(1986)]{royall}
	R.~M. Royall.
	\newblock Model robust confidence intervals using maximum likelihood
	estimators.
	\newblock \emph{International Statistical Review / Revue Internationale de
		Statistique}, 54\penalty0 (2):\penalty0 221--226, 1986.
	
	\bibitem[Schapire(1990)]{schapire1}
	R.~E. Schapire.
	\newblock The strength of weak learnability.
	\newblock \emph{Machine Learning}, 5\penalty0 (2):\penalty0 197--227, 1990.
	\newblock \doi{10.1007/BF00116037}.
	
	\bibitem[Tsiatis(2006)]{tsiatis}
	A.~A. Tsiatis.
	\newblock \emph{Semiparametric theory and missing data}, volume~1 of
	\emph{Springer series in statistics}.
	\newblock Springer, New York, 2006.
	
	\bibitem[van~der Vaart(1998)]{vandervaart}
	A.~W. van~der Vaart.
	\newblock \emph{Asymptotic Statistics}.
	\newblock Cambridge University Press, 1998.
	
	\bibitem[Vansteelandt and Dukes(2022)]{vansteelandt2022assumption}
	S.~Vansteelandt and O.~Dukes.
	\newblock Assumption-lean inference for generalised linear model parameters.
	\newblock \emph{Journal of the Royal Statistical Society Series B: Statistical
		Methodology}, 84\penalty0 (3):\penalty0 657--685, 2022.
	
	\bibitem[Wood(2017)]{gam-code}
	S.~Wood.
	\newblock \emph{Generalized Additive Models: An Introduction with R}.
	\newblock Chapman and Hall/CRC, 2 edition, 2017.
	
	\bibitem[You et~al.(2007)You, Chen, and Zhou]{you}
	J.~You, G.~Chen, and Y.~Zhou.
	\newblock Statistical inference of partially linear regression models with
	heteroscedastic errors.
	\newblock \emph{Journal of Multivariate Analysis}, 98\penalty0 (8):\penalty0
	1539--1557, 2007.
	\newblock \doi{10.1016/j.jmva.2007.06.011}.
	
	\bibitem[Zeger and Diggle(1994)]{zeger1994semiparametric}
	S.~L. Zeger and P.~J. Diggle.
	\newblock Semiparametric models for longitudinal data with application to cd4
	cell numbers in hiv seroconverters.
	\newblock \emph{Biometrics}, pages 689--699, 1994.
	
	\bibitem[Zeger and Liang(1986)]{zegerliang}
	S.~L. Zeger and K.-Y. Liang.
	\newblock Longitudinal data analysis for discrete and continuous outcomes.
	\newblock \emph{Biometrics}, 42\penalty0 (1):\penalty0 121--130, 1986.
	
	\bibitem[Ziegler(2011)]{ziegler}
	A.~Ziegler.
	\newblock \emph{Generalized estimating equations}.
	\newblock Lecture notes in statistics (Springer-Verlag). Springer, New York,
	2011.
	
\end{thebibliography}

\begin{thebibliography}{9}
	\providecommand{\natexlab}[1]{#1}
	\providecommand{\url}[1]{\texttt{#1}}
	\expandafter\ifx\csname urlstyle\endcsname\relax
	\providecommand{\doi}[1]{doi: #1}\else
	\providecommand{\doi}{doi: \begingroup \urlstyle{rm}\Url}\fi
	
	\bibitem[Chernozhukov et~al.(2018)Chernozhukov, Chetverikov, Demirer, Duflo,
	Hansen, Newey, and Robins]{chern}
	V.~Chernozhukov, D.~Chetverikov, M.~Demirer, E.~Duflo, C.~Hansen, W.~Newey, and
	J.~Robins.
	\newblock {Double/debiased machine learning for treatment and structural
		parameters}.
	\newblock \emph{The Econometrics Journal}, 21\penalty0 (1):\penalty0 C1--C68,
	2018.
	\newblock \doi{10.1111/ectj.12097}.
	
	\bibitem[Emmenegger and B{\"u}hlmann(2023)]{emmenegger}
	C.~Emmenegger and P.~B{\"u}hlmann.
	\newblock Plug-in machine learning for partially linear mixed-effects models
	with repeated measurements.
	\newblock \emph{Scandinavian Journal of Statistics}, 2023.
	\newblock \doi{10.1111/sjos.12639}.
	
	\bibitem[Halekoh et~al.(2006)Halekoh, H{\o}jsgaard, and Yan]{geepack}
	U.~Halekoh, S.~H{\o}jsgaard, and J.~Yan.
	\newblock The r package geepack for generalized estimating equations.
	\newblock \emph{Journal of Statistical Software}, 15/2:\penalty0 1--11, 2006.
	
	\bibitem[Jennrich(1969)]{jennrich}
	R.~I. Jennrich.
	\newblock {Asymptotic Properties of Non-Linear Least Squares Estimators}.
	\newblock \emph{The Annals of Mathematical Statistics}, 40\penalty0
	(2):\penalty0 633 -- 643, 1969.
	\newblock \doi{10.1214/aoms/1177697731}.
	
	\bibitem[Lundborg et~al.(2022)Lundborg, Kim, Shah, and
	Samworth]{lundborg2022projected}
	A.~R. Lundborg, I.~Kim, R.~D. Shah, and R.~J. Samworth.
	\newblock The projected covariance measure for assumption-lean variable
	significance testing.
	\newblock \emph{arXiv preprint arXiv:2211.02039}, 2022.
	
	\bibitem[Pinheiro et~al.(2022)Pinheiro, Bates, and {R Core Team}]{nlme-code}
	J.~Pinheiro, D.~Bates, and {R Core Team}.
	\newblock \emph{nlme: Linear and Nonlinear Mixed Effects Models}, 2022.
	\newblock URL \url{https://CRAN.R-project.org/package=nlme}.
	\newblock R package version 3.1-161.
	
	\bibitem[Shah and Peters(2020)]{shahpeters}
	R.~D. Shah and J.~Peters.
	\newblock The hardness of conditional independence testing and the generalised
	covariance measure.
	\newblock \emph{The Annals of Statistics}, 48\penalty0 (3), June 2020.
	\newblock \doi{10.1214/19-aos1857}.
	
	\bibitem[van~der Vaart(1998)]{vandervaart}
	A.~W. van~der Vaart.
	\newblock \emph{Asymptotic Statistics}.
	\newblock Cambridge University Press, 1998.
	
	\bibitem[Wood(2017)]{gam-code}
	S.~Wood.
	\newblock \emph{Generalized Additive Models: An Introduction with R}.
	\newblock Chapman and Hall/CRC, 2 edition, 2017.
	
\end{thebibliography}
\end{document}